\documentclass[12pt,draftcls, peerreview, onecolumn]{IEEEtran}
\usepackage[usenames]{color}
\usepackage{url}
\usepackage{cite}      % Written by Donald Arseneau
\usepackage{graphicx}  % Written by David Carlisle and Sebastian Rahtz
\usepackage{subfigure} % Written by Steven Douglas Cochran
\usepackage{url}       % Written by Donald Arseneau
\usepackage{amsmath}   % From the American Mathematical Society
\usepackage{lscape}
\newtheorem{theorem}{Theorem}[section]

\renewcommand{\baselinestretch}{2}

\makeatletter \newcommand{\vast}{\bBigg@{4}} \newcommand{\Vast}{\bBigg@{5}} \makeatother 

\begin{document}

% paper title
\title{Technical Report:
{An MGF-based Unified Framework to Determine the Joint Statistics of Partial Sums of Ordered i.n.d. Random Variables}}

\author{Sung~Sik~Nam,~\IEEEmembership{Member,~IEEE,}
	Hong-Chuan~Yang,~\IEEEmembership{Senior~Member,~IEEE,}
        Mohamed-Slim~Alouini,~\IEEEmembership{Fellow~Member,~IEEE,}
        and Dong~In~Kim,~\IEEEmembership{Senior~Member,~IEEE}
\thanks{$^{\star}$
This is an extended version of a paper which was presented in IEEE International Symposium on Information Theory (ISIT'2013), Istanbul, Turkey, July 2013.
S.~S.~Nam is with Hanyang University, Korea.  M.-S.~Alouini is with the Electrical Engineering Program, Computer, Electrical, Mathematical Sciences and Engineering (CEMSE) Division, KAUST, Thuwal, Makkah Province, Saudi Arabia. H.-C.~Yang is with Department of Electrical and Computer Engineering, University of Victoria, Canada.  D. I.~Kim are with the School of Information and Communication Engineering, Sungkyunkwan University, Korea. They can be reached by E-mail at
$<$ssnam@hanyang.ac.kr, dikim@skku.ac.kr, slim.alouini@kaust.edu.sa, hy@uvic.ca$>$.}
}

\markboth{S.S. Nam \MakeLowercase{\textit{et al.}}: An MGF-based Unified Framework ... on Ordered Statistics for i.n.d. RVs}{Shell \MakeLowercase{\textit{et al.}}: An MGF-based Unified Framework ... on Ordered Statistics for i.n.d. RVs}

\maketitle

\begin{abstract}
The joint statistics of partial sums of ordered random variables (RVs) are often needed for the accurate performance characterization of a wide variety of wireless communication systems. A unified analytical framework to determine the joint statistics of partial sums of ordered independent and identically distributed (i.i.d.) random variables was recently presented. However, the identical distribution assumption may not be valid in several real-world applications. With this motivation in mind, we consider in this paper the more general case in which the random variables are independent but not necessarily identically distributed (i.n.d.). More specifically, we extend the previous analysis and introduce a new more general unified analytical framework to determine the joint statistics of partial sums of ordered i.n.d. RVs. Our mathematical formalism is illustrated with an application on the exact performance analysis of the capture probability of generalized selection combining (GSC)-based RAKE receivers operating over frequency-selective fading channels with a non-uniform power delay profile. We also discussed a couple of other sample applications of the generic results presented in this work.
\end{abstract}


\section{Introduction}
%For wireless communications, theoretical performance and complexity analysis become invaluable in the process of the timely adoption of new technologies in real-world systems, because they can help circumvent the time-consuming computer simulation and expensive field test campaigns. These analytical results in the form of closed-form solutions often bring important insight into the dependence of the performance as well as complexity measures on system design parameters and, as such, facilitate the determination of the most suitable design choice in the face of practical implementation constraints. For instance, it is well known that due to the fading effect associated with multipath wireless channels, the receiver output signal to noise ratio (SNR) (or signal to interference plus noise ratio (SINR)) randomly varies over time. Thus, average performance measures are typically employed and as a result, the statistical distribution of SNR (or SINR) is needed. In general, statistical distributions for the conventional narrowband single-antenna faded link are well known. With the introduction of more advanced diversity techniques, the analysis often mandates some order statistics results. 

The subject of order statistics deals with the properties and distributions of the ordered random variables (RVs) and their functions. It has found applications in many areas of statistical theory and practice~\cite{kn:Order_Statistics}, with examples in life-testing, quality control, radar, as well as signal and image processing~\cite{kn:Handbook_of_Statistics,kn:Goldsmith_book,kn:Multi_PDF_Order_Selection,kn:Steven_Kay,kn:JAHwang,kn:RCHardie,kn:CKotropoulos}.
Order statistics has made over the last decade an increasing number of appearances in the design and analysis of wireless communication systems, specifically for the performance analysis of advanced diversity techniques, adaptive transmission techniques, and multiuser scheduling techniques (see for example~\cite{kn:simon_book, kn:MS_GSC, kn:OT_MRC_j, kn:Choi_finger, kn:Hong_Sumrate, kn:alouini_wi_j3, kn:ma00, kn:win2001, kn:annamalai2002, kn:Mallik_05, kn:bouida2008, kn:Alouini_Order_Statistics, kn:Chambers, kn:Smith}). 
In these performance analysis exercises, the joint statistics of partial sums of ordered RVs are often necessary for the accurate characterization of system performance~\cite{kn:Choi_finger, kn:bouida2008, kn:capture_outage_GSC}. %However, these statistics are not immediately available. The major difficulty in obtaining 
%these statistics resides in the fact
Note that even if the original unordered RVs are independently distributed, their ordered versions are dependent due to the inequality relations among them, which makes it challenging to such joint statistics. Recently, a successive conditioning approach was used to convert dependent ordered random variables into independent unordered ones \cite{kn:MS_GSC, kn:OT_MRC_j}. However, this approach requires some case-specific manipulations, which may not always be generalizable.

Motivated by these facts, we introduced in \cite{kn:unified_approach} a unified analytical framework to determine the joint statistics of partial sums of ordered independent and identically distributed (i.i.d.) RVs by extending the interesting results published in~\cite{kn:Integral_Solution_Nuttall,kn:joint_PDF_Nuttall,kn:Multi_PDF_Order_Selection}. More specifically, our approach can be applied not only to the cases when all the $N$ ordered RVs are involved but also to the cases when only the $N_s$ $\left(N_s < N\right)$ best RVs are considered.  
With the proposed approach, we can systematically derive the joint statistics of any partial sums of ordered statistics, in terms of the moment generating function (MGF) and the probability density function (PDF).
These statistical results can be used for the performance analysis of various wireless communication systems over generalized fading channels \cite{kn:simon_book}. However, the identical fading assumption on all diversity branches is not always valid in real-life applications. %For example, in \cite{kn:capture_outage_GSC}, we presented the exact performance analysis of the capture probability on GSC-based RAKE receivers under the assumption that the fading is both i.i.d. from path to path. However, in practice, t
The average fading power may vary from one path to the other because the branches of a diversity system are sometimes unbalanced and the communication system is sometimes operating over frequency-selective channels with a non-uniform power delay profile or channel multipath intensity profile (i.e. the average SNR of the diversity paths are not necessary the same). 

We therefore introduce in this paper an unified analytical framework to determine the joint statistics of partial sums of ordered independent non-identically distributed (i.n.d.) RVs by extending our previous work for i.i.d. fading scenarios \cite{kn:unified_approach}.
More specifically, we use an MGF based systematic analytical approach to investigate the joint statistics of any partial sums of ordered statistics for general i.n.d. fading, in terms of MGF and the PDF. We would like to emphasize that such generalization.%, when all the $N$ ordered RVs are involved and also when only the $N_s$ $\left(N_s < N\right)$ best RVs are considered. 
The main challenge for generalizing the work in \cite{kn:unified_approach} to i.n.d. general fading cases is that joint PDF of ordered i.n.d. RVs is much more complicated than that of ordered i.i.d. RVs. We need to carry out more detailed manipulation and introduce new mathematical representation to obtain the generic results (e.g. joint MGF and related joint PDF) for i.n.d. general cases in a compact form. 
In addition, we present the closed-form expressions for the exponential RV special case, which is most widely used in wireless literature. For other type of RVs, our approach will lead to much simpler results than the conventional approach involving multiple-fold integration. %\footnote{Note that in this paper, we only present the closed-form expressions for the exponential RV special case but we may also derive closed-form expressions for other RVs if there exists closed-form expressions of CDF or MGF for interesting RVs. Even if we cannot obtain the closed-form expression for other cases, our approach is still offering much more simplified results for other RV cases in comparison with the original multiple-fold integral expressions.} 
Furthermore, the exponential distribution is frequently used in the performance evaluation analysis of networks and telecommunication systems. It is also used to model the waiting times between occurrences of rare events, lifetimes of electrical or mechanical devices \cite{kn:Handbook_Stojmenovic, kn:Goldsmith_book,kn:Handbook_of_Statistics,kn:Introduction_Prob_and_Appli}. Finally, as an application of our analytical framework, we generalize the performance results of GSC-based RAKE receivers in \cite{kn:capture_outage_GSC} by maintaining the assumption of independence among the diversity paths but relaxing the identically distributed assumption. We also discussed a couple of other sample applications of the generic results presented in this work.

\section{Problem Statement and Main Idea}
Order statistics deals with the distributions and statistical properties of the new random variables obtained after ordering the realizations of some random variables.
Let $\left\{ \gamma_{i_l} \right\}$, $i_l = 1, 2,\cdots, N$ denote $N$ i.n.d. nonnegative random variables with PDF $p_{i_l}\left(\cdot \right)$ and CDF $P_{i_l}\left( \cdot \right)$. Let $u_i$ denote the random variable corresponding to the $i$-th largest observation of the $N$ original random variables (also called $i$-th order statistics), such that $u_1 \ge u_2 \ge \cdots \ge u_N$. The $N$-dimensional joint PDF of the ordered RVs $\left\{ {u_{i} } \right\}_{i = 1}^N$ is given by~\cite{kn:Order_Statistics}
\begin{equation} \small\label{eq:m-joint_PDF_MRC}
g\left( {{u_1},{u_2}, \ldots ,{u_N}} \right) = \sum\limits_{\scriptstyle {i_1},{i_2}, \ldots ,{i_N} \hfill \atop 
  \scriptstyle {i_1} \ne {i_2} \ne  \cdots  \ne {i_N} \hfill}^{1,2, \ldots ,N} {{p_{{i_1}}}\left( {{u_1}} \right){p_{{i_2}}}\left( {{u_2}} \right) \cdots {p_{{i_N}}}\left( {{u_N}} \right)}.
\end{equation}
Similarly, the $N_s$-dimensional joint PDF of $\left\{ {u_{i} } \right\}_{i = 1}^{N_s}$ is given by~\cite{kn:Order_Statistics}
\small\begin{eqnarray} \label{eq:m-joint_PDF_GSC}
\!\!\!\!\!\!g\left( {{u_1},{u_2}, \cdots ,{u_{{N_s}}}} \right) \!\!&=& \!\!\!\!\!\!\!\!\!\!\!\sum\limits_{\scriptstyle {i_1},{i_2}, \cdots ,{i_N} \atop 
  \scriptstyle {i_1} \ne {i_2} \ne  \cdots  \ne {i_N}}^{1,2, \cdots ,N} \!\!\! {{p_{{i_1}}}\left( {{u_1}} \right){p_{{i_2}}}\left( {{u_2}} \right) \cdots {p_{{i_{{N_s}}}}}\left( {{u_{{N_s}}}} \right)\prod\limits_{j = {N_s} + 1}^N {{P_{{i_j}}}\left( {{u_{{N_s}}}} \right)} } \nonumber \\
&{\rm{or}}& \nonumber \\
\!\!\!\!\!\!&=& \!\!\!\!\!\!\!\!\!\!\!\sum\limits_{\scriptstyle {i_1},{i_2}, \cdots ,{i_N} \atop 
  \scriptstyle {i_1} \ne {i_2} \ne  \cdots  \ne {i_{{N_s}}}}^{1,2, \cdots ,N} \!\!\! {{p_{{i_1}}}\left( {{u_1}} \right){p_{{i_2}}}\left( {{u_2}} \right) \cdots {p_{{i_{{N_s}}}}}\left( {{u_{{N_s}}}} \right)\!\!\!\!\sum\limits_{\substack{
 {i_{{N_s} + 1}}, \cdots ,{i_N} \\ 
 {i_{{N_s} + 1}} \ne  \cdots  \ne {i_N} \\ 
 {i_{{N_s} + 1}} \ne {i_1},{i_2}, \cdots ,{i_{{N_s}}} \\ 
  \vdots  \\ 
 {i_N} \ne {i_1},{i_2}, \cdots ,{i_{{N_s}}} \\ 
 }}^{1,2, \cdots ,N} {\prod\limits_{\scriptstyle l = {N_s} + 1 \atop 
  \scriptstyle \left\{ {{i_{{N_s} + 1}}, \cdots ,{i_N}} \right\}}^N \!\!\!\! {{P_{{i_l}}}\left( {{u_{{N_s}}}} \right)} } }.
\end{eqnarray}\normalsize

The objective is to derive the joint PDF of partial sums involving either all $N$ or the first $N_s$ ($N_s < N$) ordered RVs for the more general case in which the diversity paths are independent but not necessarily identically distributed.
Similar to \cite{kn:unified_approach}, we adopt a general two-step approach:
\begin{itemize}
\renewcommand{\labelitemi}{$\bullet$}
\item Step I: Obtain the analytical expressions of the joint MGF of partial sums (not necessarily the partial sums of interest as will be seen later). 

\item Step II: Apply inverse Laplace transform to derive the joint PDF of partial sums (additional integration may be required to obtain the desired joint PDF).
\end{itemize}
In step I, by interchanging the order of integration, while ensuring each pair of limits is chosen to be as tight as possible, the multiple integral can be rewritten into compact equivalent representations.
After obtaining the joint MGF in a compact form, we can derive joint PDF of selected partial sum through inverse Laplace transform. For most cases of our interest, the joint MGF involves basic functions, for which the inverse Laplace transform can be calculated analytically. In the worst case, we may rely on the Bromwich contour integral. In most of the case, the result involves a single one-dimensional contour integration, which can be easily and accurately evaluated numerically with the help of integral tables~\cite{kn:Mathematical_handbook, kn:gradshteyn_6} or using standard mathematical packages such as Mathematica and Matlab.

The above general steps  can be directly applied when all $N$ ordered RVs are considered and the RVs in the partial sums are continuous. When either of these conditions do not hold, we need to apply some extra steps in the analysis in order to obtain a valid joint MGF \cite{kn:unified_approach}.  For example, when the RVs involved in one partial sum is not continuous, i.e., separated by the other RVs, we need to divide these RVs into smaller sums. For example in~Fig.~\ref{Example_2}, we consider 3-dimensional joint PDF of $\{\gamma_{1:K}$, $\gamma_{2:K}$, $\gamma_{5:K}$, $\gamma_{6:K}\}$, $\{\gamma_{3:K},\gamma_{4:K}\}$, and $\{\gamma_{7:K},\gamma_{8:K}\}$ for $K>8$. Note that the first group is not continuous.
As a result, we will derive 5-dimensional joint MGF in step I, $\{\gamma_{1:K},\gamma_{2:K}\}$, $\{\gamma_{3:K},\gamma_{4:K}\}$, $\{\gamma_{5:K},\gamma_{6:K}\}$, $\{\gamma_{7:K}\}$, $\{\gamma_{8:K}\}$. After the joint PDF of the new substituted partial sums are derived with inverse Laplace transform in step II, we can transform it to a lower dimensional desired joint PDF with finite integration.

\section{Common Functions and Useful Relations}
In the following sections, we present several examples to illustrate the proposed analytical framework. Our focus is on how to obtain compact expressions of the joint MGFs for i.n.d. general fading conditions, which can be greatly simplified with the application of the following function and relations.

\subsection{Common Functions}

\begin{enumerate}
\item[i)]A mixture of a CDF and an MGF ${c_{{i_l}}}\left( {{\gamma},\lambda } \right)$:
\begin{equation}\small\label{eq:CDF_MGF}
{c_{{i_l}}}\left( {{\gamma},\lambda } \right) = \int_0^{{\gamma}} {{p_{{i_1}}}\left( x \right)\exp \left( {\lambda x} \right)dx},
\end{equation}
where $p_{{i_1}}\left( x \right)$ denotes the PDF of the RV of interest.
Note that $c_{{i_l}}\left( {{\gamma},0 } \right)=c_{{i_l}}\left( {{\gamma}} \right)$ is the CDF and $c_{{i_l}}\left( {\infty,\lambda } \right)$ leads to the MGF. 
Here, the variable $\gamma$ is real, while $\lambda$ can be complex.

\item[ii)]A mixture of an exceedance distribution function (EDF) and an MGF, ${e_{{i_l}}}\left( {{\gamma},\lambda } \right)$:
\begin{equation} \small\label{eq:EDF_MGF} 
{e_{{i_l}}}\left( {{\gamma},\lambda } \right) = \int_{{\gamma}}^\infty  {{p_{{i_1}}}\left( x \right)\exp \left( {\lambda x} \right)dx}.
\end{equation}

Note that ${e_{{i_l}}}\left( {{\gamma},0 } \right)={e_{{i_l}}}\left( {\gamma}\right)$ is the EDF while ${e_{{i_l}}}\left( {0,\lambda } \right)$ gives the MGF.

\item[iii)]An interval MGF ${\mu _{{i_l}}}\left( {{\gamma_a},{\gamma_b},\lambda } \right)$:
\begin{equation} \small\label{eq:IntervalMGF}
 {\mu _{{i_l}}}\left( {{z_a},{z_b},\lambda } \right) = \int_{{z_a}}^{{z_b}} {{p_{{i_1}}}\left( x \right)\exp \left( {\lambda x} \right)dx}. 
\end{equation}

Note that ${\mu _{{i_l}}}\left( {0, \infty, \lambda } \right)$ gives the MGF.
\end{enumerate}
Note that the functions defined in (\ref{eq:CDF_MGF}), (\ref{eq:EDF_MGF}) and (\ref{eq:IntervalMGF}) are related as follows:
\small\begin{eqnarray}
 {\mu _{{i_l}}}\left( {{z_a},{z_b},\lambda } \right)&=& {c_{{i_l}}}\left( {{z_b},\lambda } \right) - {c_{{i_l}}}\left( {{z_a},\lambda } \right) \\ 
  &=& {e_{{i_l}}}\left( {{z_b},\lambda } \right) - {e_{{i_l}}}\left( {{z_a},\lambda } \right). \label{eq:Interrelation_of_Interval_MGF}
\end{eqnarray}\normalsize

\subsection{Simplifying Relationship}
\begin{enumerate}
\item[i)] Integral $J_m$:
\\
Based on the derivation given in Appendix~\ref{AP:B}, the integral $J_m$ defined as:
\small \begin{eqnarray} \label{eq:common_fnt_1}
J_m &=& \sum\limits_{\substack{ 
 {i_m},{i_{m + 1}}, \ldots ,{i_N} \\ 
 {i_m} \ne {i_{m + 1}} \ne  \cdots  \ne {i_N} \\ 
 {i_m} \ne {i_1},{i_2}, \ldots ,{i_{m - 1}} \\ 
 {i_{m + 1}} \ne {i_1},{i_2}, \ldots ,{i_{m - 1}} \\ 
 \vdots  \\ 
 {i_N} \ne {i_1},{i_2}, \ldots ,{i_{m - 1}} 
 }}^{1,2, \ldots ,N} {\int\limits_0^{{u_{m - 1}}} {d{u_m}{p_{{i_m}}}\left( {{u_m}} \right)\exp \left( {\lambda {u_m}} \right)\int\limits_0^{{u_m}} {d{u_{m + 1}}{p_{{i_{m + 1}}}}\left( {{u_{m + 1}}} \right)\exp \left( {\lambda {u_{m + 1}}} \right) } } }  \nonumber \\
&& \cdots \int\limits_0^{u_{N - 1}} d{u_N} p_{i_N}\left( {u_N} \right)\exp \left( \lambda {u_N} \right),
\end{eqnarray}\normalsize
can be simply expressed in terms of the function $c_{i_l}\left( \gamma,\lambda  \right)$ as
\begin{equation} \small \label{eq:CDF_MGF_multiple}
{J_m} = \sum\limits_{\left\{ {{i_m},{i_{m + 1}}, \ldots ,{i_N}} \right\} \in {{\mathop{\rm P}\nolimits} _{N - m + 1}}\left( {{I_N} - \left\{ {{i_1},{i_2}, \ldots ,{i_{m - 1}}} \right\}} \right)} {\prod\limits_{\scriptstyle l = m \atop 
  \scriptstyle \left\{ {{i_m},{i_{m + 1}}, \ldots ,{i_N}} \right\}}^N {{c_{{i_l}}}\left( {{u_{m - 1}},\lambda } \right)} }.
\end{equation}
In here, the complicated summation notation used in eq. (\ref{eq:common_fnt_1}) is simplified based on the following power set definition.
We define index set $I_N$ as $I_N=\left\{1,2,\cdots,N \right\}$. The subset of $I_N$ with $n$ $\left(n\le N \right)$ elements is denoted by $\mathcal{P}_n\left(I_N\right)$. The remaining index can be grouped in the set $I_N - \mathcal{P}_n\left(I_N\right)$. Based on these definitions, a summation in (\ref{eq:common_fnt_1}) includes all possible subsets of the index set $I_N$ $\left(I_N=\left\{i_1,i_2,\cdots,i_N \right\}\right)$ excluding the subset $\left\{ {{i_1},{i_2}, \ldots ,{i_{m - 1}}} \right\}$ with $N-\left(m-1\right)$ elements and these subsets with $N-\left(m-1\right)$ elements can be denoted by  $\mathcal{P}_{N-m+1}\left(I_N-\left\{ {{i_1},{i_2}, \ldots ,{i_{m - 1}}} \right\}\right)$.

\item[ii)]Integral $J'_m$:
\\
Following the similar derivation as given in Appendix~\ref{AP:C}, the integral $J'_m$, defined as
\small\begin{eqnarray}  \label{eq:common_fnt_2}
{{J'}_m} &=& \sum\limits_{\substack{
 {i_1},{i_2}, \ldots ,{i_m} \\ 
 {i_1} \ne {i_2} \ne  \cdots  \ne {i_m} \\ 
 {i_1} \ne {i_{m + 1}},{i_{m + 2}}, \ldots ,{i_N} \\ 
 {i_2} \ne {i_{m + 1}},{i_{m + 2}}, \ldots ,{i_N} \\ 
  \vdots  \\ 
 {i_m} \ne {i_{m + 1}},{i_{m + 2}}, \ldots ,{i_N} \\ 
 }}^{1,2, \ldots ,N} {\int\limits_{{u_{m + 1}}}^\infty  {d{u_m}{p_{{i_m}}}\left( {{u_m}} \right)\exp \left( {\lambda {u_m}} \right)\int\limits_{{u_m}}^\infty  {d{u_{m - 1}}{p_{{i_{m - 1}}}}\left( {{u_{m - 1}}} \right)\exp \left( {\lambda {u_{m - 1}}} \right)  } } } \nonumber \\
&&\cdots \int\limits_{{u_2}}^\infty  {d{u_1}{p_{{i_1}}}\left( {{u_1}} \right)\exp \left( {\lambda {u_1}} \right)},
\end{eqnarray}\normalsize
can be simply re-written in terms of the function $e_{i_l}\left( \gamma,\lambda \right)$ with the help of the definition of power set used in III-B-i) as
\begin{equation}\small\label{eq:EDF_MGF_multiple}
{{J'}_m} = \sum\limits_{\left\{ {{i_1},{i_2}, \ldots ,{i_m}} \right\} \in {{\mathop{\rm P}\nolimits} _m}\left( {{I_N} - \left\{ {{i_{m + 1}},{i_{m + 2}}, \ldots ,{i_N}} \right\}} \right)} {\prod\limits_{\scriptstyle l = 1 \atop 
  \scriptstyle \left\{ {{i_1},{i_2}, \ldots ,{i_m}} \right\}}^m {{e_{{i_l}}}\left( {{u_{m + 1}},\lambda } \right)} }.
\end{equation}

\item[iii)] Integral $J''_{a,b}$:
\\
Based on the derivation given in Appendix~\ref{AP:D}, the integral $J''_{a,b}$, defined as
\small\begin{eqnarray}  \label{eq:common_fnt_3}
{{J'}_{a,b}} &=& \sum\limits_{\substack{
 {i_{a + 1}}, \ldots ,{i_{b - 1}} \\ 
 {i_{a + 1}} \ne {i_{a + 2}} \ne  \cdots  \ne {i_{b - 1}} \\ 
 {i_{a + 1}} \ne {i_1}, \cdots ,{i_a},{i_b}, \ldots ,{i_N} \\ 
 {i_{a + 2}} \ne {i_1}, \cdots ,{i_a},{i_b}, \ldots ,{i_N} \\ 
  \vdots  \\ 
 {i_{b - 1}} \ne {i_1}, \cdots ,{i_a},{i_b}, \ldots ,{i_N} \\ 
 }}^{1,2, \ldots ,N} {\int\limits_{{u_b}}^{{u_a}} {d{u_{b - 1}}{p_{{i_{b - 1}}}}\left( {{u_{b - 1}}} \right)\exp \left( {\lambda {u_{b - 1}}} \right)\int\limits_{{u_{b - 1}}}^{{u_a}} {d{u_{b - 2}}{p_{{i_{b - 2}}}}\left( {{u_{b - 2}}} \right)\exp \left( {\lambda {u_{b - 2}}} \right)  } } } \nonumber \\
&&\cdots \int\limits_{{u_{a + 2}}}^{{u_a}} {d{u_{a + 1}}{p_{{i_{a + 1}}}}\left( {{u_{a + 1}}} \right)\exp \left( {\lambda {u_{a + 1}}} \right)},
\end{eqnarray}\normalsize
can be simply re-written in terms of the function $\mu\left(\cdot,\cdot\right)$ as
\begin{equation} \small\label{eq:IntervalMGF_multiple}
{{J''}_{a,b}} = \sum\limits_{\left\{ {{i_{a + 1}}, \ldots ,{i_{b - 1}}} \right\} \in {{\mathop{\rm P}\nolimits} _{b - a + 1}}\left( {{I_N} - \left\{ {{i_1}, \cdots ,{i_a},{i_b}, \ldots ,{i_N}} \right\}} \right)} {\prod\limits_{\scriptstyle l = a + 1 \atop 
  \scriptstyle \left\{ {{i_{a + 1}}, \ldots ,{i_{b - 1}}} \right\}}^{b - 1} {{\mu _{{i_l}}}\left( {{u_b},{u_a},\lambda } \right)} }\quad\quad\text{for }b>a.
\end{equation}
\end{enumerate}

\section{Sample Cases when All $N$ Ordered RVs are Considered}
\begin{theorem} \label{th:case1_1} (PDF of $\sum\limits_{n = 1}^N {u_n }$ among $N$ ordered RVs)\\
Let $Z_1 = \sum\limits_{n=1}^{N} {u_n}$ for convenience. We can derive the PDF of $Z=\left[Z_1\right]$ as
\small\begin{eqnarray}
{p_Z}\left( {{z_1}} \right) &=& L_{{S_1}}^{ - 1}\left\{ {{\mu _Z}\left( { - {S_1}} \right)} \right\} \nonumber \\
&=& \sum\limits_{\left\{ {{i_1},{i_2}, \ldots ,{i_N}} \right\} \in {{\mathop{\rm P}\nolimits} _N}\left( {{I_N}} \right)} {L_{{S_1}}^{ - 1}\left\{ {\prod\limits_{\scriptstyle l = 1 \atop 
  \scriptstyle \left\{ {{i_1},{i_2}, \ldots ,{i_N}} \right\}}^N {{c_{{i_l}}}\left( {\infty , - {S_1}} \right)} } \right\}},
\end{eqnarray}\normalsize
where $\mathcal{L}_{S_1 }^{ - 1}\{\cdot\}$ denotes the inverse Laplace transform with respect to $S_1$. 
\end{theorem}

\begin{proof}
The MGF of $Z=\left[Z_1\right]$ is given by the expectation
\small \begin{eqnarray} \label{eq:MGF_of_pure_MRC_integralform}
MGF_Z\left( {{\lambda _1}} \right) &=& E\left\{ {\exp \left( {{\lambda _1}{z_1}} \right)} \right\} \nonumber
\\
&=& \sum\limits_{\substack{
   {{i_1},{i_2}, \cdots ,{i_N}}  \\
   {{i_1} \ne {i_2} \ne  \cdots  \ne {i_N}}  \\
}}^{1,2, \cdots ,N} {\int\limits_0^\infty  {d{u_1}{p_{{i_1}}}\left( {{u_1}} \right)\exp \left( {{\lambda _1}{u_1}} \right)\int\limits_0^{{u_1}} {d{u_2}{p_{{i_2}}}\left( {{u_2}} \right)\exp \left( {{\lambda _1}{u_2}} \right)} } } \nonumber
\\
&&\times  \cdots  \times \int\limits_0^{{u_{N - 1}}} {d{u_N}{p_{{i_N}}}\left( {{u_N}} \right)\exp \left( {{\lambda _1}{u_N}} \right)},
\end{eqnarray} \normalsize
where ${\rm E}\left\{ \cdot \right\}$ denotes the expectation operator.
By applying (\ref{eq:CDF_MGF_multiple}), we can obtain the MGF of $Z_1 = \sum_{n=1}^{N} {u_n}$ as
\small\begin{eqnarray} \label{eq:MGF_of_pure_MRC}
MGF_Z\left( {{\lambda _1}} \right) = \sum\limits_{\left\{ {{i_1},{i_2}, \ldots ,{i_N}} \right\} \in {{\mathop{\rm P}\nolimits} _N}\left( {{I_N}} \right)} {\prod\limits_{\scriptstyle l = 1 \atop 
  \scriptstyle \left\{ {{i_1},{i_2}, \ldots ,{i_N}} \right\}}^N {{c_{{i_l}}}\left( {\infty ,{\lambda _1}} \right)} }.
\end{eqnarray} \normalsize
Therefore, we can derive the PDF of $Z_1 = \sum_{m=1}^{N} {u_{n}}$ by applying the inverse Laplace transform as
\small\begin{eqnarray} 
{p_Z}\left( {{z_1}} \right) &=& L_{{S_1}}^{ - 1}\left\{ {{\mu _Z}\left( { - {S_1}} \right)} \right\} \nonumber
\\
&=& \sum\limits_{\left\{ {{i_1},{i_2}, \ldots ,{i_N}} \right\} \in {{\mathop{\rm P}\nolimits} _N}\left( {{I_N}} \right)} {L_{{S_1}}^{ - 1}\left\{ {\prod\limits_{\scriptstyle l = 1 \atop 
  \scriptstyle \left\{ {{i_1},{i_2}, \ldots ,{i_N}} \right\}}^N {{c_{{i_l}}}\left( {\infty , - {S_1}} \right)} } \right\}}.
\end{eqnarray}\normalsize
\end{proof}

\begin{theorem} \label{th:case1_2}({Joint PDF of $\sum\limits_{n = 1}^m {u_n }$ and $\sum\limits_{n = m + 1}^N {u_n }$})
\\
Let $Z_1  = \sum\limits_{n = 1}^m {u_n }$ and $Z_2  = \sum\limits_{n = m + 1}^N {u_n }$ for convenience, then we can derive the 2-dimensional joint PDF of $Z=\left[Z_1, Z_2 \right]$  as
\small\begin{eqnarray} \label{eq:joint_PDF_3}
{p_Z}\left( {{z_1},{z_2}} \right) \!\!&=& \!\!L_{{S_1},{S_2}}^{ - 1}\left\{ {{\mu _Z}\left( { - {S_1}, - {S_2}} \right)} \right\} \nonumber \\
&=& \!\!\sum\limits_{{i_m} = 1}^N {\int\limits_0^\infty  {d{u_m}{p_{{i_m}}}\left( {{u_m}} \right)} } \sum\limits_{\left\{ {{i_1}, \ldots ,{i_{m - 1}}} \right\} \in {{\mathop{\rm P}\nolimits} _{m - 1}}\left( {{I_N} - \left\{ {{i_m}} \right\}} \right)} {L_{{S_1}}^{ - 1}\left\{ {\prod\limits_{\scriptstyle k = 1 \atop 
  \scriptstyle \left\{ {{i_1}, \ldots ,{i_{m - 1}}} \right\}}^{m - 1}\!\! {{e_{{i_k}}}\left( {{u_m}, - {S_1}} \right)\exp \left( { - {S_1}{u_m}} \right)} } \right\}} \nonumber \\
&&\!\!\times \sum\limits_{\left\{ {{i_{m + 1}}, \ldots ,{i_N}} \right\} \in {{\mathop{\rm P}\nolimits} _{N - m}}\left( {{I_N} - \left\{ {{i_m}} \right\} - \left\{ {{i_1}, \ldots ,{i_{m - 1}}} \right\}} \right)} {L_{{S_2}}^{ - 1}\left\{ {\prod\limits_{\scriptstyle l = m + 1 \atop 
  \scriptstyle \left\{ {{i_{m + 1}}, \ldots ,{i_N}} \right\}}^N \!\!{{c_{{i_l}}}\left( {{u_m}, - {S_2}} \right)} } \right\}} \nonumber
\\
&& {\rm for}\; z_1\ge \frac{m}{N-m}z_2.
\end{eqnarray}\normalsize
\end{theorem}
\begin{proof}
The second order MGF of $Z=\left[Z_1,Z_2\right]$ is given by the expectation
\small\begin{eqnarray} \label{eq:joint_MGF_7_integralform}
MGF_Z\left( {{\lambda _1},{\lambda _2}} \right) &=& \sum\limits_{\substack{
   {{i_1},{i_2}, \cdots ,{i_N}}  \\
   {{i_1} \ne {i_2} \ne  \cdots  \ne {i_N}}  \\
}}^{1,2, \cdots ,N} {\int\limits_0^\infty  {d{u_1}{p_{{i_1}}}\left( {{u_1}} \right)\exp \left( {{\lambda _1}{u_1}} \right) \cdots \int\limits_0^{{u_{m - 1}}} {d{u_m}{p_{{i_m}}}\left( {{u_m}} \right)\exp \left( {{\lambda _1}{u_m}} \right)} } } \nonumber \\
&&\times \int\limits_0^{{u_m}} {d{u_{m + 1}}{p_{{i_{m + 1}}}}\left( {{u_{m + 1}}} \right)\exp \left( {{\lambda _2}{u_{m + !}}} \right) \cdots \int\limits_0^{{u_{N - 1}}} {d{u_N}{p_{{i_N}}}\left( {{u_N}} \right)\exp \left( {{\lambda _2}{u_N}} \right)} }.
\end{eqnarray}\normalsize
We show in Appendix~\ref{AP:E} that by applying (\ref{eq:CDF_MGF_multiple}) and \cite[Eq. (2)]{kn:unified_approach} and then (\ref{eq:EDF_MGF_multiple}), we can obtain the second order MGF of $Z$ as
\small\begin{eqnarray} \label{eq:joint_MGF_2}
MGF_Z \left( {\lambda _1 ,\lambda _2 } \right) \!\!\!&=& \!\!\!\sum\limits_{{i_m} = 1}^N {\int\limits_0^\infty  {d{u_m}{p_{{i_m}}}\left( {{u_m}} \right)\exp \left( {{\lambda _1}{u_m}} \right)} } \nonumber \\
&&\!\!\times \sum\limits_{\left\{ {{i_1}, \ldots ,{i_{m - 1}}} \right\} \in {{\mathop{\rm P}\nolimits} _{m - 1}}\left( {{I_N} - \left\{ {{i_m}} \right\}} \right)} {\prod\limits_{\scriptstyle k = 1 \atop 
  \scriptstyle \left\{ {{i_1}, \ldots ,{i_{m - 1}}} \right\}}^{m - 1} {{e_{{i_k}}}\left( {{u_m},{\lambda _1}} \right)} } \nonumber \\
&&\!\!\times \sum\limits_{\left\{ {{i_{m + 1}}, \ldots ,{i_N}} \right\} \in {{\mathop{\rm P}\nolimits} _{N - m}}\left( {{I_N} - \left\{ {{i_m}} \right\} - \left\{ {{i_1}, \ldots ,{i_{m - 1}}} \right\}} \right)} {\prod\limits_{\scriptstyle l = m + 1 \atop 
  \scriptstyle \left\{ {{i_{m + 1}}, \ldots ,{i_N}} \right\}}^N {{c_{{i_l}}}\left( {{u_m},{\lambda _2}} \right)} }.
\end{eqnarray}\normalsize

Again, letting $\lambda _1  =  - S_1$ and $\lambda _2  =  - S_2$, we can obtain the desired 2-dimensional joint PDF of $Z_1  = \sum\limits_{n = 1}^m {u_n }$ and $Z_2  = \sum\limits_{n = m + 1}^N {u_n }$ by applying the inverse Laplace transform as
\small\begin{eqnarray}
{p_Z}\left( {{z_1},{z_2}} \right) \!\!\!&=& \!\!\!L_{{S_1},{S_2}}^{ - 1}\left\{ {{\mu _Z}\left( { - {S_1}, - {S_2}} \right)} \right\} \nonumber \\
&=& \!\!\!\sum\limits_{{i_m} = 1}^N {\int\limits_0^\infty  {d{u_m}{p_{{i_m}}}\left( {{u_m}} \right)} } \sum\limits_{\left\{ {{i_1}, \ldots ,{i_{m - 1}}} \right\} \in {{\mathop{\rm P}\nolimits} _{m - 1}}\left( {{I_N} - \left\{ {{i_m}} \right\}} \right)} {L_{{S_1}}^{ - 1}\left\{ {\prod\limits_{\scriptstyle k = 1 \atop 
  \scriptstyle \left\{ {{i_1}, \ldots ,{i_{m - 1}}} \right\}}^{m - 1} {{e_{{i_k}}}\left( {{u_m}, - {S_1}} \right)\exp \left( { - {S_1}{u_m}} \right)} } \right\}} \nonumber \\
&&\!\!\!\times \sum\limits_{\left\{ {{i_{m + 1}}, \ldots ,{i_N}} \right\} \in {{\mathop{\rm P}\nolimits} _{N - m}}\left( {{I_N} - \left\{ {{i_m}} \right\} - \left\{ {{i_1}, \ldots ,{i_{m - 1}}} \right\}} \right)} {L_{{S_2}}^{ - 1}\left\{ {\prod\limits_{\scriptstyle l = m + 1 \atop 
  \scriptstyle \left\{ {{i_{m + 1}}, \ldots ,{i_N}} \right\}}^N {{c_{{i_l}}}\left( {{u_m}, - {S_2}} \right)} } \right\}}.
\end{eqnarray}\normalsize
\end{proof}

\begin{theorem} (Joint PDF of $u_m$ and $\sum\limits_{\scriptstyle n = 1 \hfill \atop \scriptstyle n \ne m \hfill}^N {u_n } $)
\\
Let ${Z_1} = {u_m}$ and ${Z_2} = \sum\limits_{\scriptstyle n = 1 \atop \scriptstyle n \ne m}^N {{u_n}}$ for convenience. We can obtain the 2-dimensional joint PDF of $Z=\left[Z_1, Z_2 \right]$ as
\small\begin{eqnarray}
\!\!\!\!\!\!\!\!\!\!\!\!\!\!\!\!\!\!\!\!\!\!\!\!\!\!\!\!\!\!\!\!&&{p_Z}\left( {{z_1},{z_2}} \right) = L_{{S_1},{S_2}}^{ - 1}\left\{ {{\mu _Z}\left( { - {S_1}, - {S_2}} \right)} \right\} \nonumber \\
\!\!\!\!\!\!\!\!\!\!\!\!\!\!\!\!\!\!\!\!\!\!\!\!\!\!\!\!\!\!\!\!&&= \sum\limits_{{i_m} = 1}^N {\int\limits_0^\infty  {d{u_m}{p_{{i_m}}}\left( {{u_m}} \right)L_{{S_1}}^{ - 1}\left\{ {\exp \left( { - {S_1}{u_m}} \right)} \right\}} } \sum\limits_{\left\{ {{i_1}, \ldots ,{i_{m - 1}}} \right\} \in {{\mathop{\rm P}\nolimits} _{m - 1}}\left( {{I_N} - \left\{ {{i_m}} \right\}} \right)}  \nonumber \\
\!\!\!\!\!\!\!\!\!\!\!\!\!\!\!\!\!\!\!\!\!\!\!\!\!\!\!\!\!\!\!\!
&&\quad\times {\sum\limits_{\left\{ {{i_{m + 1}}, \ldots ,{i_N}} \right\} \in {{\mathop{\rm P}\nolimits} _{N - m}}\left( {{I_N} - \left\{ {{i_m}} \right\} - \left\{ {{i_1}, \ldots ,{i_{m - 1}}} \right\}} \right)} {L_{{S_2}}^{ - 1}\left\{ {\prod\limits_{\scriptstyle k = 1 \atop 
  \scriptstyle \left\{ {{i_1}, \ldots ,{i_{m - 1}}} \right\}}^{m - 1} {{e_{{i_k}}}\left( {{u_m}, - {S_2}} \right)} \prod\limits_{\scriptstyle l = m + 1 \atop 
  \scriptstyle \left\{ {{i_{m + 1}}, \ldots ,{i_N}} \right\}}^N {{c_{{i_l}}}\left( {{u_m}, - {S_2}} \right)} } \right\}} } .
\end{eqnarray}\normalsize

\end{theorem}

\begin{proof}
Similarly to {\it Theorem \ref{th:case1_1}} and {\it \ref{th:case1_2}}, by applying (\ref{eq:CDF_MGF_multiple}), \cite[Eq. (2)]{kn:unified_approach} and (\ref{eq:EDF_MGF_multiple}), we can obtain the second order MGF of $Z_1  = u_m $ and $Z_2  = \sum\limits_{\scriptstyle n = 1 \hfill \atop \scriptstyle n \ne m \hfill}^N {u_n }$. Detailed derivation is omitted.
\end{proof}

\section{Sample cases when only $N_s$ ordered RVs are considered}
Let us now consider the cases where only the best $N_s$($\le N$) ordered RVs are involved.

\begin{theorem} \label{th:case2_1} ({PDF of $\sum\limits_{n = 1}^{N_s } {u_n }$}, $N_s\ge 2$)
\\
Let $Z' = \sum\limits_{n=1}^{N_s} {u_n}$ for convenience, then we can derive the PDF of $Z'$ as
\small\begin{eqnarray} \label{eq:PDF_of_pure_GSC}
p_{Z'} \left( x \right) = p_{\sum_{n = 1}^{N_s } {u_n } } \left( x \right)=
\int_0^{\frac{x}{{N_s }}} {p_Z \left( {x - z_2 ,z_2 } \right)dz_2 } & \text{for } N_s \ge 2, 
\end{eqnarray}\normalsize
\end{theorem}
where
\small\begin{eqnarray} \label{eq:PDF_of_pure_GSC_m}
{p_Z}\left( {{z_1},{z_2}} \right) &=& L_{{S_1},{S_2}}^{ - 1}\left\{ {{\mu _Z}\left( { - {S_1}, - {S_2}} \right)} \right\} \nonumber \\
&=& \sum\limits_{{i_{{N_s}}} = 1}^N {\int\limits_0^\infty  {d{u_{{N_s}}}{p_{{i_{{N_s}}}}}\left( {{u_{{N_s}}}} \right)L_{{S_2}}^{ - 1}\left\{ {\exp \left( { - {S_2}{u_{{N_s}}}} \right)} \right\}\sum\limits_{\substack{
 {i_{{N_s} + 1}}, \ldots ,{i_N} \\ 
 {i_{{N_s} + 1}} \ne  \cdots  \ne {i_N} \\ 
 {i_{{N_s} + 1}} \ne {i_{{N_s}}} \\ 
  \vdots  \\ 
 {i_N} \ne {i_{{N_s}}} \\ 
 }}^{1,2, \ldots ,N} {\prod\limits_{\scriptstyle k = {N_s} + 1 \atop 
  \scriptstyle \left\{ {{i_{{N_s} + 1}}, \ldots ,{i_N}} \right\}}^N {{P_{{i_k}}}\left( {{u_{{N_s}}}} \right)} } } } \nonumber \\
&&\times \sum\limits_{\left\{ {{i_1}, \ldots ,{i_{{N_s} - 1}}} \right\} \in {{\mathop{\rm P}\nolimits} _{{N_s} - 1}}\left( {{I_N} - \left\{ {{i_{{N_s}}}} \right\} - \left\{ {{i_{{N_s} + 1}}, \ldots ,{i_N}} \right\}} \right)} {L_{{S_1}}^{ - 1}\left\{ {\prod\limits_{\scriptstyle l = 1 \atop 
  \scriptstyle \left\{ {{i_1}, \ldots ,{i_{{N_s} - 1}}} \right\}}^{{N_s} - 1} {{e_{{i_l}}}\left( {{u_{{N_s}}}, - {S_1}} \right)} } \right\}}.
\end{eqnarray}\normalsize

\begin{proof}
We only need to consider $u_{N_s}$ separately in this case. 
Let $Z_1=\sum\limits_{n=1}^{N_s-1} {u_n}$ and $Z_2=u_{N_s}$. The target second order MGF of $Z=[Z_1, Z_2]$ is given by the expectation in 
\small\begin{eqnarray} \label{eq:MGF_of_pure_GSC_integralform}
MGF_Z \left( {\lambda _{1,} \lambda _2 } \right) &=& {\rm E}\left\{ {\exp \left( {\lambda _1 z_1  + \lambda _2 z_2 } \right)} \right\} \nonumber \\
&=&\sum\limits_{\substack{
   {{i_1},{i_2}, \cdots ,{i_N}}  \\
   {{i_1} \ne {i_2} \ne  \cdots  \ne {i_N}}  \\
}}^{1,2, \cdots ,N} {\int\limits_0^\infty  {d{u_1}{p_{{i_1}}}\left( {{u_1}} \right)\exp \left( {{\lambda _1}{u_1}} \right) \cdots \int\limits_0^{{u_{{N_s} - 2}}} {d{u_{{N_s} - 1}}{p_{{i_{{N_s} - 1}}}}\left( {{u_{{N_s} - 1}}} \right)\exp \left( {{\lambda _1}{u_{{N_s} - 1}}} \right)} } } \nonumber \\
&&\times \int\limits_0^{{u_{{N_s} - 1}}} {d{u_{{N_s}}}{p_{{i_{{N_s}}}}}\left( {{u_{{N_s}}}} \right)\exp \left( {{\lambda _2}{u_{{N_s}}}} \right)\prod\limits_{j = {N_s} + 1}^N {{P_{{i_j}}}\left( {{u_{{N_s}}}} \right)} }.
\end{eqnarray}\normalsize
By simply applying \cite[Eq. (2)]{kn:unified_approach} and then (\ref{eq:EDF_MGF_multiple}) to (\ref{eq:MGF_of_pure_GSC_integralform}), we can obtain the second order MGF result as
\small\begin{eqnarray} \label{MGF_pure_GSC}
MGF_Z \left( {\lambda _{1,} \lambda _2 } \right) \!\!&=&\!\! \sum\limits_{{i_{{N_s}}} = 1}^N {\int\limits_0^\infty  {d{u_{{N_s}}}{p_{{i_{{N_s}}}}}\left( {{u_{{N_s}}}} \right)\exp \left( {{\lambda _2}{u_{{N_s}}}} \right)\sum\limits_{\substack{
 {i_{{N_s} + 1}}, \ldots ,{i_N} \\ 
 {i_{{N_s} + 1}} \ne  \cdots  \ne {i_N} \\ 
 {i_{{N_s} + 1}} \ne {i_{{N_s}}} \\ 
  \vdots  \\ 
 {i_N} \ne {i_{{N_s}}} \\ 
 }}^{1,2, \ldots ,N} {\prod\limits_{\scriptstyle k = {N_s} + 1 \atop 
  \scriptstyle \left\{ {{i_{{N_s} + 1}}, \ldots ,{i_N}} \right\}}^N {{P_{{i_k}}}\left( {{u_{{N_s}}}} \right)} } } } \nonumber \\
&&\!\!\times \sum\limits_{\left\{ {{i_1}, \ldots ,{i_{{N_s} - 1}}} \right\} \in {{\mathop{\rm P}\nolimits} _{{N_s} - 1}}\left( {{I_N} - \left\{ {{i_{{N_s}}}} \right\} - \left\{ {{i_{{N_s} + 1}}, \ldots ,{i_N}} \right\}} \right)} {\prod\limits_{\scriptstyle l = 1 \atop 
  \scriptstyle \left\{ {{i_1}, \ldots ,{i_{{N_s} - 1}}} \right\}}^{{N_s} - 1} {{e_{{i_l}}}\left( {{u_{{N_s}}},{\lambda _1}} \right)} }.
\end{eqnarray} \normalsize
Again, letting $\lambda _1  =  - S_1$ and $\lambda _2  =  - S_2$, we can obtain the 2-dimensional joint PDF of $Z_1=\sum\limits_{n=1}^{N_s-1} {u_n}$ and $Z_2=u_{N_s}$ by applying the inverse Laplace transform as
\small\begin{eqnarray}
{p_Z}\left( {{z_1},{z_2}} \right)\!\!&=& \!\!L_{{S_1},{S_2}}^{ - 1}\left\{ {{\mu _Z}\left( { - {S_1}, - {S_2}} \right)} \right\} \nonumber \\
&=&\!\! \sum\limits_{{i_{{N_s}}} = 1}^N {\int\limits_0^\infty  {d{u_{{N_s}}}{p_{{i_{{N_s}}}}}\left( {{u_{{N_s}}}} \right)L_{{S_2}}^{ - 1}\left\{ {\exp \left( { - {S_2}{u_{{N_s}}}} \right)} \right\}\sum\limits_{\substack{
 {i_{{N_s} + 1}}, \ldots ,{i_N} \\ 
 {i_{{N_s} + 1}} \ne  \cdots  \ne {i_N} \\ 
 {i_{{N_s} + 1}} \ne {i_{{N_s}}} \\ 
  \vdots  \\ 
 {i_N} \ne {i_{{N_s}}} \\ 
 }}^{1,2, \ldots ,N} {\prod\limits_{\scriptstyle k = {N_s} + 1 \atop 
  \scriptstyle \left\{ {{i_{{N_s} + 1}}, \ldots ,{i_N}} \right\}}^N {{P_{{i_k}}}\left( {{u_{{N_s}}}} \right)} } } } \nonumber \\
&&\!\! \times \sum\limits_{\left\{ {{i_1}, \ldots ,{i_{{N_s} - 1}}} \right\} \in {{\mathop{\rm P}\nolimits} _{{N_s} - 1}}\left( {{I_N} - \left\{ {{i_{{N_s}}}} \right\} - \left\{ {{i_{{N_s} + 1}}, \ldots ,{i_N}} \right\}} \right)} {L_{{S_1}}^{ - 1}\left\{ {\prod\limits_{\scriptstyle l = 1 \atop 
  \scriptstyle \left\{ {{i_1}, \ldots ,{i_{{N_s} - 1}}} \right\}}^{{N_s} - 1} {{e_{{i_l}}}\left( {{u_{{N_s}}}, - {S_1}} \right)} } \right\}}.
\end{eqnarray}\normalsize
Finally, noting that  $Z'=Z_1+Z_2$, we can obtain the target PDF of $Z'$ with the following finite integration
\begin{equation} \small
p_{Z'}(x)=\int_0^{\frac{x}{{N_s }}} {p_Z \left( {x - z_2 ,z_2 } \right)dz_2 }.
\end{equation}
\end{proof}

\begin{theorem} \label{th:case2_2} ({Joint PDF of $u_{m}$ and $\sum\limits_{\scriptstyle n = 1 \hfill \atop \scriptstyle n \ne m \hfill}^{N_s } {u_{n} }$} for $1 < m < N_s  - 1$)

Let $X=u_n$ and $Y=\sum\limits_{\scriptstyle n = 1 \hfill \atop \scriptstyle n \ne m \hfill}^{N_s } {u_n }$, then the joint PDF of $Z=\left[X, Y\right]$ can be obtained as
\small\begin{eqnarray} \label{eq:joint_PDF_1}
\!\!\!\!\!\!\!\!\!\!\!\!\!\!  p_{Z}\left(x, y\right) &=& p_{u_m ,\sum\limits_{\scriptstyle n = 1 \hfill \atop 
  \scriptstyle n \ne m \hfill}^{N_s } {u_n } } \left( {x,y} \right) \nonumber
\\
\!\!\!\!\!\!\!\!\!\!\!\!\!\! &=&\!\!\!\!\!\!\!
\int_0^x {\int_{\left( {m - 1} \right)x}^{y - \left( {N_s  - m} \right)z_4 } {p_{\sum\limits_{n = 1}^{m - 1} {u_n } , u_m ,\sum\limits_{n = m + 1}^{N_s  - 1} {u_n } , u_{N_s} } \!\!\!\!\!\!\! \left( {z_1 ,x,y \!-\!z_1\!-\!z_4 ,z_4 } \right)dz_1 } dz_4 }.
\end{eqnarray}\normalsize
\end{theorem}

\begin{proof}
For the joint PDF of $u_m$ and $\sum\limits_{\scriptstyle n = 1 \hfill \atop \scriptstyle n \ne m \hfill}^{N_s } {u_n }$, as one of original groups is split by $u_m$, we should consider substituted groups for the split group instead of original groups as shown in Fig.~\ref{Example_3_2}. As a result, we will start by obtaining a four order MGF. In this case, the higher dimensional joint PDF can then be used to find the desired 2-dimensional joint PDF of interest by transformation.

Applying the results in \cite[Eq. (2)]{kn:unified_approach}, (\ref{eq:CDF_MGF_multiple}), (\ref{eq:EDF_MGF_multiple}) and (\ref{eq:IntervalMGF_multiple}), we derive in Appendix~\ref{AP:G} the target joint MGF.
Let $Z_1  = \sum\limits_{n = 1}^{m - 1} {u_n }$, $Z_2  = u_m$, $Z_3 = \sum\limits_{n = m + 1}^{N_s  - 1} {u_n }$, and $Z_4  = u_{N_s}$, then
\small \begin{eqnarray} \label{eq:joint_MGF_GSC_2}
\!\!\!\!\!\!MGF_Z \left( {\lambda _1 ,\lambda _2 ,\lambda _3 ,\lambda _4 } \right) \!\!\!\!&=&\!\!\!\! \sum\limits_{\scriptstyle {i_{{N_s}}}, \ldots ,{i_N} \atop 
  \scriptstyle {i_{{N_s}}} \ne  \cdots  \ne {i_N}}^{1,2, \ldots ,N} {\int\limits_0^\infty  {d{u_{{N_s}}}{p_{{i_{{N_s}}}}}\left( {{u_{{N_s}}}} \right)\exp \left( {{\lambda _4}{u_{{N_s}}}} \right)\!\!\prod\limits_{\scriptstyle j = {N_s} + 1 \atop 
  \scriptstyle \left\{ {{i_{{N_s} + 1}}, \ldots ,{i_N}} \right\}}^N {{P_{{i_j}}}\left( {{u_{{N_s}}}} \right)} } } \nonumber \\
\!\!\!\!\!\!&&\!\!\!\!\times \sum\limits_{\scriptstyle {i_m} = 1 \atop 
  \scriptstyle {i_m} \ne {i_{{N_s}, \ldots ,}}{i_N}}^N {\int\limits_{{u_{{N_s}}}}^\infty  {d{u_m}{p_{{i_m}}}\left( {{u_m}} \right)\exp \left( {{\lambda _2}{u_m}} \right)} } \nonumber \\
\!\!\!\!\!\!&&\!\!\!\!\times \sum\limits_{\left\{ {{i_{m + 1}}, \ldots ,{i_{{N_s} - 1}}} \right\} \in {{\mathop{\rm P}\nolimits} _{{N_s} - m - 1}}\left( {{I_N} - \left\{ {{i_m}} \right\} - \left\{ {{i_{{N_s}}}, \ldots ,{i_N}} \right\}} \right)} {\prod\limits_{\scriptstyle k = m + 1 \atop 
  \scriptstyle \left\{ {{i_{m + 1}}, \ldots ,{i_{{N_s} - 1}}} \right\}}^{{N_s} - 1} {{\mu _{{i_k}}}\left( {{u_{{N_s}}},{u_m},{\lambda _3}} \right)} } \nonumber \\
\!\!\!\!\!\!&&\!\!\!\!\times \sum\limits_{\left\{ {{i_1}, \ldots ,{i_{m - 1}}} \right\} \in {{\mathop{\rm P}\nolimits} _{m - 1}}\left( {{I_N} - \left\{ {{i_m}} \right\} - \left\{ {{i_{{N_s}}}, \ldots ,{i_N}} \right\} - \left\{ {{i_{m + 1}}, \ldots ,{i_{{N_s} - 1}}} \right\}} \right)} {\prod\limits_{\scriptstyle l = 1 \atop 
  \scriptstyle \left\{ {{i_1}, \ldots ,{i_{m - 1}}} \right\}}^{m - 1}\!\! {{e_{{i_l}}}\left( {{u_m},{\lambda _1}} \right)} }.
\end{eqnarray} \normalsize

Starting from the MGF expressions given above, we apply inverse Laplace transforms in Appendix~\ref{AP:G} in order to derive the following joint PDFs

\small \begin{eqnarray} \label{eq:joint_PDF_GSC_2} 
{p_Z}\left( {{z_1},{z_2},{z_3},{z_4}} \right) \!\!\!&=& \!\!\!L_{{S_1},{S_2},{S_3},{S_4}}^{ - 1}\left\{ {{\mu _Z}\left( { - {S_1}, - {S_2}, - {S_3}, - {S_4}} \right)} \right\} \nonumber \\
&=&\!\!\! \sum\limits_{\scriptstyle {i_{{N_s}}}, \ldots ,{i_N} \atop 
  \scriptstyle {i_{{N_s}}} \ne  \cdots  \ne {i_N}}^{1,2, \ldots ,N} {\int\limits_0^\infty  {d{u_{{N_s}}}{p_{{i_{{N_s}}}}}\left( {{u_{{N_s}}}} \right)L_{{S_4}}^{ - 1}\left\{ {\exp \left( { - {S_4}{u_{{N_s}}}} \right)} \right\}\prod\limits_{\scriptstyle j = {N_s} + 1 \atop 
  \scriptstyle \left\{ {{i_{{N_s} + 1}}, \ldots ,{i_N}} \right\}}^N {{P_{{i_j}}}\left( {{u_{{N_s}}}} \right)} } } \nonumber \\
&&\!\!\!\times \sum\limits_{\scriptstyle {i_m} = 1 \atop 
  \scriptstyle {i_m} \ne {i_{{N_s}, \ldots ,}}{i_N}}^N {\int\limits_{{u_{{N_s}}}}^\infty  {d{u_m}{p_{{i_m}}}\left( {{u_m}} \right)L_{{S_2}}^{ - 1}\left\{ {\exp \left( { - {S_2}{u_m}} \right)} \right\}} } \nonumber \\
&&\!\!\!\times \sum\limits_{\left\{ {{i_{m + 1}}, \ldots ,{i_{{N_s} - 1}}} \right\} \in {{\mathop{\rm P}\nolimits} _{{N_s} - m - 1}}\left( {{I_N} - \left\{ {{i_m}} \right\} - \left\{ {{i_{{N_s}}}, \ldots ,{i_N}} \right\}} \right)} \!\! {L_{{S_3}}^{ - 1}\!\left\{\! {\prod\limits_{\scriptstyle k = m + 1 \atop 
  \scriptstyle \left\{\! {{i_{m + 1}}, \ldots ,{i_{{N_s} - 1}}} \!\right\}}^{{N_s} - 1} \!{{\mu _{{i_k}}}\left(\! {{u_{{N_s}}},{u_m}, - {S_3}} \!\right)} } \!\right\}} \nonumber \\
&&\!\!\!\times \!\sum\limits_{\left\{\! {{i_1}, \ldots ,{i_{m - 1}}} \!\right\} \in {{\mathop{\rm P}\nolimits} _{m - 1}}\left(\! {{I_N} - \left\{\! {{i_m}} \!\right\} - \left\{\! {{i_{{N_s}}}, \ldots ,{i_N}} \!\right\} - \left\{\! {{i_{m + 1}}, \ldots ,{i_{{N_s} - 1}}}\! \right\}} \!\right)} \!\! {L_{{S_1}}^{ - 1}\left\{\! {\prod\limits_{\scriptstyle l = 1 \atop 
  \scriptstyle \left\{\! {{i_1}, \ldots ,{i_{m - 1}}} \!\right\}}^{m - 1} {{e_{{i_l}}}\left(\! {{u_m}, - {S_1}} \!\right)} } \!\right\}}, \nonumber
\\
&&\text{for } z_4<z_2,\; z_1>(m-1)z_2\;\text{and}\; \left(N_s-m-1\right)z_4<z_3<\left(N_s-m-1\right)z_2. 
\end{eqnarray} \normalsize
\end{proof}
Note that (\ref{eq:joint_PDF_1}) involves only finite integrations of joint PDFs. Therefore, while a generic closed-form expression is not possible, the desired joint PDF can be easily numerically evaluated with the help of integral tables~\cite{kn:Mathematical_handbook, kn:gradshteyn_6} or using standard mathematical packages, such as Mathematica or Matlab etc.

\begin{theorem} \label{th:case2_3} (Joint PDF of $\sum\limits_{n = 1}^m {u_n}$ and $\sum\limits_{n = m + 1}^{N_s } {u_n}$)
\\
Let $X=\sum\limits_{n = 1}^m {u_n}$ and $Y=\sum\limits_{n = m + 1}^{N_s } {u_n}$, then we can simply obtain the joint PDF of $ Z=[X,Y]$ as
\small\begin{eqnarray}
 \!\!\!\! \!\!\!\!\!\!\!\!\!\!\!\! \!\!\!\!  && \!\!\!\!\!\!\!\! p_Z \left( {x,y} \right) = p_{\sum\limits_{n = 1}^m {u_n},\sum\limits_{n = m + 1}^{N_s } {u_n}} \left( {x,y} \right)\nonumber
\\
\!\!\!\! \!\!\!\!\!\!\!\!\!\!\!\! \!\!\!\!  &=& \!\!\!\!  \int_0^{\frac{y}{{N_s  - m}}} \!\! {\int_{\frac{y}{{N_s  - m}}}^{\frac{x}{m}} {p_{\sum\limits_{n = 1}^{m - 1} {u_n} ,u_m,\sum\limits_{n = m + 1}^{K_s  - 1} {u_n} , u_{N_s} } \!\! \left( {x - z_2 ,z_2 ,y - z_4 ,z_4 } \right)dz_2 } dz_4 }, \text{ for } x>\frac{m}{N_s-m}y.
\end{eqnarray} \normalsize
\end{theorem}
\begin{proof}Omitted.
\end{proof}
Note again that only the finite integrations of joint PDFs are involved.

\section{Closed-form expressions for exponential RV case} \label{SEC:VI}
Now, we focus on obtaining the joint PDFs for i.n.d. exponential RV special cases in a ready-to-use form. The PDF and the CDF of the RVs are given by ${p_{{i_l}}}\left( x \right) = \frac{1}{{{{\bar \gamma }_{{i_l}}}}}\exp \left( { - \frac{x}{{{{\bar \gamma }_{{i_l}}}}}} \right)$ and $P_{{i_l}}\left( x \right) = 1-\exp \left( { - \frac{x}{{{{\bar \gamma }_{{i_l}}}}}} \right)$ for $\gamma\ge 0$, respectively, where ${\bar \gamma }_{i_l}$ is the average of the $l$-th RV.

The above novel generic results are quite general and apply to any RVs. We now focus on obtaining the joint PDFs for i.n.d. exponential RV special cases in a ready-to-use form and illustrate in this section some results for the independent non-identical exponential RV special case, where the PDF and the CDF of $\gamma$ are given by ${p_{{i_l}}}\left( x \right) = \frac{1}{{{{\bar \gamma }_{{i_l}}}}}\exp \left( { - \frac{x}{{{{\bar \gamma }_{{i_l}}}}}} \right)$ and $P_{{i_l}}\left( x \right) = 1-\exp \left( { - \frac{x}{{{{\bar \gamma }_{{i_l}}}}}} \right)$ for $\gamma\ge 0$, respectively, where ${\bar \gamma }_{i_l}$ is the average of the $l$-th RV. 
As shown in Appendix~\ref{AP:H}, (\ref{eq:CDF_MGF_multiple}), (\ref{eq:EDF_MGF_multiple}) and (\ref{eq:IntervalMGF_multiple}) can be specialized to
\begin{enumerate}
\item[i)] For special case:
\begin{equation}\small \label{eq:common_function_Rayleigh_1_s}
 {c_{{i_l}}}\left( {{z_a},\lambda } \right) = \frac{1}{{1 - {{\bar \gamma }_{{i_l}}}\lambda }}\left[ {1 - \exp \left( {\left( {\lambda  - \frac{1}{{{{\bar \gamma }_{{i_l}}}}}} \right){z_a}} \right)} \right], 
\end{equation}
\begin{equation} \small \label{eq:common_function_Rayleigh_2_s}
 {e_{{i_l}}}\left( {{z_a},\lambda } \right) = \frac{1}{{1 - {{\bar \gamma }_{{i_l}}}\lambda }}\left[ {\exp \left( {\left( {\lambda  - \frac{1}{{{{\bar \gamma }_{{i_l}}}}}} \right){z_a}} \right)} \right], 
\end{equation}
\begin{equation} \small \label{eq:common_function_Rayleigh_3_s}
 {\mu _{{i_l}}}\left( {{z_a},{z_b},\lambda } \right) = \frac{1}{{1 - {{\bar \gamma }_{{i_l}}}\lambda }}\left[ {\exp \left( {\left( {\lambda  - \frac{1}{{{{\bar \gamma }_{{i_l}}}}}} \right){z_b}} \right) - \exp \left( {\left( {\lambda  - \frac{1}{{{{\bar \gamma }_{{i_l}}}}}} \right){z_a}} \right)} \right]. 
\end{equation}
\item[ii)] For general case:
\small
\begin{eqnarray} \label{eq:common_function_Rayleigh_1}
\!\!\!\!\!\!\!\!\!\!\!\!\!\!\!\!\!\!&& \prod\limits_{l = {n_1}}^{{n_2}} {{c_{{i_l}}}\left( {{z_a},\lambda } \right)}  = \frac{1}{{\prod\limits_{l = {n_1}}^{{n_2}} {\left( {1 - {{\bar \gamma }_{{i_l}}}\lambda } \right)} }}\prod\limits_{l = {n_1}}^{{n_2}} {\left[ {1 - \exp \left( {\left( {\lambda  - \frac{1}{{{{\bar \gamma }_{{i_l}}}}}} \right){z_a}} \right)} \right]}  \nonumber\\ 
\!\!\!\!\!\!\!\!\!\!\!\!\!\!\!\!\!\!&&= \!\sum\limits_{k = {n_1}}^{{n_2}} \!{{C_{k,n_1,n_2}}\left[\! {\frac{{1 + \left[\! {\sum\limits_{l = 1}^{{n_2} - {n_1} + 1} {\exp \left(\! {l \cdot {z_a} \cdot \lambda } \!\right)\left\{\! {{{\left( { - 1} \right)}^l}\sum\limits_{{j_1} = {j_0} + {n_1}}^{{n_2} - l + 1} { \cdots \sum\limits_{{j_l} = {j_{l - 1}} + 1}^{{n_2}} {\exp \left(\! { - \sum\limits_{m = 1}^l {\frac{{{z_a}}}{{{{\bar \gamma }_{{i_{{j_m}}}}}}}} } \!\right)} } } \!\right\}} } \!\right]}}{{\left( \!{\lambda  - \frac{1}{{{{\bar \gamma }_{{i_k}}}}}} \!\right)}}} \!\right]}, 
\end{eqnarray}
\begin{eqnarray} \label{eq:common_function_Rayleigh_2}
\!\!\!\!\!\!\!\!\!\!\!\!\!\!\!\!\!\!&&\prod\limits_{l = {n_1}}^{{n_2}} {{e_{{i_l}}}\left( {{z_a},\lambda } \right)}  = \frac{1}{{\prod\limits_{l = {n_1}}^{{n_2}} {\left( {1 - {{\bar \gamma }_{{i_l}}}\lambda } \right)} }}\exp \left( {\left\{ {\sum\limits_{l = {n_1}}^{{n_2}} {\left( {\lambda  - \frac{1}{{{{\bar \gamma }_{{i_l}}}}}} \right)} } \right\}{z_a}} \right) \nonumber\\ 
\!\!\!\!\!\!\!\!\!\!\!\!\!\!\!\!\!\!&&= \sum\limits_{k = {n_1}}^{{n_2}} {\frac{C_{k,n_1,n_2}}{{\left( {\lambda  - \frac{1}{{{{\bar \gamma }_{{i_k}}}}}} \right)}}\exp \left( { - \sum\limits_{l = {n_1}}^{{n_2}} {\left( {\frac{{{z_a}}}{{{{\bar \gamma }_{{i_l}}}}}} \right)} } \right)\exp \left( {\left( {{n_2} - {n_1} + 1} \right){z_a}\lambda } \right)},
\end{eqnarray}
\begin{eqnarray} \label{eq:common_function_Rayleigh_3}
\!\!\!\!\!\!\!\!\!\!\!\!\!\!\!\!\!\!&& \prod\limits_{l = {n_1}}^{{n_2}} {{\mu _{{i_l}}}\left( {{z_a},{z_b},\lambda } \right)} 
= \frac{1}{{\prod\limits_{l = {n_1}}^{{n_2}} {\left( {1 - {{\bar \gamma }_{{i_l}}}\lambda } \right)} }}\prod\limits_{l = {n_1}}^{{n_2}} {\left[ {\exp \left( {\left( {\lambda  - \frac{1}{{{{\bar \gamma }_{{i_l}}}}}} \right){z_a}} \right) - \exp \left( {\left( {\lambda  - \frac{1}{{{{\bar \gamma }_{{i_l}}}}}} \right){z_b}} \right)} \right]}  \nonumber\\ 
\!\!\!\!\!\!\!\!\!\!\!\!\!\!\!\!\!\!&&= \sum\limits_{k = {n_1}}^{{n_2}} {{C_{k,n_1,n_2}}\left[ {\frac{{\exp \left( {\left( {{n_2} - {n_1} + 1} \right) \cdot {z_a} \cdot \lambda } \right)\exp \left( { - \sum\limits_{l = {n_1}}^{{n_2}} {\left( {\frac{{{z_a}}}{{{{\bar \gamma }_{{i_l}}}}}} \right)} } \right)}}{{\left( {\lambda  - \frac{1}{{{{\bar \gamma }_{{i_k}}}}}} \right)}}} \right.} \nonumber
\nonumber\\
\!\!\!\!\!\!\!\!\!\!\!\!\!\!\!\!\!\!&&\quad\left. { \times \left\{ {1 + \sum\limits_{l = 1}^{{n_2} - {n_1} + 1} {\exp \left( {l \cdot \left( {{z_b} - {z_a}} \right) \cdot \lambda } \right)\left\{ {{{\left( { - 1} \right)}^l}\sum\limits_{{j_1} = {j_0} + {n_1}}^{{n_2} - l + 1} { \cdots \sum\limits_{{j_l} = {j_{l - 1}} + 1}^{{n_2}} {\exp \left( { - \sum\limits_{m = 1}^l {\frac{{{z_b} - {z_a}}}{{{{\bar \gamma }_{{i_{{j_m}}}}}}}} } \right)} } } \right\}} } \right\}} \right], 
\end{eqnarray} \normalsize
where 
\begin{equation} \small
{C_{l,n_1,n_2}} = \frac{1}{{\prod\limits_{l = {n_1}}^{{n_2}} {\left( { - {{\bar \gamma }_{{i_l}}}} \right)} }{F'\left( {\frac{1}{{{{\bar \gamma }_{{i_l}}}}}} \right)}},
\end{equation}
\begin{equation} \small
F'\left( x \right) \!= \left[\! {\sum\limits_{l = 1}^{{n_2} - {n_1}}\! {\left(\! {{n_2} - {n_1} - l + 1} \!\right){x^{{n_2} - {n_1} - l}}{{\left(\! { - 1} \!\right)}^l}\sum\limits_{{j_1} = {j_0} + {n_1}}^{{n_2} - l + 1} { \cdots \sum\limits_{{j_l} = {j_{l - 1}} + 1}^{{n_2}} \!{\prod\limits_{m = 1}^l {\frac{1}{{{{\bar \gamma }_{{i_{{j_m}}}}}}}} } } } } \!\right] \!+ \!\left(\! {{n_2} - {n_1} + 1} \!\right){x^{{n_2} - {n_1}}}.
\end{equation}
\end{enumerate}

After substituting (\ref{eq:common_function_Rayleigh_1}), (\ref{eq:common_function_Rayleigh_2}) and (\ref{eq:common_function_Rayleigh_3}) into the derived expressions of the joint PDF of partial sums of ordered statistics presented in the previous sections, it is easy to derive the following closed-form expressions for the PDFs by applying the classical inverse Laplace transform pair and the property given in \cite[Appendix I]{kn:unified_approach}. While some of these results have been derived using the successive conditioning approach previously, we list them here for the sake of convenience and completeness in the next page.

\begin{landscape}
\begin{enumerate}
\item[1)] PDF of $\sum\limits_{n = 1}^N {u_n}$:
\begin{equation} \footnotesize\label{eq:closed_form_1}
{p_Z}\left( {{z_1}} \right)= \sum\limits_{\left\{ {{i_1},{i_2}, \ldots ,{i_N}} \right\} \in {{\mathop{\rm P}\nolimits} _N}\left( {{I_N}} \right)} {\sum\limits_{l = 1}^N {{C_{l,1,N}}L_{{S_1}}^{ - 1}\left\{ {\frac{1}{{\left( { - {S_1} - \frac{1}{{{{\bar \gamma }_{{i_l}}}}}} \right)}}} \right\}} },
\end{equation}
where
\begin{equation}\footnotesize
L_{{S_1}}^{ - 1}\left\{ {\frac{1}{{\left( { - {S_1} - \frac{1}{{{{\bar \gamma }_{{i_l}}}}}} \right)}}} \right\} =  - \exp \left( { - \frac{{{z_1}}}{{{{\bar \gamma }_{{i_l}}}}}} \right).
\end{equation}

\item[2)] Joint PDF of $u_m$ and $\sum\limits_{\scriptstyle n = 1 \hfill \atop \scriptstyle n \ne m \hfill}^N {u_n} $:
\footnotesize \begin{eqnarray} \label{eq:non_closed_form_2}
\!\!\!\!\!\!\!\!\!\!\!\!\!\!\!\!\!\!\!\!\!\!\!\!\!\!\!\!\!\!\!\!\!\!\!\!\!\!\!\!\!\!\!\!&&{p_Z}\left( {{z_1},{z_2}} \right) = L_{{S_1},{S_2}}^{ - 1}\left\{ {{\mu _Z}\left( { - {S_1}, - {S_2}} \right)} \right\} \nonumber \\
\!\!\!\!\!\!\!\!\!\!\!\!\!\!\!\!\!\!\!\!\!\!\!\!\!\!\!\!\!\!\!\!\!\!\!\!\!\!\!\!\!\!\!\!&&= \sum\limits_{{i_m} = 1}^N {\int\limits_0^\infty  {d{u_m}\frac{1}{{{{\bar \gamma }_{{i_m}}}}}\exp \left( { - \frac{{{u_m}}}{{{{\bar \gamma }_{{i_m}}}}}} \right)L_{{S_1}}^{ - 1}\left\{ {\exp \left( { - {S_1}{u_m}} \right)} \right\}} } \nonumber \\
\!\!\!\!\!\!\!\!\!\!\!\!\!\!\!\!\!\!\!\!\!\!\!\!\!\!\!\!\!\!\!\!\!\!\!\!\!\!\!\!\!\!\!\!\!\!\!\!\!\!\!\!\!\!\!\!\!\!\!\!\!\!\!\!\!\!\!\!\!\!\!\!\!\!\!\!\!\!\!\!&&\quad \times \!\! \sum\limits_{\left\{ {{i_1}, \ldots ,{i_{m - 1}}} \right\} \in {{\mathop{\rm P}\nolimits} _{m - 1}}\left( {{I_N} - \left\{ {{i_m}} \right\}} \right)} \!\! {\sum\limits_{\scriptstyle k = 1 \atop 
  \scriptstyle \left\{ {{i_1}, \ldots ,{i_{m - 1}}} \right\}}^{m - 1} \!\! {{C_{k,1,m-1}}\exp \left( \!{ - \sum\limits_{l = 1}^{m - 1} {\left(\! {\frac{{{u_m}}}{{{{\bar \gamma }_{{i_l}}}}}} \!\right)} } \!\right)\! \sum\limits_{\left\{ {{i_{m + 1}}, \ldots ,{i_N}} \right\} \in {{\mathop{\rm P}\nolimits} _{N - m}}\left( \! {{I_N} - \left\{ \!{{i_m}} \!\right\} - \left\{\! {{i_1}, \ldots ,{i_{m - 1}}} \!\right\}} \!\right)} {\sum\limits_{\scriptstyle q = m + 1 \atop 
  \scriptstyle \left\{ {{i_{m + 1}}, \ldots ,{i_N}} \right\}}^N \!\! {{C_{q,m+1,N}}L_{{S_2}}^{ - 1}\left\{ \! {\frac{{\exp \left( \!{ - \left( \!{m - 1} \! \right){u_m}{S_2}}\! \right)}}{{\left( \!{ - {S_2} - \frac{1}{{{{\bar \gamma }_{{i_k}}}}}} \!\right)\left( \!{ - {S_2} - \frac{1}{{{{\bar \gamma }_{{i_q}}}}}} \!\right)}}} \!\right\}} } } } \nonumber \\
\!\!\!\!\!\!\!\!\!\!\!\!\!\!\!\!\!\!\!\!\!\!\!\!\!\!\!\!\!\!\!\!\!\!\!\!\!\!\!\!\!\!\!\!&&\quad + \sum\limits_{{i_m} = 1}^N {\int\limits_0^\infty  {d{u_m}\frac{1}{{{{\bar \gamma }_{{i_m}}}}}\exp \left( { - \frac{{{u_m}}}{{{{\bar \gamma }_{{i_m}}}}}} \right)L_{{S_1}}^{ - 1}\left\{ {\exp \left( { - {S_1}{u_m}} \right)} \right\}} } \sum\limits_{\left\{ {{i_1}, \ldots ,{i_{m - 1}}} \right\} \in {{\mathop{\rm P}\nolimits} _{m - 1}}\left( {{I_N} - \left\{ {{i_m}} \right\}} \right)} {\sum\limits_{\scriptstyle k = 1 \atop 
  \scriptstyle \left\{ {{i_1}, \ldots ,{i_{m - 1}}} \right\}}^{m - 1} {{C_{k,1,m-1}}\exp \left( { - \sum\limits_{l = 1}^{m - 1} {\left( {\frac{{{u_m}}}{{{{\bar \gamma }_{{i_l}}}}}} \right)} } \right)} } \nonumber \\
\!\!\!\!\!\!\!\!\!\!\!\!\!\!\!\!\!\!\!\!\!\!\!\!\!\!\!\!\!\!\!\!\!\!\!\!\!\!\!\!\!\!\!\!&&\quad \times \!\! \sum\limits_{\left\{ {{i_{m + 1}}, \ldots ,{i_N}} \right\} \in {{\mathop{\rm P}\nolimits} _{N - m}}\left( {{I_N} - \left\{ {{i_m}} \right\} - \left\{ {{i_1}, \ldots ,{i_{m - 1}}} \right\}} \right)} \!\! {\sum\limits_{\scriptstyle q = m + 1 \atop 
  \scriptstyle \left\{ {{i_{m + 1}}, \ldots ,{i_N}} \right\}}^N {{C_{q,m+1,N}}\left[ \! {\sum\limits_{h = 1}^{N - m} \! {\left\{ \! {{{\left(\! { - 1} \!\right)}^h}\! \sum\limits_{{j_1} = {j_0} + m + 1}^{N - h + 1} {\! \cdots\! \sum\limits_{{j_h} = {j_{h - 1}} + 1}^N {\exp \left( \!{ - \sum\limits_{m = 1}^h {\frac{{{u_m}}}{{{{\bar \gamma }_{{i_{{j_m}}}}}}}} } \!\right)} } } \!\right\}L_{{S_2}}^{ - 1}\left\{\! {\frac{{\exp \left( \!{ - \left( \!{h + m - 1} \!\right){u_m}{S_2}} \!\right)}}{{\left(\! { - {S_2} - \frac{1}{{{{\bar \gamma }_{{i_k}}}}}}\! \right)\left(\! { - {S_2} - \frac{1}{{{{\bar \gamma }_{{i_q}}}}}} \!\right)}}} \!\right\}} } \!\right]} },
\end{eqnarray} \normalsize
where
\begin{equation} \footnotesize
L_{{S_1}}^{ - 1}\left\{ {\exp \left( { - {S_1}{u_m}} \right)} \right\} = \delta \left( {{z_1} - {u_m}} \right),
\end{equation}
\begin{equation} \footnotesize 
\!\!\!\!\!\!\!\!\!\!\!\!L_{{S_2}}^{ - 1}\left\{ {\frac{{\exp \left( { - \left( {m - 1} \right){u_m}{S_2}} \right)}}{{\left( { - {S_2} - \frac{1}{{{{\bar \gamma }_{{i_k}}}}}} \right)\left( { - {S_2} - \frac{1}{{{{\bar \gamma }_{{i_q}}}}}} \right)}}} \right\} = \frac{{\exp \left( { - \left( {{z_2} - \left( {m - 1} \right){u_m}} \right)\left( {\frac{1}{{{{\bar \gamma }_{{i_k}}}}} + \frac{1}{{{{\bar \gamma }_{{i_q}}}}}} \right)} \right)\left\{ {\exp \left( {\frac{{{z_2} - \left( {m - 1} \right){u_m}}}{{{{\bar \gamma }_{{i_q}}}}}} \right) - \exp \left( {\frac{{{z_2} - \left( {m - 1} \right){u_m}}}{{{{\bar \gamma }_{{i_k}}}}}} \right)} \right\}U\left( {{z_2} - \left( {m - 1} \right){u_m}} \right)}}{{\left( {\frac{1}{{{{\bar \gamma }_{{i_q}}}}} - \frac{1}{{{{\bar \gamma }_{{i_k}}}}}} \right)}},
\end{equation}
\begin{equation} \footnotesize 
\!\!\!\!\!\!\!\!\!\!\!\!L_{{S_2}}^{ - 1}\left\{ {\frac{{\exp \left( { - \left( {h + m - 1} \right){u_m}{S_2}} \right)}}{{\left( { - {S_2} - \frac{1}{{{{\bar \gamma }_{{i_k}}}}}} \right)\left( { - {S_2} - \frac{1}{{{{\bar \gamma }_{{i_q}}}}}} \right)}}} \right\} = \frac{{\exp \left( { - \left( {{z_2} - \left( {h + m - 1} \right){u_m}} \right)\left( {\frac{1}{{{{\bar \gamma }_{{i_k}}}}} + \frac{1}{{{{\bar \gamma }_{{i_q}}}}}} \right)} \right)\left\{ {\exp \left( {\frac{{{z_2} - \left( {h + m - 1} \right){u_m}}}{{{{\bar \gamma }_{{i_q}}}}}} \right) - \exp \left( {\frac{{{z_2} - \left( {h + m - 1} \right){u_m}}}{{{{\bar \gamma }_{{i_k}}}}}} \right)} \right\}U\left( {{z_2} - \left( {h + m - 1} \right){u_m}} \right)}}{{\left( {\frac{1}{{{{\bar \gamma }_{{i_q}}}}} - \frac{1}{{{{\bar \gamma }_{{i_k}}}}}} \right)}}.
\end{equation}

\item[3)] Joint PDF of $\sum\limits_{n = 1}^m {u_n}$ and $\sum\limits_{n = m + 1}^N {u_n}$:
\footnotesize \begin{eqnarray} \label{eq:non_closed_form_3}
\!\!\!\!\!\!\!\!\!\!\!\!\!\!\!\!\!\!\!\!\!\!\!\!\!\!\!\!\!\!\!\!\!\!\!\!\!\!\!\!\!\!\!\!&& {p_Z}\left( {{z_1},{z_2}} \right) = L_{{S_1},{S_2}}^{ - 1}\left\{ {{\mu _Z}\left( { - {S_1}, - {S_2}} \right)} \right\} \nonumber \\
\!\!\!\!\!\!\!\!\!\!\!\!\!\!\!\!\!\!\!\!\!\!\!\!\!\!\!\!\!\!\!\!\!\!\!\!\!\!\!\!\!\!\!\!&&= \sum\limits_{{i_m} = 1}^N {\int\limits_0^\infty  {d{u_m}\frac{1}{{{{\bar \gamma }_{{i_m}}}}}\exp \left( { - \frac{{{u_m}}}{{{{\bar \gamma }_{{i_m}}}}}} \right)} } \nonumber \\
\!\!\!\!\!\!\!\!\!\!\!\!\!\!\!\!\!\!\!\!\!\!\!\!\!\!\!\!\!\!\!\!\!\!\!\!\!\!\!\!\!\!\!\!&&\quad \times \sum\limits_{\left\{ {{i_1}, \ldots ,{i_{m - 1}}} \right\} \in {{\mathop{\rm P}\nolimits} _{m - 1}}\left( {{I_N} - \left\{ {{i_m}} \right\}} \right)} {\sum\limits_{\scriptstyle k = 1 \atop 
  \scriptstyle \left\{ {{i_1}, \ldots ,{i_{m - 1}}} \right\}}^{m - 1} {{C_{k,1,m-1}}\exp \left( { - \sum\limits_{l = 1}^{m - 1} {\left( {\frac{{{u_m}}}{{{{\bar \gamma }_{{i_l}}}}}} \right)} } \right)L_{{S_1}}^{ - 1}\left\{ {\frac{{\exp \left( { - m{u_m}{S_1}} \right)}}{{\left( { - {S_1} - \frac{1}{{{{\bar \gamma }_{{i_k}}}}}} \right)}}} \right\}\sum\limits_{\left\{ {{i_{m + 1}}, \ldots ,{i_N}} \right\} \in {{\mathop{\rm P}\nolimits} _{N - m}}\left( {{I_N} - \left\{ {{i_m}} \right\} - \left\{ {{i_1}, \ldots ,{i_{m - 1}}} \right\}} \right)}  } } \nonumber \\
\!\!\!\!\!\!\!\!\!\!\!\!\!\!\!\!\!\!\!\!\!\!\!\!\!\!\!\!\!\!\!\!\!\!\!\!\!\!\!\!\!\!\!\!&&\quad \times \sum\limits_{\scriptstyle q = m + 1 \atop 
  \scriptstyle \left\{ {{i_{m + 1}}, \ldots ,{i_N}} \right\}}^N {C_{q,m+1,N}}L_{{S_2}}^{ - 1}\left\{ {\frac{1}{{\left( { - {S_2} - \frac{1}{{{{\bar \gamma }_{{i_q}}}}}} \right)}}} \right\}
\nonumber\\
\!\!\!\!\!\!\!\!\!\!\!\!\!\!\!\!\!\!\!\!\!\!\!\!\!\!\!\!\!\!\!\!\!\!\!\!\!\!\!\!\!\!\!\!&&\quad + \sum\limits_{{i_m} = 1}^N {\int\limits_0^\infty  {d{u_m}\frac{1}{{{{\bar \gamma }_{{i_m}}}}}\exp \left( { - \frac{{{u_m}}}{{{{\bar \gamma }_{{i_m}}}}}} \right)} } \sum\limits_{\left\{ {{i_1}, \ldots ,{i_{m - 1}}} \right\} \in {{\mathop{\rm P}\nolimits} _{m - 1}}\left( {{I_N} - \left\{ {{i_m}} \right\}} \right)} {\sum\limits_{\scriptstyle k = 1 \atop 
  \scriptstyle \left\{ {{i_1}, \ldots ,{i_{m - 1}}} \right\}}^{m - 1} {{C_{k,1,m-1}}\exp \left( { - \sum\limits_{l = 1}^{m - 1} {\left( {\frac{{{u_m}}}{{{{\bar \gamma }_{{i_l}}}}}} \right)} } \right)L_{{S_1}}^{ - 1}\left\{ {\frac{{\exp \left( { - m{u_m}{S_1}} \right)}}{{\left( { - {S_1} - \frac{1}{{{{\bar \gamma }_{{i_k}}}}}} \right)}}} \right\}} } \nonumber \\
\!\!\!\!\!\!\!\!\!\!\!\!\!\!\!\!\!\!\!\!\!\!\!\!\!\!\!\!\!\!\!\!\!\!\!\!\!\!\!\!\!\!\!\!&&\quad \times \sum\limits_{\left\{ {{i_{m + 1}}, \ldots ,{i_N}} \right\} \in {{\mathop{\rm P}\nolimits} _{N - m}}\left( {{I_N} - \left\{ {{i_m}} \right\} - \left\{ {{i_1}, \ldots ,{i_{m - 1}}} \right\}} \right)} {\sum\limits_{\scriptstyle q = m + 1 \atop 
  \scriptstyle \left\{ {{i_{m + 1}}, \ldots ,{i_N}} \right\}}^N {{C_{q,m+1,N}}\left[ {\sum\limits_{h = 1}^{N - m} {\left\{ {{{\left( { - 1} \right)}^h}\sum\limits_{{j_1} = {j_0} + m + 1}^{N - h + 1} { \cdots \sum\limits_{{j_h} = {j_{h - 1}} + 1}^N {\exp \left( { - \sum\limits_{m = 1}^h {\frac{{{u_m}}}{{{{\bar \gamma }_{{i_{{j_m}}}}}}}} } \right)} } } \right\}L_{{S_2}}^{ - 1}\left\{ {\frac{{\exp \left( { - h{u_m}{S_2}} \right)}}{{\left( { - {S_2} - \frac{1}{{{{\bar \gamma }_{{i_q}}}}}} \right)}}} \right\}} } \right]} },
\end{eqnarray} \normalsize
where
\begin{equation} \footnotesize
L_{{S_2}}^{ - 1}\left\{ {\frac{1}{{\left( { - {S_2} - \frac{1}{{{{\bar \gamma }_{{i_q}}}}}} \right)}}} \right\} =  - \exp \left( { - \frac{{{z_2}}}{{{{\bar \gamma }_{{i_q}}}}}} \right),
\end{equation}
\begin{equation} \footnotesize
L_{{S_1}}^{ - 1}\left\{ {\frac{{\exp \left( { - m{u_m}{S_1}} \right)}}{{ - {S_1} - \frac{1}{{{{\bar \gamma }_{{i_k}}}}}}}} \right\} =  - \exp \left( -{\frac{{{z_1} - m{u_m}}}{{{{\bar \gamma }_{{i_k}}}}}} \right)U\left( {{z_1} - m{u_m}} \right),
\end{equation}
\begin{equation} \footnotesize
L_{{S_2}}^{ - 1}\left\{ {\frac{{\exp \left( { - h{u_m}{S_2}} \right)}}{{ - {S_2} - \frac{1}{{{{\bar \gamma }_{{i_q}}}}}}}} \right\} =  - \exp \left( -{\frac{{{z_2} - h{u_m}}}{{{{\bar \gamma }_{{i_q}}}}}} \right)U\left( {{z_2} - h{u_m}} \right).
\end{equation}

\item[4)] PDF of $\sum\limits_{n = 1}^{N_s } {u_n}$:
\footnotesize\begin{eqnarray} \label{eq:non_closed_form_4}
\!\!\!\!\!\!\!\!\!\!\!\!\!\!\!\!\!\!\!\!\!\!\!\!\!\!\!\!\!\!\!\!\!\!\!\!\!\!\!\!\!\!\!\!&&{p_Z}\left( {{z_1},{z_2}} \right) = L_{{S_1},{S_2}}^{ - 1}\left\{ {{\mu _Z}\left( { - {S_1}, - {S_2}} \right)} \right\} \nonumber \\
\!\!\!\!\!\!\!\!\!\!\!\!\!\!\!\!\!\!\!\!\!\!\!\!\!\!\!\!\!\!\!\!\!\!\!\!\!\!\!\!\!\!\!\!&&= \sum\limits_{{i_{{N_s}}} = 1}^N {\int\limits_0^\infty  {d{u_{{N_s}}}\frac{1}{{{{\bar \gamma }_{{i_{{N_s}}}}}}}\exp \left( { - \frac{{{u_{{N_s}}}}}{{{{\bar \gamma }_{{i_{{N_s}}}}}}}} \right)L_{{S_2}}^{ - 1}\left\{ {\exp \left( { - {S_2}{u_{{N_s}}}} \right)} \right\}\sum\limits_{\substack{
 {i_{{N_s} + 1}}, \ldots ,{i_N} \\ 
 {i_{{N_s} + 1}} \ne  \cdots  \ne {i_N} \\ 
 {i_{{N_s} + 1}} \ne {i_{{N_s}}} \\ 
  \vdots  \\ 
 {i_N} \ne {i_{{N_s}}} \\ 
 }}^{1,2, \ldots ,N} {\prod\limits_{\scriptstyle k = {N_s} + 1 \atop 
  \scriptstyle \left\{ {{i_{{N_s} + 1}}, \ldots ,{i_N}} \right\}}^N {\left\{ {1 - \exp \left( { - \frac{{{u_{{N_s}}}}}{{{{\bar \gamma }_{{i_k}}}}}} \right)} \right\}} } } } \nonumber \\
\!\!\!\!\!\!\!\!\!\!\!\!\!\!\!\!\!\!\!\!\!\!\!\!\!\!\!\!\!\!\!\!\!\!\!\!\!\!\!\!\!\!\!\!&&\quad \times \sum\limits_{\left\{ {{i_1}, \ldots ,{i_{{N_s} - 1}}} \right\} \in {{\mathop{\rm P}\nolimits} _{{N_s} - 1}}\left( {{I_N} - \left\{ {{i_{{N_s}}}} \right\} - \left\{ {{i_{{N_s} + 1}}, \ldots ,{i_N}} \right\}} \right)} {\prod\limits_{\scriptstyle q = 1 \atop 
  \scriptstyle \left\{ {{i_1}, \ldots ,{i_{{N_s} - 1}}} \right\}}^{{N_s} - 1} {{C_{q,1,{N_s} - 1}}\exp \left( { - \sum\limits_{l = 1}^{{N_s} - 1} {\left( {\frac{{{u_{{N_s}}}}}{{{{\bar \gamma }_{{i_l}}}}}} \right)} } \right)L_{{S_1}}^{ - 1}\left\{ {\frac{{\exp \left( { - \left( {{N_s} - 1} \right){u_{{N_s}}}{S_1}} \right)}}{{\left( { - {S_1} - \frac{1}{{{{\bar \gamma }_{{i_q}}}}}} \right)}}} \right\}} },
\end{eqnarray} \normalsize
where
\begin{equation} \footnotesize
L_{{S_2}}^{ - 1}\left\{ {\exp \left( { - {S_2}{u_{{N_s}}}} \right)} \right\} = \delta \left( {{z_2} - {u_{{N_s}}}} \right),
\end{equation}
\begin{equation} \footnotesize 
L_{{S_1}}^{ - 1}\left\{ {\frac{{\exp \left( { - \left( {{N_s} - 1} \right){u_{{N_s}}}{S_1}} \right)}}{{\left( { - {S_1} - \frac{1}{{{{\bar \gamma }_{{i_q}}}}}} \right)}}} \right\} =  - \exp \left( { - \frac{{{z_1} - \left( {{N_s} - 1} \right){u_{{N_s}}}}}{{{{\bar \gamma }_{{i_q}}}}}} \right)U\left( {{z_1} - \left( {{N_s} - 1} \right){u_{{N_s}}}} \right).
\end{equation}

\item[5)] Joint PDF of $u_m$ and $\sum\limits_{\scriptstyle n = 1 \hfill \atop \scriptstyle n \ne m \hfill}^{N_s } {u_{n} }$ for $1<m<N_s-1$:
\footnotesize\begin{eqnarray} \label{eq:non_closed_form_5_2}
\!\!\!\!\!\!\!\!\!\!\!\!\!\!\!\!\!\!\!\!\!\!\!\!\!\!\!\!\!\!\!\!\!\!\!\!\!\!\!\!\!\!\!\!&&{p_Z}\left( {{z_1},{z_2},{z_3},{z_4}} \right) =
L_{{S_1},{S_2},{S_3},{S_4}}^{ - 1}\left\{ {{\mu _Z}\left( { - {S_1}, - {S_2}, - {S_3}, - {S_4}} \right)} \right\} \nonumber \\
\!\!\!\!\!\!\!\!\!\!\!\!\!\!\!\!\!\!\!\!\!\!\!\!\!\!\!\!\!\!\!\!\!\!\!\!\!\!\!\!\!\!\!\!&& = \sum\limits_{\scriptstyle {i_{{N_s}}}, \ldots ,{i_N} \atop 
  \scriptstyle {i_{{N_s}}} \ne  \cdots  \ne {i_N}}^{1,2, \ldots ,N} {\int\limits_0^\infty  {d{u_{{N_s}}}\frac{1}{{{{\bar \gamma }_{{i_{{N_s}}}}}}}\exp \left( { - \frac{{{u_{{N_s}}}}}{{{{\bar \gamma }_{{i_{{N_s}}}}}}}} \right)L_{{S_4}}^{ - 1}\left\{ {\exp \left( { - {u_{{N_s}}}{S_4}} \right)} \right\}\prod\limits_{\scriptstyle j = {N_s} + 1 \atop 
  \scriptstyle \left\{ {{i_{{N_s} + 1}}, \ldots ,{i_N}} \right\}}^N {\left\{ {1 - \exp \left( { - \frac{{{u_{{N_s}}}}}{{{{\bar \gamma }_{{i_j}}}}}} \right)} \right\}} } } 
\times \sum\limits_{\scriptstyle {i_m} = 1 \atop 
  \scriptstyle {i_m} \ne {i_{{N_s}, \ldots ,}}{i_N}}^N {\int\limits_{{u_{{N_s}}}}^\infty  {d{u_m}\frac{1}{{{{\bar \gamma }_{{i_m}}}}}\exp \left( { - \frac{{{u_m}}}{{{{\bar \gamma }_{{i_m}}}}}} \right)L_{{S_2}}^{ - 1}\left\{ {\exp \left( { - {u_m}{S_2}} \right)} \right\}} } \nonumber \\
\!\!\!\!\!\!\!\!\!\!\!\!\!\!\!\!\!\!\!\!\!\!\!\!\!\!\!\!\!\!\!\!\!\!\!\!\!\!\!\!\!\!\!\!&&\quad \times \sum\limits_{\left\{ {{i_{m + 1}}, \ldots ,{i_{{N_s} - 1}}} \right\} \in {{\mathop{\rm P}\nolimits} _{{N_s} - m - 1}}\left( {{I_N} - \left\{ {{i_m}} \right\} - \left\{ {{i_{{N_s}}}, \ldots ,{i_N}} \right\}} \right)} {\sum\limits_{k = m + 1}^{{N_s} - 1} {{C_{k,m+1,{N_s} - 1}}\left[ {L_{{S_3}}^{ - 1}\left\{ {\frac{{\exp \left( { - \left( {{N_s} - m - 1} \right) \cdot {u_{{N_s}}} \cdot {S_3}} \right)}}{{\left( { - {S_3} - \frac{1}{{{{\bar \gamma }_{{i_k}}}}}} \right)}}} \right\}} \right.\exp \left( { - \sum\limits_{l = m + 1}^{{N_s} - 1} {\left( {\frac{{{u_{{N_s}}}}}{{{{\bar \gamma }_{{i_l}}}}}} \right)} } \right)} } \nonumber \\
\!\!\!\!\!\!\!\!\!\!\!\!\!\!\!\!\!\!\!\!\!\!\!\!\!\!\!\!\!\!\!\!\!\!\!\!\!\!\!\!\!\!\!\!&&\quad \left. { + \exp \left( { - \sum\limits_{l = m + 1}^{{N_s} - 1} {\left( {\frac{{{u_{{N_s}}}}}{{{{\bar \gamma }_{{i_l}}}}}} \right)} } \right)\sum\limits_{l = 1}^{{N_s} - m - 1} {L_{{S_3}}^{ - 1}\left\{ {\frac{{\exp \left( { - \left( {l \cdot {u_m} + \left( {{N_s} - m - l - 1} \right) \cdot {u_{{N_s}}}} \right) \cdot {S_3}} \right)}}{{\left( { - {S_3} - \frac{1}{{{{\bar \gamma }_{{i_k}}}}}} \right)}}} \right\}\left\{ {{{\left( { - 1} \right)}^l}\sum\limits_{{j_1} = {j_0} + m + 1}^{{N_s} - l} { \cdots \sum\limits_{{j_l} = {j_{l - 1}} + 1}^{{N_s} - 1} {\exp \left( { - \sum\limits_{m = 1}^l {\frac{{{u_m} - {u_{{N_s}}}}}{{{{\bar \gamma }_{{i_{{j_m}}}}}}}} } \right)} } } \right\}} } \right] \nonumber \\
\!\!\!\!\!\!\!\!\!\!\!\!\!\!\!\!\!\!\!\!\!\!\!\!\!\!\!\!\!\!\!\!\!\!\!\!\!\!\!\!\!\!\!\!&&\quad \times \sum\limits_{\left\{ {{i_1}, \ldots ,{i_{m - 1}}} \right\} \in {{\mathop{\rm P}\nolimits} _{m - 1}}\left( {{I_N} - \left\{ {{i_m}} \right\} - \left\{ {{i_{{N_s}}}, \ldots ,{i_N}} \right\} - \left\{ {{i_{m + 1}}, \ldots ,{i_{{N_s} - 1}}} \right\}} \right)} {\sum\limits_{h = 1}^{m - 1} {{C_{h,1,m-1}}\exp \left( { - \sum\limits_{l = 1}^{m - 1} {\left( {\frac{{{u_m}}}{{{{\bar \gamma }_{{i_l}}}}}} \right)} } \right)L_{{S_1}}^{ - 1}\left\{ {\frac{{\exp \left( { - \left( {m - 1} \right){u_m}{S_1}} \right)}}{{\left( { - {S_1} - \frac{1}{{{{\bar \gamma }_{{i_h}}}}}} \right)}}} \right\}} },
\end{eqnarray} \normalsize
where
\begin{equation} \footnotesize
L_{{S_4}}^{ - 1}\left\{ {\exp \left( { - {u_{{N_s}}}{S_4}} \right)} \right\} = \delta \left( {{z_4} - {u_{{N_s}}}} \right),
\end{equation}
\begin{equation} \footnotesize
L_{{S_2}}^{ - 1}\left\{ {\exp \left( { - {u_m}{S_2}} \right)} \right\} = \delta \left( {{z_2} - {u_m}} \right),
\end{equation}
\begin{equation} \footnotesize
L_{{S_3}}^{ - 1}\left\{ {\frac{{\exp \left( { - \left( {{N_s} - m - 1} \right) \cdot {u_{{N_s}}} \cdot {S_3}} \right)}}{{\left( { - {S_3} - \frac{1}{{{{\bar \gamma }_{{i_k}}}}}} \right)}}} \right\} = - \exp \left( { - \frac{{{z_3} - \left( {{N_s} - m - 1} \right) \cdot {u_{{N_s}}}}}{{{{\bar \gamma }_{{i_k}}}}}} \right)U\left( {{z_3} - \left( {{N_s} - m - 1} \right) \cdot {u_{{N_s}}}} \right),
\end{equation}
\begin{equation}\footnotesize
L_{{S_3}}^{ - 1}\left\{ {\frac{{\exp \left( { - \left( {l \cdot {u_m} + \left( {{N_s} - m - l - 1} \right) \cdot {u_{{N_s}}}} \right) \cdot {S_3}} \right)}}{{\left( { - {S_3} - \frac{1}{{{{\bar \gamma }_{{i_k}}}}}} \right)}}} \right\} =  - \exp \left( { - \frac{{{z_3} - \left( {l \cdot {u_m} + \left( {{N_s} - m - l - 1} \right) \cdot {u_{{N_s}}}} \right)}}{{{{\bar \gamma }_{{i_k}}}}}} \right)U\left( {{z_3} - \left( {l \cdot {u_m} + \left( {{N_s} - m - l - 1} \right) \cdot {u_{{N_s}}}} \right)} \right),
\end{equation}
\begin{equation} \footnotesize
L_{{S_1}}^{ - 1}\left\{ {\frac{{\exp \left( { - \left( {m - 1} \right){u_m}{S_1}} \right)}}{{\left( { - {S_1} - \frac{1}{{{{\bar \gamma }_{{i_h}}}}}} \right)}}} \right\} =  - \exp \left( { - \frac{{{z_1} - \left( {m - 1} \right) \cdot {u_m}}}{{{{\bar \gamma }_{{i_h}}}}}} \right)U\left( {{z_1} - \left( {m - 1} \right) \cdot {u_m}} \right),
\end{equation}
\begin{equation} \footnotesize
\prod\limits_{k = {n_1}}^{{n_2}} {\left( {1 - \exp \left( { - \frac{{{u_{{N_s}}}}}{{{{\bar \gamma }_{{i_k}}}}}} \right)} \right) = 1 + } \sum\limits_{k = 1}^{{n_2} - {n_1} + 1} {{{\left( { - 1} \right)}^k}\sum\limits_{{j_1} = {j_0} + {n_1}}^{{n_2} - k + 1} { \cdots \sum\limits_{{j_k} = {j_{k - 1}} + 1}^{{n_2}} {\exp \left( { - \sum\limits_{m = 1}^k {\frac{{{u_{{N_s}}}}}{{{{\bar \gamma }_{{i_{{j_m}}}}}}}} } \right)} } }.
\end{equation} \normalsize

\end{enumerate}
\end{landscape}

\section{Application Example}
The above derived joint PDFs of partial sums of ordered statistics can be applied to the performance analysis of various wireless communication systems.
In this section, we discuss several selected application examples.

\subsection{Derivation of the Capture Probability of GSC RAKE receiver over i.n.d. Rayleigh fading conditions}
Recently, we presented the exact performance analyses of the capture probability on GSC RAKE receivers in \cite{kn:capture_outage_GSC}. For analytical simplification, the fading was assumed both independent and identically distributed from path to path. However, the average SNR of each path (or branch) is different for most practical channel models, especially for wide-band SS signals since the average fading power may vary from one path to the other. For example, experimental measurements indicate that the radio channel is characterized by an exponentially decaying multipath intensity profile (MIP) for indoor office buildings \cite{kn:Hashemi} as well as urban \cite{kn:Erceg} and suburban areas \cite{kn:Turin}. Based on this motivation in mind, with the help of our derived results in Sec V, we can extend our previous result (a closed-form formula of the capture probability on GSC RAKE receivers) by maintaining the assumption of independence among the diversity paths but relaxing the identically distributed assumption. 

Let $u_i$ be the order statistics obtained by arranging $N$ $(N\ge2)$ nonnegative i.n.d. RVs, $\left\{ {\gamma_{i_l} } \right\}_{i_l = 1}^N$, in decreasing order of magnitude such that $u_1 \ge u_2 \ge \cdots \ge u_N$. Based on the system model and definition in \cite{kn:capture_outage_GSC}, the capture probability can be written as
\begin{equation} \small \label{eq:Prob_capture}
\text{Prob}_{GSC-capture}=\text{Pr} \left[ {\frac{\sum\limits_{n = 1}^m {{{u}_n}}}{\sum\limits_{n = 1}^N {{{u}_n}}} > T} \right],
\end{equation}
where $0 < T < 1$ and $m < N$.
If we assume $Z=[Z_1, Z_2]$, $Z_1=\sum\limits_{n = 1}^{m}{u_{n}}$ and $Z_2=\sum\limits_{n = m+1}^N{u_{n}}$, then (\ref{eq:Prob_capture}) can be calculated in terms of the 2-dimensional joint PDF of $Z_1$ and $Z_2$ easily as
\begin{equation} \small \label{eq:Capture_probability_closed_form_1}
\text{Prob}_{GSC-capture}=\text{Pr} \left[ {\frac{{{Z_1}}}{{{Z_1} + {Z_2}}} > T} \right] = \int_0^\infty  {\int_0^{\left( {\frac{{1 - T}}{T}} \right){z_1}} {{p_Z}\left( {{z_1},{z_2}} \right)d{z_2}d{z_1}} }. 
\end{equation}
%Here, based on the derived capture probability equation, we still need 
The joint PDF of $\sum\limits_{n = 1}^{m}{u_{n}}$ and $\sum\limits_{n = m+1}^N{u_{n}}$, ${p_Z}\left( {{z_1},{z_2}} \right)$ can be derived with the help of our extended approach in this paper. More specifically, inserting (\ref{eq:non_closed_form_3}) into (\ref{eq:Capture_probability_closed_form_1}),  the closed-form expression for i.n.d. Rayleigh fading conditions is shown at the top of the next page (refer to Appendix-\ref{AP:capture_prob_GSC} for details).
\begin{figure*} [!h]
\setcounter{equation}{62}
\tiny
\begin{eqnarray} \label{eq:Capture_probability_closed_form_2}
\!\!\!\!\!\!\!\!\!\!\!\!\!\!\!\!\!\!\!\!\!\!\!\!\!\!\!\!\!\!&&\!\!\!\!\!\!\!\!\!\!\!\!{\text{Prob}_{GSC-capture}} \nonumber
\\
\!\!\!\!\!\!\!\!\!\!\!\!\!\!\!\!\!\!\!\!\!\!\!\!\!\!\!\!\!\!&=&\!\!\!\!\!\!\!\sum\limits_{{i_m} = 1}^N {\frac{1}{{{{\bar \gamma }_{{i_m}}}}}\sum\limits_{\left\{ {{i_1}, \ldots ,{i_{m - 1}}} \right\} \in {{\mathop{\rm P}\nolimits} _{m - 1}}\left( {{I_N} - \left\{ {{i_m}} \right\}} \right)} {\sum\limits_{\scriptstyle k = 1 \atop 
  \scriptstyle \left\{ {{i_1}, \ldots ,{i_{m - 1}}} \right\}}^{m - 1} {{C_{k,1,m-1}}\sum\limits_{\left\{ {{i_{m + 1}}, \ldots ,{i_N}} \right\} \in {{\mathop{\rm P}\nolimits} _{N - m}}\left( {{I_N} - \left\{ {{i_m}} \right\} - \left\{ {{i_1}, \ldots ,{i_{m - 1}}} \right\}} \right)} {\sum\limits_{\scriptstyle q = m + 1 \atop 
  \scriptstyle \left\{ {{i_{m + 1}}, \ldots ,{i_N}} \right\}}^N {{C_{q,m+1,N}}} } } } } \nonumber
\\
\!\!\!\!\!\!\!\!\!\!\!\!\!\!\!\!\!\!\!\!\!\!\!\!\!\!\!\!\!\!&&\!\!\!\!\!\!\!\quad \times\left[ {\frac{1}{{\left( {\sum\limits_{l = 1}^m {\left( {\frac{1}{{{{\bar \gamma }_{{i_l}}}}}} \right) - \frac{m}{{{{\bar \gamma }_{{i_k}}}}}} } \right)}}} \right]\int_0^\infty  {\int_0^{\left( {\frac{{1 - T}}{T}} \right){z_1}} {\exp \left( { - \frac{{{z_2}}}{{{{\bar \gamma }_{{i_q}}}}}} \right)\exp \left( { - \frac{{{z_1}}}{{{{\bar \gamma }_{{i_k}}}}}} \right)d{z_2}d{z_1}} } \nonumber
\\
\!\!\!\!\!\!\!\!\!\!\!\!\!\!\!\!\!\!\!\!\!\!\!\!\!\!\!\!\!\!&&\!\!\!\!\!\!\!- \sum\limits_{{i_m} = 1}^N {\frac{1}{{{{\bar \gamma }_{{i_m}}}}}\sum\limits_{\left\{ {{i_1}, \ldots ,{i_{m - 1}}} \right\} \in {{\mathop{\rm P}\nolimits} _{m - 1}}\left( {{I_N} - \left\{ {{i_m}} \right\}} \right)} {\sum\limits_{\scriptstyle k = 1 \atop 
  \scriptstyle \left\{ {{i_1}, \ldots ,{i_{m - 1}}} \right\}}^{m - 1} {{C_{k,1,m-1}}\sum\limits_{\left\{ {{i_{m + 1}}, \ldots ,{i_N}} \right\} \in {{\mathop{\rm P}\nolimits} _{N - m}}\left( {{I_N} - \left\{ {{i_m}} \right\} - \left\{ {{i_1}, \ldots ,{i_{m - 1}}} \right\}} \right)} {\sum\limits_{\scriptstyle q = m + 1 \atop 
  \scriptstyle \left\{ {{i_{m + 1}}, \ldots ,{i_N}} \right\}}^N {{C_{q,m+1,N}}} } } } } \nonumber
\\
\!\!\!\!\!\!\!\!\!\!\!\!\!\!\!\!\!\!\!\!\!\!\!\!\!\!\!\!\!\!&&\!\!\!\!\!\!\!\quad \times \left[ {\frac{1}{{\left( {\sum\limits_{l = 1}^m {\left( {\frac{1}{{{{\bar \gamma }_{{i_l}}}}}} \right) - \frac{m}{{{{\bar \gamma }_{{i_k}}}}}} } \right)}}} \right]\int_0^\infty  {\int_0^{\left( {\frac{{1 - T}}{T}} \right){z_1}} {\exp \left( { - \frac{{{z_2}}}{{{{\bar \gamma }_{{i_q}}}}}} \right)\exp \left( { - \left( {\sum\limits_{l = 1}^m {\left( {\frac{1}{{{{\bar \gamma }_{{i_l}}}}}} \right)} } \right)\frac{{{z_1}}}{m}} \right)d{z_2}d{z_1}} }\nonumber
\\
\!\!\!\!\!\!\!\!\!\!\!\!\!\!\!\!\!\!\!\!\!\!\!\!\!\!\!\!\!\!&&\!\!\!\!\!\!\!+ \sum\limits_{{i_m} = 1}^N {\frac{1}{{{{\bar \gamma }_{{i_m}}}}}\sum\limits_{\left\{ {{i_1}, \ldots ,{i_{m - 1}}} \right\} \in {{\mathop{\rm P}\nolimits} _{m - 1}}\left( {{I_N} - \left\{ {{i_m}} \right\}} \right)} {\sum\limits_{\scriptstyle k = 1 \atop 
  \scriptstyle \left\{ {{i_1}, \ldots ,{i_{m - 1}}} \right\}}^{m - 1} {{C_{k,1,m-1}}\sum\limits_{\left\{ {{i_{m + 1}}, \ldots ,{i_N}} \right\} \in {{\mathop{\rm P}\nolimits} _{N - m}}\left( {{I_N} - \left\{ {{i_m}} \right\} - \left\{ {{i_1}, \ldots ,{i_{m - 1}}} \right\}} \right)} {\sum\limits_{\scriptstyle q = m + 1 \atop 
  \scriptstyle \left\{ {{i_{m + 1}}, \ldots ,{i_N}} \right\}}^N {{C_{q,m+1,N}}} } } } } \nonumber
\\
\!\!\!\!\!\!\!\!\!\!\!\!\!\!\!\!\!\!\!\!\!\!\!\!\!\!\!\!\!\!&&\!\!\!\!\!\!\!\quad \times \left[ \!{\sum\limits_{h = 1}^{N - m} \!{\left\{\! {{{\left(\! { - 1} \!\right)}^h}\!\sum\limits_{{j_1} = {j_0} + m + 1}^{N - h + 1} \!{ \cdots \!\sum\limits_{{j_h} = {j_{h - 1}} + 1}^N \!{\left(\! {\frac{1}{{\left( \!{\sum\limits_{m = 1}^h \!{\left( \!{\frac{1}{{{{\bar \gamma }_{{i_{{j_m}}}}}}}} \!\right) \!+ \!\sum\limits_{l = 1}^m \!{\left(\! {\frac{1}{{{{\bar \gamma }_{{i_l}}}}}}\! \right) \!- \!\frac{m}{{{{\bar \gamma }_{{i_k}}}}} - \frac{h}{{{{\bar \gamma }_{{i_q}}}}}} } } \!\right)}}} \!\right)\!\int_0^\infty \! {\int_0^{\left(\! {\frac{{1 - T}}{T}} \!\right){z_1}}\! {\exp \!\left(\! { - \frac{{{z_1}}}{{{{\bar \gamma }_{{i_k}}}}}} \!\right)\!\exp \!\left(\! { - \frac{{{z_2}}}{{{{\bar \gamma }_{{i_q}}}}}}\!\right)\!U\left(\! {\frac{{{z_1}}}{m} - \frac{{{z_2}}}{h}} \!\right)\!d{z_2}d{z_1}} } } } } \!\right\}} } \!\right] \nonumber
\\
\!\!\!\!\!\!\!\!\!\!\!\!\!\!\!\!\!\!\!\!\!\!\!\!\!\!\!\!\!\!&&\!\!\!\!\!\!\! - \sum\limits_{{i_m} = 1}^N {\frac{1}{{{{\bar \gamma }_{{i_m}}}}}\sum\limits_{\left\{ {{i_1}, \ldots ,{i_{m - 1}}} \right\} \in {{\mathop{\rm P}\nolimits} _{m - 1}}\left( {{I_N} - \left\{ {{i_m}} \right\}} \right)} {\sum\limits_{\scriptstyle k = 1 \atop 
  \scriptstyle \left\{ {{i_1}, \ldots ,{i_{m - 1}}} \right\}}^{m - 1} {{C_{k,1,m-1}}\sum\limits_{\left\{ {{i_{m + 1}}, \ldots ,{i_N}} \right\} \in {{\mathop{\rm P}\nolimits} _{N - m}}\left( {{I_N} - \left\{ {{i_m}} \right\} - \left\{ {{i_1}, \ldots ,{i_{m - 1}}} \right\}} \right)} {\sum\limits_{\scriptstyle q = m + 1 \atop 
  \scriptstyle \left\{ {{i_{m + 1}}, \ldots ,{i_N}} \right\}}^N {{C_{q,m+1,N}}} } } } } \nonumber
\\
\!\!\!\!\!\!\!\!\!\!\!\!\!\!\!\!\!\!\!\!\!\!\!\!\!\!\!\!\!\!&&\!\!\!\!\!\!\!\quad \times \Vast[ {\sum\limits_{h = 1}^{N - m} {{{\left( { - 1} \right)}^h}\sum\limits_{{j_1} = {j_0} + m + 1}^{N - h + 1} { \cdots \sum\limits_{{j_h} = {j_{h - 1}} + 1}^N {\left( {\frac{1}{{\left( {\sum\limits_{m = 1}^h {\left( {\frac{1}{{{{\bar \gamma }_{{i_{{j_m}}}}}}}} \right) + \sum\limits_{l = 1}^m {\left( {\frac{1}{{{{\bar \gamma }_{{i_l}}}}}} \right) - \frac{m}{{{{\bar \gamma }_{{i_k}}}}} - \frac{h}{{{{\bar \gamma }_{{i_q}}}}}} } } \right)}}} \right)} } } }\nonumber
\\
\!\!\!\!\!\!\!\!\!\!\!\!\!\!\!\!\!\!\!\!\!\!\!\!\!\!\!\!\!\!&&\!\!\!\!\!\!\!\quad \quad \quad \times {\int_0^\infty  {\int_0^{\left( {\frac{{1 - T}}{T}} \right){z_1}} {\exp \left( { - \frac{{{z_1}}}{{{{\bar \gamma }_{{i_k}}}}}} \right)\exp \left( { - \left( {\sum\limits_{m = 1}^h {\left( {\frac{1}{{{{\bar \gamma }_{{i_{{j_m}}}}}}}} \right) + \sum\limits_{l = 1}^m {\left( {\frac{1}{{{{\bar \gamma }_{{i_l}}}}}} \right) - \frac{m}{{{{\bar \gamma }_{{i_k}}}}}} } } \right)\frac{{{z_2}}}{h}} \right)U\left( {\frac{{{z_1}}}{m} - \frac{{{z_2}}}{h}} \right)d{z_2}d{z_1}} } }\Vast] \nonumber 
\\
\!\!\!\!\!\!\!\!\!\!\!\!\!\!\!\!\!\!\!\!\!\!\!\!\!\!\!\!\!\!&&\!\!\!\!\!\!\! + \sum\limits_{{i_m} = 1}^N {\frac{1}{{{{\bar \gamma }_{{i_m}}}}}\sum\limits_{\left\{ {{i_1}, \ldots ,{i_{m - 1}}} \right\} \in {{\mathop{\rm P}\nolimits} _{m - 1}}\left( {{I_N} - \left\{ {{i_m}} \right\}} \right)} {\sum\limits_{\scriptstyle k = 1 \atop 
  \scriptstyle \left\{ {{i_1}, \ldots ,{i_{m - 1}}} \right\}}^{m - 1} {{C_{k,1,m-1}}\sum\limits_{\left\{ {{i_{m + 1}}, \ldots ,{i_N}} \right\} \in {{\mathop{\rm P}\nolimits} _{N - m}}\left( {{I_N} - \left\{ {{i_m}} \right\} - \left\{ {{i_1}, \ldots ,{i_{m - 1}}} \right\}} \right)} {\sum\limits_{\scriptstyle q = m + 1 \atop 
  \scriptstyle \left\{ {{i_{m + 1}}, \ldots ,{i_N}} \right\}}^N {{C_{q,m+1,N}}} } } } } \nonumber
\\
\!\!\!\!\!\!\!\!\!\!\!\!\!\!\!\!\!\!\!\!\!\!\!\!\!\!\!\!\!\!&&\!\!\!\!\!\!\!\quad \times \Vast[ {\sum\limits_{h = 1}^{N - m} {{{\left( { - 1} \right)}^h}\sum\limits_{{j_1} = {j_0} + m + 1}^{N - h + 1} { \cdots \sum\limits_{{j_h} = {j_{h - 1}} + 1}^N {\left( {\frac{1}{{\left( {\sum\limits_{m = 1}^h {\left( {\frac{1}{{{{\bar \gamma }_{{i_{{j_m}}}}}}}} \right) + \sum\limits_{l = 1}^m {\left( {\frac{1}{{{{\bar \gamma }_{{i_l}}}}}} \right) - \frac{m}{{{{\bar \gamma }_{{i_k}}}}} - \frac{h}{{{{\bar \gamma }_{{i_q}}}}}} } } \right)}}} \right)} } } } \nonumber
\\
\!\!\!\!\!\!\!\!\!\!\!\!\!\!\!\!\!\!\!\!\!\!\!\!\!\!\!\!\!\!&&\!\!\!\!\!\!\!\quad \quad \quad \times {\int_0^\infty  {\int_0^{\left( {\frac{{1 - T}}{T}} \right){z_1}} {\exp \left( { - \frac{{{z_1}}}{{{{\bar \gamma }_{{i_k}}}}}} \right)\exp \left( { - \frac{{{z_2}}}{{{{\bar \gamma }_{{i_q}}}}}} \right)\left[ {1 - U\left( {\frac{{{z_1}}}{m} - \frac{{{z_2}}}{h}} \right)} \right]d{z_2}d{z_1}} } } \Vast] \nonumber
\\
\!\!\!\!\!\!\!\!\!\!\!\!\!\!\!\!\!\!\!\!\!\!\!\!\!\!\!\!\!\!&&\!\!\!\!\!\!\! - \sum\limits_{{i_m} = 1}^N {\frac{1}{{{{\bar \gamma }_{{i_m}}}}}\sum\limits_{\left\{ {{i_1}, \ldots ,{i_{m - 1}}} \right\} \in {{\mathop{\rm P}\nolimits} _{m - 1}}\left( {{I_N} - \left\{ {{i_m}} \right\}} \right)} {\sum\limits_{\scriptstyle k = 1 \atop 
  \scriptstyle \left\{ {{i_1}, \ldots ,{i_{m - 1}}} \right\}}^{m - 1} {{C_{k,1,m-1}}\sum\limits_{\left\{ {{i_{m + 1}}, \ldots ,{i_N}} \right\} \in {{\mathop{\rm P}\nolimits} _{N - m}}\left( {{I_N} - \left\{ {{i_m}} \right\} - \left\{ {{i_1}, \ldots ,{i_{m - 1}}} \right\}} \right)} {\sum\limits_{\scriptstyle q = m + 1 \atop 
  \scriptstyle \left\{ {{i_{m + 1}}, \ldots ,{i_N}} \right\}}^N {{C_{q,m+1,N}}} } } } } \nonumber
\\
\!\!\!\!\!\!\!\!\!\!\!\!\!\!\!\!\!\!\!\!\!\!\!\!\!\!\!\!\!\!&&\!\!\!\!\!\!\!\quad \times \Vast[{\sum\limits_{h = 1}^{N - m} {{{\left( { - 1} \right)}^h}\sum\limits_{{j_1} = {j_0} + m + 1}^{N - h + 1} { \cdots \sum\limits_{{j_h} = {j_{h - 1}} + 1}^N {\left( {\frac{1}{{\left( {\sum\limits_{m = 1}^h {\left( {\frac{1}{{{{\bar \gamma }_{{i_{{j_m}}}}}}}} \right) + \sum\limits_{l = 1}^m {\left( {\frac{1}{{{{\bar \gamma }_{{i_l}}}}}} \right) - \frac{m}{{{{\bar \gamma }_{{i_k}}}}} - \frac{h}{{{{\bar \gamma }_{{i_q}}}}}} } } \right)}}} \right)} } } } \nonumber
\\
\!\!\!\!\!\!\!\!\!\!\!\!\!\!\!\!\!\!\!\!\!\!\!\!\!\!\!\!\!\!&&\!\!\!\!\!\!\!\quad \quad \quad \times {\int_0^\infty  {\int_0^{\left( {\frac{{1 - T}}{T}} \right){z_1}} {\exp \left( { - \frac{{{z_2}}}{{{{\bar \gamma }_{{i_q}}}}}} \right)\exp \left( { - \left( {\sum\limits_{m = 1}^h {\left( {\frac{1}{{{{\bar \gamma }_{{i_{{j_m}}}}}}}} \right) + \sum\limits_{l = 1}^m {\left( {\frac{1}{{{{\bar \gamma }_{{i_l}}}}}} \right) - \frac{h}{{{{\bar \gamma }_{{i_q}}}}}} } } \right)\frac{{{z_1}}}{m}} \right)\left[ {1 - U\left( {\frac{{{z_1}}}{m} - \frac{{{z_2}}}{h}} \right)} \right]d{z_2}d{z_1}} } } \Vast].
\end{eqnarray}
\hrulefill
\end{figure*}
\clearpage
The closed-form expressions of integral parts in the expression presented in (\ref{eq:Capture_probability_closed_form_2}) can be derived as 
\begin{enumerate}
\item[i)] The first integral part:
\begin{equation} \small \label{eq:Capture_probability_closed_form_int_1}
\int_0^\infty  {\int_0^{\left( {\frac{{1 - T}}{T}} \right){z_1}} {\exp \left( { - \frac{{{z_2}}}{{{{\bar \gamma }_{{i_q}}}}}} \right)\exp \left( { - \frac{{{z_1}}}{{{{\bar \gamma }_{{i_k}}}}}} \right)d{z_2}d{z_1}} } ={{\bar \gamma }_{{i_q}}}{{\bar \gamma }_{{i_k}}} - \frac{{{{\bar \gamma }_{{i_q}}}}}{{\left( {\frac{1}{{T \cdot {{\bar \gamma }_{{i_q}}}}} + \frac{1}{{{{\bar \gamma }_{{i_k}}}}} - \frac{1}{{{{\bar \gamma }_{{i_q}}}}}} \right)}}.
\end{equation}
\item[ii)] The second integral part:
 \small\begin{eqnarray} \label{eq:Capture_probability_closed_form_int_2}
&&\int_0^\infty  {\int_0^{\left( {\frac{{1 - T}}{T}} \right){z_1}} {\exp \left( { - \frac{{{z_2}}}{{{{\bar \gamma }_{{i_q}}}}}} \right)\exp \left( { - \left( {\sum\limits_{l = 1}^m {\left( {\frac{1}{{{{\bar \gamma }_{{i_l}}}}}} \right)} } \right)\frac{{{z_1}}}{m}} \right)d{z_2}d{z_1}} } \nonumber
\\
&&=\frac{{{{\bar \gamma }_{{i_q}}}}}{{\left( {\sum\limits_{l = 1}^m {\left( {\frac{1}{{m \cdot {{\bar \gamma }_{{i_l}}}}}} \right)} } \right)}} - \frac{{{{\bar \gamma }_{{i_q}}}}}{{\left( {\sum\limits_{l = 1}^m {\left( {\frac{1}{{m \cdot {{\bar \gamma }_{{i_l}}}}}} \right) + \frac{{1 - T}}{{T \cdot {{\bar \gamma }_{{i_q}}}}}} } \right)}}.
\end{eqnarray}\normalsize
\item[iii)] The third integral part:
\small\begin{eqnarray} \label{eq:Capture_probability_closed_form_int_3}
&&\int_0^\infty  {\int_0^{\left( {\frac{{1 - T}}{T}} \right){z_1}} {\exp \left( { - \frac{{{z_1}}}{{{{\bar \gamma }_{{i_k}}}}}} \right)\exp \left( { - \frac{{{z_2}}}{{{{\bar \gamma }_{{i_q}}}}}} \right)U\left( {\frac{{{z_1}}}{m} - \frac{{{z_2}}}{h}} \right)d{z_2}d{z_1}} } \nonumber
\\
&&={{\bar \gamma }_{{i_q}}}{{\bar \gamma }_{{i_k}}}U\left( {\frac{1}{m} - \frac{{1 - T}}{{T \cdot h}}} \right) - \frac{{{{\bar \gamma }_{{i_q}}}}}{{\left( {\frac{{1 - T}}{{{{\bar \gamma }_{{i_q}}}T}} + \frac{1}{{{{\bar \gamma }_{{i_k}}}}}} \right)}}U\left( {\frac{1}{m} - \frac{{1 - T}}{{T \cdot h}}} \right) \nonumber
\\
&&\quad+ {{\bar \gamma }_{{i_q}}}{{\bar \gamma }_{{i_k}}}\left[ {1 - U\left( {\frac{1}{m} - \frac{{1 - T}}{{T \cdot h}}} \right)} \right] - \frac{{{{\bar \gamma }_{{i_q}}}}}{{\left( {\frac{h}{{{{\bar \gamma }_{{i_q}}}m}} + \frac{1}{{{{\bar \gamma }_{{i_k}}}}}} \right)}}\left[ {1 - U\left( {\frac{1}{m} - \frac{{1 - T}}{{T \cdot h}}} \right)} \right].
\end{eqnarray} \normalsize
\item[iv)] The forth integral part:
\small\begin{eqnarray} \label{eq:Capture_probability_closed_form_int_4}
\!\!\!\!\!\!\!\!\!\!\!\!\!\!\!\!\!\!\!\!\!\!\!\!\!\!\!\!\!\!\!\!\!\!\!\!\!\!\!\!&&\int_0^\infty \! {\int_0^{\left( \!{\frac{{1 - T}}{T}} \!\right){z_1}} \!{\exp \left( \!{ - \frac{{{z_1}}}{{{{\bar \gamma }_{{i_k}}}}}} \!\right)\!\exp \left( \!{ - \left( {\sum\limits_{m = 1}^h {\left(\! {\frac{1}{{{{\bar \gamma }_{{i_{{j_m}}}}}}}}\! \right) + \sum\limits_{l = 1}^m \!{\left( \!{\frac{1}{{{{\bar \gamma }_{{i_l}}}}}} \!\right) - \frac{m}{{{{\bar \gamma }_{{i_k}}}}}} } } \!\right)\frac{{{z_2}}}{h}} \!\right)\!U\left(\! {\frac{{{z_1}}}{m} - \frac{{{z_2}}}{h}} \!\right)d{z_2}d{z_1}} } \nonumber
\\
\!\!\!\!\!\!\!\!\!\!\!\!\!\!\!\!\!\!\!\!\!\!\!\!\!\!\!\!\!\!\!\!\!\!\!\!\!\!\!\!&&=\frac{{{{\bar \gamma }_{{i_k}}}h}}{{\left( {\sum\limits_{m = 1}^h {\left( {\frac{1}{{{{\bar \gamma }_{{i_{{j_m}}}}}}}} \right) + \sum\limits_{l = 1}^m {\left( {\frac{1}{{{{\bar \gamma }_{{i_l}}}}}} \right) - \frac{m}{{{{\bar \gamma }_{{i_k}}}}}} } } \right)}}U\left( {\frac{1}{m} - \frac{{1 - T}}{{T \cdot h}}} \right)\nonumber
\\
\!\!\!\!\!\!\!\!\!\!\!\!\!\!\!\!\!\!\!\!\!\!\!\!\!\!\!\!\!\!\!\!\!\!\!\!\!\!\!\!&&\quad- \frac{h}{{\left(\! {\sum\limits_{m = 1}^h \!{\left( \!{\frac{1}{{{{\bar \gamma }_{{i_{{j_m}}}}}}}} \!\right) + \sum\limits_{l = 1}^m \!{\left(\! {\frac{1}{{{{\bar \gamma }_{{i_l}}}}}} \!\right) - \frac{m}{{{{\bar \gamma }_{{i_k}}}}}} } } \!\right)\!\left\{\! {\left(\! {\sum\limits_{m = 1}^h \!{\left( \!{\frac{1}{{{{\bar \gamma }_{{i_{{j_m}}}}}}}} \!\right) + \sum\limits_{l = 1}^m \!{\left( \!{\frac{1}{{{{\bar \gamma }_{{i_l}}}}}} \!\right) - \frac{m}{{{{\bar \gamma }_{{i_k}}}}}} } } \!\right)\frac{{1 - T}}{{T \cdot h}} + \frac{1}{{{{\bar \gamma }_{{i_k}}}}}} \!\right\}}}U\left( \!{\frac{1}{m} - \frac{{1 - T}}{{T \cdot h}}} \!\right)\nonumber
\\
\!\!\!\!\!\!\!\!\!\!\!\!\!\!\!\!\!\!\!\!\!\!\!\!\!\!\!\!\!\!\!\!\!\!\!\!\!\!\!\!&& \quad + \frac{{{{\bar \gamma }_{{i_k}}}h}}{{\left( {\sum\limits_{m = 1}^h {\left( {\frac{1}{{{{\bar \gamma }_{{i_{{j_m}}}}}}}} \right) + \sum\limits_{l = 1}^m {\left( {\frac{1}{{{{\bar \gamma }_{{i_l}}}}}} \right) - \frac{m}{{{{\bar \gamma }_{{i_k}}}}}} } } \right)}}\left[ {1 - U\left( {\frac{1}{m} - \frac{{1 - T}}{{T \cdot h}}} \right)} \right]\nonumber
\\
\!\!\!\!\!\!\!\!\!\!\!\!\!\!\!\!\!\!\!\!\!\!\!\!\!\!\!\!\!\!\!\!\!\!\!\!\!\!\!\!&&\quad- \frac{h}{{\left(\! {\sum\limits_{m = 1}^h \!{\left(\! {\frac{1}{{{{\bar \gamma }_{{i_{{j_m}}}}}}}}\! \right) + \sum\limits_{l = 1}^m \!{\left( \!{\frac{1}{{{{\bar \gamma }_{{i_l}}}}}} \!\right) - \frac{m}{{{{\bar \gamma }_{{i_k}}}}}} } }\! \right)\!\left\{\! {\left(\! {\sum\limits_{m = 1}^h \!{\left(\! {\frac{1}{{{{\bar \gamma }_{{i_{{j_m}}}}}}}} \!\right) + \sum\limits_{l = 1}^m \!{\left( \!{\frac{1}{{{{\bar \gamma }_{{i_l}}}}}}\! \right) - \frac{m}{{{{\bar \gamma }_{{i_k}}}}}} } } \!\right)\frac{1}{m} + \frac{1}{{{{\bar \gamma }_{{i_k}}}}}} \!\right\}}}\left[\! {1 - U\left(\! {\frac{1}{m} - \frac{{1 - T}}{{T \cdot h}}} \!\right)} \!\right].
\end{eqnarray} \normalsize
\item[v)] The fifth integral part:
\small\begin{eqnarray} \label{eq:Capture_probability_closed_form_int_5}
&&\int_0^\infty  {\int_0^{\left( {\frac{{1 - T}}{T}} \right){z_1}} {\exp \left( { - \frac{{{z_1}}}{{{{\bar \gamma }_{{i_k}}}}}} \right)\exp \left( { - \frac{{{z_2}}}{{{{\bar \gamma }_{{i_q}}}}}} \right)\left[ {1 - U\left( {\frac{{{z_1}}}{m} - \frac{{{z_2}}}{h}} \right)} \right]d{z_2}d{z_1}} } \nonumber
\\
&&=
\frac{{{{\bar \gamma }_{{i_q}}}}}{{\left( {\frac{h}{{m \cdot {{\bar \gamma }_{{i_q}}}}} + \frac{1}{{{{\bar \gamma }_{{i_k}}}}}} \right)}}U\left( {\frac{{1 - T}}{{T \cdot h}} - \frac{1}{m}} \right) - \frac{{{{\bar \gamma }_{{i_q}}}}}{{\left( {\frac{{1 - T}}{{T \cdot {{\bar \gamma }_{{i_q}}}}} + \frac{1}{{{{\bar \gamma }_{{i_k}}}}}} \right)}}U\left( {\frac{{1 - T}}{{T \cdot h}} - \frac{1}{m}} \right).
\end{eqnarray} \normalsize
\item[vi)] The sixth integral part:
\small\begin{eqnarray} \label{eq:Capture_probability_closed_form_int_6}
\!\!\!\!\!\!\!\!\!\!\!\!\!\!\!\!\!\!\!\!\!\!\!\!\!\!\!&&\int_0^\infty  {\int_0^{\left( {\frac{{1 - T}}{T}} \right){z_1}} {\exp \left( { - \frac{{{z_2}}}{{{{\bar \gamma }_{{i_q}}}}}} \right)\exp \left( { - \left( {\sum\limits_{m = 1}^h {\left( {\frac{1}{{{{\bar \gamma }_{{i_{{j_m}}}}}}}} \right) + \sum\limits_{l = 1}^m {\left( {\frac{1}{{{{\bar \gamma }_{{i_l}}}}}} \right) - \frac{h}{{{{\bar \gamma }_{{i_q}}}}}} } } \right)\frac{{{z_1}}}{m}} \right)\left[ {1 - U\left( {\frac{{{z_1}}}{m} - \frac{{{z_2}}}{h}} \right)} \right]d{z_2}d{z_1}} } \nonumber
\\
\!\!\!\!\!\!\!\!\!\!\!\!\!\!\!\!\!\!\!\!\!\!\!\!\!\!\!&& =\frac{{m \cdot {{\bar \gamma }_{{i_q}}}}}{{\left( {\sum\limits_{m = 1}^h {\left( {\frac{1}{{{{\bar \gamma }_{{i_{{j_m}}}}}}}} \right) + \sum\limits_{l = 1}^m {\left( {\frac{1}{{{{\bar \gamma }_{{i_l}}}}}} \right)} } } \right)}}U\left( {\frac{{1 - T}}{{T \cdot h}} - \frac{1}{m}} \right) \nonumber
\\
\!\!\!\!\!\!\!\!\!\!\!\!\!\!\!\!\!\!\!\!\!\!\!\!\!\!\!&&\quad- \frac{{m \cdot {{\bar \gamma }_{{i_q}}}}}{{\left\{ {\left( {\sum\limits_{m = 1}^h {\left( {\frac{1}{{{{\bar \gamma }_{{i_{{j_m}}}}}}}} \right) + \sum\limits_{l = 1}^m {\left( {\frac{1}{{{{\bar \gamma }_{{i_l}}}}}} \right) - \frac{h}{{{{\bar \gamma }_{{i_q}}}}}} } } \right) + \frac{{m\left( {1 - T} \right)}}{{T \cdot {{\bar \gamma }_{{i_q}}}}}} \right\}}}U\left( {\frac{{1 - T}}{{T \cdot h}} - \frac{1}{m}} \right).
\end{eqnarray} \normalsize
\end{enumerate}

\subsection{Finger Replacement Schemes for RAKE Receivers in the Soft Handover Region over i.n.d. fading channels}
Recently, new finger replacement techniques for RAKE reception in the soft handover (SHO) region \cite{kn:S_Choi_2008_1} has been proposed and analyzed over independent and identical fading (i.i.d.) channel. The proposed schemes are basically based on the block comparison among groups of resolvable paths from different base stations and lead to the reduction of complexity while offering commensurate performance. If we let $Y = \sum\limits_{i = 1}^{{L_c} - {L_s}} {{u_i}}$, ${W_1} = \sum\limits_{i = {L_c} - {L_s} + 1}^{{L_c}} {{u_i}}$ and ${W_n} = \sum\limits_{i = 1}^{{L_s}} {{v_i^n}}$ (for  $n = 2, \ldots ,N$), where  $u_i$ ($i=1,2,\ldots,L_1$) and $v_i^n$ ($i=1,2,\ldots,L_n$) are the order statistics obtained by arranging $L_n$ nonnegative i.n.d. path SNRs corresponding to the $n$th base station ($2\le n\le N$) in descending order, then the RAKE combiner output SNR with GSC is given by  $Y + \max_n{W_n}$. $Y$ and $W_1$ are dependent but $Y$ and $W_n$ are independent.  In practice, the i.i.d. fading assumption on the diversity paths is not always realistic due to, for example, the different adjacent multipath routes with the same path loss. Although non-identical fading is important for practical implementation, \cite{kn:S_Choi_2008_1} have only investigated the non-uniform power delay profile case  only through computer simulation due to the high analysis complexity. The major difficulty in this problem is to derive the joint statistics of ordered exponential variates over non-identical fading assumptions, which can be obtained by applying Theorem~\ref{th:case2_1} and \ref{th:case2_3} of section V. Due to space limitation, the analytical details are omitted in this work.

\subsection{Outage Probability of GSC RAKE Receivers Over i.n.d. Rayleigh Fading Channel subject to self-interference}
Recently, the outage probability of GSC RAKE receivers subject to self-interference over independent and identically distributed  Rayleigh fading channels has been investigated in~\cite{kn:capture_outage_GSC}. Let $\gamma_{i}$ be the SNR of the $i$-th diversity path and $u_i$ ($i=1,2,\ldots,N$) be the order statistics obtained by arranging $N$ $(N\ge2)$ nonnegative i.n.d. RVs, $\left\{ {\gamma_{i} } \right\}_{i = 1}^N$, in decreasing order of magnitude such that $u_1 \ge u_2 \ge \cdots \ge u_N$. Then, the outage probability, denoted by $\rm{P_{Out}}$, is then defined as~\cite{kn:capture_outage_GSC},
\begin{equation} \label{eq:Prob_Outage} \small
{\rm{P_{Out}}}={\rm{Pr}}\left[ \frac{\sum\limits_{n = 1}^{m}{{u_{n}}}}{1+\alpha\sum\limits_{n = m+1}^N{{u_{n}}}}<T\right],
\end{equation}
where $T$ $(0 \le T)$ is the outage threshold and $\alpha$ $(0 \le \alpha \le 1)$ is the self-interference cancellation coefficient (in practice, each path may have the different value of $\alpha$). The closed-form expression for this outage probability over i.i.d. Rayleigh fading paths has been derived and compared to that of partial RAKE receivers. However, the average signal-to-noise ratio (SNR) of each path (or branch) is different for most practical channel models, especially for wide-band spread spectrum signals. As results, to evaluate the outage probability over i.n.d. fading channel subject to self-interference, the major difficulty is to derive the joint PDF of $\sum\limits_{n = 1}^{m}{{u_{n}}}$ and $\sum\limits_{n = m+1}^N{{u_{n}}}$ for i.n.d. case. Fortunately, the target joint PDF can be obtained with the help of Theorem~\ref{th:case1_2} in Section IV.

\section*{Appendices}

\appendices

\section{Derivation of $J_m$} \label{AP:B}
\setcounter{section}{1}
In this appendix, we derive Eq. (\ref{eq:CDF_MGF_multiple}). At first, we derive special case $N=3$ and then we extend this result to general case for arbitrary $N$ and $m$.

\subsection{Special Case} \label{AP:B_Special}
Let us first consider $N=3$ and $m=3$ case as
\begin{equation} \small \label{AP:B_1}
\sum\limits_{\scriptstyle {i_1},{i_2},{i_3} \atop 
  \scriptstyle {i_1} \ne {i_2} \ne {i_3}}^{1,2,3} {\int\limits_0^\infty  {d{u_1}{p_{{i_1}}}\left( {{u_1}} \right)\exp \left( {\lambda {u_1}} \right)\int\limits_0^{{u_1}} {d{u_2}{p_{{i_2}}}\left( {{u_2}} \right)\exp \left( {\lambda {u_2}} \right)\int\limits_0^{{u_2}} {d{u_3}{p_{{i_3}}}\left( {{u_3}} \right)\exp \left( {\lambda {u_3}} \right)} } } }.
\end{equation}
In here, we can rewrite (\ref{AP:B_1}) as

\small\begin{eqnarray} \label{AP:B_2}
 &&\sum\limits_{\scriptstyle {i_1},{i_2},{i_3} \atop 
  \scriptstyle {i_1} \ne {i_2} \ne {i_3}}^{1,2,3} {\int\limits_0^\infty  {d{u_1}{p_{{i_1}}}\left( {{u_1}} \right)\exp \left( {\lambda {u_1}} \right)\int\limits_0^{{u_1}} {d{u_2}{p_{{i_2}}}\left( {{u_2}} \right)\exp \left( {\lambda {u_2}} \right)\int\limits_0^{{u_2}} {d{u_3}{p_{{i_3}}}\left( {{u_3}} \right)\exp \left( {\lambda {u_3}} \right)} } } } \nonumber \\ 
 &&= \sum\limits_{{i_1} = 1}^3 {\int\limits_0^\infty  {d{u_1}{p_{{i_1}}}\left( {{u_1}} \right)\exp \left( {\lambda {u_1}} \right)\sum\limits_{\substack{
 {i_2},{i_3} \\ 
 {i_2} \ne {i_3} \\ 
 {i_2} \ne {i_1} \\ 
 {i_3} \ne {i_1} \\ 
 }}^{1,2,3} {\int\limits_0^{{u_1}} {d{u_2}{p_{{i_2}}}\left( {{u_2}} \right)\exp \left( {\lambda {u_2}} \right)\int\limits_0^{{u_2}} {d{u_3}{p_{{i_3}}}\left( {{u_3}} \right)\exp \left( {\lambda {u_3}} \right)} } } } }.  
\end{eqnarray}\normalsize
To simply (\ref{AP:B_2}), we consider $i_1 = 1,2,3$ separately.
\begin{enumerate}
\item[i)] for $i_1= 1$ \label{AP:B_i}
\\
In this case, we can obtain the following result by deploying (\ref{AP:B_2}) as
\small\begin{eqnarray} \label{AP:B_3}
\!\!\!\!\!\!\!\!\!\!\!\!\!\!\!\!\!\!\!\!\!\!\!\! &&\sum\limits_{\substack{
 {i_2},{i_3} \\ 
 {i_2} \ne {i_3} \\ 
 {i_2} \ne {i_1} \\ 
 {i_3} \ne {i_1} \\ 
 }}^{1,2,3} {\int\limits_0^{{u_1}} {d{u_2}{p_{{i_2}}}\left( {{u_2}} \right)\exp \left( {\lambda {u_2}} \right)\int\limits_0^{{u_2}} {d{u_3}{p_{{i_3}}}\left( {{u_3}} \right)\exp \left( {\lambda {u_3}} \right)} } }  \nonumber \\
\!\!\!\!\!\!\!\!\!\!\!\!\!\!\!\!\!\!\!\!\!\!\!\! &&= \int\limits_0^{{u_1}} {d{u_2}{p_2}\left( {{u_2}} \right)\exp \left( {\lambda {u_2}} \right)\int\limits_0^{{u_2}} {d{u_3}{p_3}\left( {{u_3}} \right)\exp \left( {\lambda {u_3}} \right)} }  + \int\limits_0^{{u_1}} {d{u_2}{p_3}\left( {{u_2}} \right)\exp \left( {\lambda {u_2}} \right)\int\limits_0^{{u_2}} {d{u_3}{p_2}\left( {{u_3}} \right)\exp \left( {\lambda {u_3}} \right)} }.
\end{eqnarray} \normalsize
In (\ref{AP:B_3}), noting that $p_n\left(u_m\right)\exp\left( \lambda u_m \right)={c_n}^{'}\left(u_m,\lambda\right)$, after applying integration by part similar to \cite{kn:unified_approach}, we can obtain the following result
\small \begin{eqnarray} \label{AP:B_4}
&&\int\limits_0^{{u_1}} {d{u_2}{p_2}\left( {{u_2}} \right)\exp \left( {\lambda {u_2}} \right)\int\limits_0^{{u_2}} {d{u_3}{p_3}\left( {{u_3}} \right)\exp \left( {\lambda {u_3}} \right)} } \nonumber
\\
&&= \int\limits_0^{{u_1}} {d{u_2}{c_2}^\prime \left( {{u_2},\lambda } \right){c_3}\left( {{u_2},\lambda } \right)} \nonumber
\\
&&= {c_2}\left( {{u_1},\lambda } \right){c_3}\left( {{u_1},\lambda } \right) - \int\limits_0^{{u_1}} {d{u_2}{c_2}\left( {{u_2},\lambda } \right){c_3}^\prime \left( {{u_2},\lambda } \right)} \nonumber
\\
&&= {c_2}\left( {{u_1},\lambda } \right){c_3}\left( {{u_1},\lambda } \right) - \int\limits_0^{{u_1}} {d{u_2}{p_3}\left( {{u_2}} \right)\exp \left( {\lambda {u_2}} \right)\int\limits_0^{{u_2}} {d{u_3}{p_2}\left( {{u_3}} \right)\exp \left( {\lambda {u_3}} \right)} }.
\end{eqnarray} \normalsize

Using (\ref{AP:B_4}) in (\ref{AP:B_3}) and then some manipulation, we can show

\small\begin{eqnarray} \label{AP:B_5}
&&\sum\limits_{\substack{
 {i_2},{i_3} \\ 
 {i_2} \ne {i_3} \\ 
 {i_2} \ne {i_1} \\ 
 {i_3} \ne {i_1} \\ 
 }}^{1,2,3} {\int\limits_0^{{u_1}} {d{u_2}{p_{{i_2}}}\left( {{u_2}} \right)\exp \left( {\lambda {u_2}} \right)\int\limits_0^{{u_2}} {d{u_3}{p_{{i_3}}}\left( {{u_3}} \right)\exp \left( {\lambda {u_3}} \right)} } }  \nonumber \\
&&= {c_2}\left( {{u_1},\lambda } \right){c_3}\left( {{u_1},\lambda } \right) - \int\limits_0^{{u_1}} {d{u_2}{p_3}\left( {{u_2}} \right)\exp \left( {\lambda {u_2}} \right)\int\limits_0^{{u_2}} {d{u_3}{p_2}\left( {{u_3}} \right)\exp \left( {\lambda {u_3}} \right)} }  \nonumber
\\
&&\quad + \int\limits_0^{{u_1}} {d{u_2}{p_3}\left( {{u_2}} \right)\exp \left( {\lambda {u_2}} \right)\int\limits_0^{{u_2}} {d{u_3}{p_2}\left( {{u_3}} \right)\exp \left( {\lambda {u_3}} \right)} }  \\ 
  &&= {c_2}\left( {{u_1},\lambda } \right){c_3}\left( {{u_1},\lambda } \right).
\end{eqnarray}\normalsize
\item[ii)] for $i_1= 2$ \label{AP:B_ii}
\\
In this case, we can obtain the following result by deploying (\ref{AP:B_2}) as
\small\begin{eqnarray} \label{AP:B_6}
\!\!\!\!\!\!\!\!\!\!\!\!\!\!\!&&\sum\limits_{\substack{
 {i_2},{i_3} \\ 
 {i_2} \ne {i_3} \\ 
 {i_2} \ne {i_1} \\ 
 {i_3} \ne {i_1} \\ 
 }}^{1,2,3} {\int\limits_0^{{u_1}} {d{u_2}{p_{{i_2}}}\left( {{u_2}} \right)\exp \left( {\lambda {u_2}} \right)\int\limits_0^{{u_2}} {d{u_3}{p_{{i_3}}}\left( {{u_3}} \right)\exp \left( {\lambda {u_3}} \right)} } }  \nonumber \\
\!\!\!\!\!\!\!\!\!\!\!\!\!\!\!&&= \int\limits_0^{{u_1}} \!{d{u_2}{p_1}\left(\! {{u_2}} \!\right)\exp \!\left( \!{\lambda {u_2}} \!\right)\!\int\limits_0^{{u_2}} \!{d{u_3}{p_3}\left( \!{{u_3}} \!\right)\exp \!\left(\! {\lambda {u_3}} \!\right)} }  + \int\limits_0^{{u_1}}\! {d{u_2}{p_3}\left(\! {{u_2}} \!\right)\exp \!\left(\! {\lambda {u_2}} \!\right)\!\int\limits_0^{{u_2}}\! {d{u_3}{p_1}\left(\! {{u_3}} \!\right)\exp \!\left( \!{\lambda {u_3}} \!\right)} }.  \end{eqnarray}\normalsize
With (\ref{AP:B_6}), by applying similar approach like \ref{AP:B_i}-i), we can show the following result
\small\begin{eqnarray} \label{AP:B_7}
 &&\sum\limits_{\substack{
 {i_2},{i_3} \\ 
 {i_2} \ne {i_3} \\ 
 {i_2} \ne {i_1} \\ 
 {i_3} \ne {i_1} \\ 
 }}^{1,2,3} {\int\limits_0^{{u_1}} {d{u_2}{p_{{i_2}}}\left( {{u_2}} \right)\exp \left( {\lambda {u_2}} \right)\int\limits_0^{{u_2}} {d{u_3}{p_{{i_3}}}\left( {{u_3}} \right)\exp \left( {\lambda {u_3}} \right)} } }  \nonumber \\
&&= {c_1}\left( {{u_1},\lambda } \right){c_3}\left( {{u_1},\lambda } \right).
\end{eqnarray}\normalsize

\item[iii)] for $i_1= 3$
\\
In this case, we can also obtain the following result by deploying (\ref{AP:B_2}) as
\small\begin{eqnarray}\label{AP:B_8}
\!\!\!\!\!\!\!\!\!\!\!\!\!\!\!&&\sum\limits_{\substack{
 {i_2},{i_3} \\ 
 {i_2} \ne {i_3} \\ 
 {i_2} \ne {i_1} \\ 
 {i_3} \ne {i_1} \\ 
 }}^{1,2,3} {\int\limits_0^{{u_1}} {d{u_2}{p_{{i_2}}}\left( {{u_2}} \right)\exp \left( {\lambda {u_2}} \right)\int\limits_0^{{u_2}} {d{u_3}{p_{{i_3}}}\left( {{u_3}} \right)\exp \left( {\lambda {u_3}} \right)} } }  \nonumber \\
\!\!\!\!\!\!\!\!\!\!\!\!\!\!\!&&= \int\limits_0^{{u_1}} \!{d{u_2}{p_1}\left(\! {{u_2}} \!\right)\exp \!\left( \!{\lambda {u_2}} \!\right)\!\int\limits_0^{{u_2}}\! {d{u_3}{p_2}\left(\! {{u_3}} \!\right)\exp \!\left(\! {\lambda {u_3}} \!\right)} }  + \int\limits_0^{{u_1}}\! {d{u_2}{p_2}\left(\! {{u_2}} \!\right)\exp\! \left( \!{\lambda {u_2}} \!\right)\!\int\limits_0^{{u_2}}\! {d{u_3}{p_1}\left(\! {{u_3}} \!\right)\exp \!\left(\! {\lambda {u_3}} \!\right)} }.
\end{eqnarray} \normalsize
With (\ref{AP:B_8}), by applying similar approach like \ref{AP:B_i}-i) and \ref{AP:B_i}-ii), we can show the following result
\small\begin{eqnarray}\label{AP:B_9}
 &&\sum\limits_{\substack{
 {i_2},{i_3} \\ 
 {i_2} \ne {i_3} \\ 
 {i_2} \ne {i_1} \\ 
 {i_3} \ne {i_1} \\ 
 }}^{1,2,3} {\int\limits_0^{{u_1}} {d{u_2}{p_{{i_2}}}\left( {{u_2}} \right)\exp \left( {\lambda {u_2}} \right)\int\limits_0^{{u_2}} {d{u_3}{p_{{i_3}}}\left( {{u_3}} \right)\exp \left( {\lambda {u_3}} \right)} } }  \nonumber \\
&&= {c_1}\left( {{u_1},\lambda } \right){c_2}\left( {{u_1},\lambda } \right). 
\end{eqnarray} \normalsize
\end{enumerate}
From results (\ref{AP:B_5}), (\ref{AP:B_7}), and (\ref{AP:B_9}), we can finally simplify (\ref{AP:B_1}) as
\small\begin{eqnarray} \label{AP:B_special_final}
&&\sum\limits_{\substack{
 {i_2},{i_3} \\ 
 {i_2} \ne {i_3} \\ 
 {i_2} \ne {i_1} \\ 
 {i_3} \ne {i_1} \\ 
 }}^{1,2,3} {\int\limits_0^{{u_1}} {d{u_2}{p_{{i_2}}}\left( {{u_2}} \right)\exp \left( {\lambda {u_2}} \right)\int\limits_0^{{u_2}} {d{u_3}{p_{{i_3}}}\left( {{u_3}} \right)\exp \left( {\lambda {u_3}} \right)} } }  \nonumber
\\
&&= \sum\limits_{\left\{ {{i_2},{i_3}} \right\} \in {{\mathop{\rm P}\nolimits} _2}\left( {{I_3} - \left\{ {{i_1}} \right\}} \right)} {\prod\limits_{\scriptstyle l = 1 \atop 
  \scriptstyle \left\{ {{i_2},{i_3}} \right\}}^2 {{c_{{i_l}}}\left( {{u_1},\lambda } \right)} } \end{eqnarray} \normalsize

\subsection{General Case}  \label{AP:B_General}
With arbitrary $N$ and $m$, we can re-write (\ref{AP:B_1}) as
\small\begin{eqnarray} \label{AP:B_general}
{J_m} &=& \sum\limits_{\substack{
 {i_m},{i_{m + 1}}, \ldots ,{i_N} \\ 
 {i_m} \ne {i_{m + 1}} \ne  \cdots  \ne {i_N} \\ 
 {i_m} \ne {i_1},{i_2}, \ldots ,{i_{m - 1}} \\ 
 {i_{m + 1}} \ne {i_1},{i_2}, \ldots ,{i_{m - 1}} \\ 
  \vdots  \\ 
 {i_N} \ne {i_1},{i_2}, \ldots ,{i_{m - 1}} \\ 
 }}^{1,2, \ldots ,N} {\int\limits_0^{{u_{m - 1}}} {d{u_m}{p_{{i_m}}}\left( {{u_m}} \right)\exp \left( {\lambda {u_m}} \right)\int\limits_0^{{u_m}} {d{u_{m + 1}}{p_{{i_{m + 1}}}}\left( {{u_{m + 1}}} \right)\exp \left( {\lambda {u_{m + 1}}} \right)}}} \nonumber 
\\
&&{{{\cdots \int\limits_0^{{u_{N - 1}}} {d{u_N}{p_{{i_N}}}\left( {{u_N}} \right)\exp \left( {\lambda {u_N}} \right)} } } }.
\end{eqnarray} \normalsize

By applying the process presented in \ref{AP:B_Special} to (\ref{AP:B_general}) similarly, the (\ref{AP:B_special_final}) can be generalized to arbitrary $N$ and $m$, which leads to the result in Eq. (\ref{eq:CDF_MGF_multiple}) as 
\begin{equation} \small \label{AP:B_general_final}
{J_m} = \sum\limits_{\left\{ {{i_m},{i_{m + 1}}, \ldots ,{i_N}} \right\} \in {{\mathop{\rm P}\nolimits} _{N - m + 1}}\left( {{I_N} - \left\{ {{i_1},{i_2}, \ldots ,{i_{m - 1}}} \right\}} \right)} {\prod\limits_{\scriptstyle l = m \atop 
  \scriptstyle \left\{ {{i_m},{i_{m + 1}}, \ldots ,{i_N}} \right\}}^N {{c_{{i_l}}}\left( {{u_{m - 1}},\lambda } \right)} }.
\end{equation}

\section{Derivation of $J'_{m}$} \label{AP:C}
In this appendix, we derive Eq. (\ref{eq:EDF_MGF_multiple}). At first, we similarly derive special case $N=3$ and $m=3$ and then we extend this result to general case for arbitrary $N$ and $m$.

\subsection{Special Case} \label{AP:C_Special}
Let us first consider $N=3$ and $m=3$ case as
\begin{equation} \small \label{AP:C_1}
\sum\limits_{\scriptstyle {i_1},{i_2},{i_3} \atop 
  \scriptstyle {i_1} \ne {i_2} \ne {i_3}}^{1,2,3} {\int\limits_{{u_4}}^\infty  {d{u_3}{p_{{i_3}}}\left( {{u_3}} \right)\exp \left( {\lambda {u_3}} \right)\int\limits_{{u_3}}^\infty  {d{u_2}{p_{{i_2}}}\left( {{u_2}} \right)\exp \left( {\lambda {u_2}} \right)\int\limits_{{u_2}}^\infty  {d{u_1}{p_{{i_1}}}\left( {{u_1}} \right)\exp \left( {\lambda {u_1}} \right)} } } }.
\end{equation}
In here, similar to \ref{AP:B_Special}, after deploying (\ref{AP:C_1}) and then some manipulation with the help of integral by part based on $p_n\left( u_m \right) \exp \left(\lambda u_m\right)=-{e_n}'\left(\lambda u_m\right)$, we can finally simplify (\ref{AP:C_1}) as
\small \begin{eqnarray} \label{AP:C_special_final}
&&\sum\limits_{\scriptstyle {i_1},{i_2},{i_3} \atop 
  \scriptstyle {i_1} \ne {i_2} \ne {i_3}}^{1,2,3} {\int\limits_{{u_4}}^\infty  {d{u_3}{p_{{i_3}}}\left( {{u_3}} \right)\exp \left( {\lambda {u_3}} \right)\int\limits_{{u_3}}^\infty  {d{u_2}{p_{{i_2}}}\left( {{u_2}} \right)\exp \left( {\lambda {u_2}} \right)\int\limits_{{u_2}}^\infty  {d{u_1}{p_{{i_1}}}\left( {{u_1}} \right)\exp \left( {\lambda {u_1}} \right)} } } }  \nonumber \\
&& = {e_1}\left( {{u_4},\lambda } \right){e_2}\left( {{u_4},\lambda } \right){e_3}\left( {{u_4},\lambda } \right).
\end{eqnarray} \normalsize

\subsection{General Case}  \label{AP:C_General}
With arbitrary $N$ and $m$, we can re-write (\ref{AP:C_1}) as
\small\begin{eqnarray} \label{AP:C_general}
{{J'}_m} &=& \sum\limits_{\substack{
 {i_1},{i_2}, \ldots ,{i_m} \\ 
 {i_1} \ne {i_2} \ne  \cdots  \ne {i_m} \\ 
 {i_1} \ne {i_{m + 1}},{i_{m + 2}}, \ldots ,{i_N} \\ 
 {i_2} \ne {i_{m + 1}},{i_{m + 2}}, \ldots ,{i_N} \\ 
  \vdots  \\ 
 {i_m} \ne {i_{m + 1}},{i_{m + 2}}, \ldots ,{i_N} \\ 
 }}^{1,2, \ldots ,N} {\int\limits_{{u_{m + 1}}}^\infty  {d{u_m}{p_{{i_m}}}\left( {{u_m}} \right)\exp \left( {\lambda {u_m}} \right)\int\limits_{{u_m}}^\infty  {d{u_{m - 1}}{p_{{i_{m - 1}}}}\left( {{u_{m - 1}}} \right)\exp \left( {\lambda {u_{m - 1}}} \right) }}}\nonumber
\\
&&{{{\cdots \int\limits_{{u_2}}^\infty  {d{u_1}{p_{{i_1}}}\left( {{u_1}} \right)\exp \left( {\lambda {u_1}} \right)} } } }.
\end{eqnarray} \normalsize

By applying the process presented in \ref{AP:C_Special} to (\ref{AP:C_general}) similar to \ref{AP:B}, the (\ref{AP:C_special_final}) can be generalized to arbitrary $N$ and $m$, which leads to the result in Eq. (\ref{eq:EDF_MGF_multiple}) as the closed-form
\begin{equation} \small \label{AP:C_general_final}
{{J'}_m} = \sum\limits_{\left\{ {{i_1},{i_2}, \ldots ,{i_m}} \right\} \in {{\mathop{\rm P}\nolimits} _m}\left( {{I_N} - \left\{ {{i_{m + 1}},{i_{m + 2}}, \ldots ,{i_N}} \right\}} \right)} {\prod\limits_{\scriptstyle l = 1 \atop 
  \scriptstyle \left\{ {{i_1},{i_2}, \ldots ,{i_m}} \right\}}^m {{e_{{i_l}}}\left( {{u_{m + 1}},\lambda } \right)} }.
\end{equation}

\section{Derivation of $J''_{a,b}$} \label{AP:D}
In this appendix, we show the derivation of Eq.(\ref{eq:IntervalMGF_multiple}). Similar to the derivation progress of (\ref{eq:CDF_MGF_multiple}) and (\ref{eq:EDF_MGF_multiple}), we first derive special case $N=3$ and $m=3$ and then we extend this result to general case for arbitrary $N$ and $m$.

\subsection{Special Case} \label{AP:D_Special}
Let us first consider $N=3$ and $m=3$ case  as
\begin{equation} \small \label{AP:D_1}
\sum\limits_{\scriptstyle {i_1},{i_2},{i_3} \atop 
  \scriptstyle {i_1} \ne {i_2} \ne {i_3}}^{1,2,3} {\int\limits_{{u_5}}^{{u_1}} {d{u_4}{p_{{i_3}}}\left( {{u_4}} \right)\exp \left( {\lambda {u_4}} \right)\int\limits_{{u_4}}^{{u_1}} {d{u_3}{p_{{i_2}}}\left( {{u_3}} \right)\exp \left( {\lambda {u_3}} \right)\int\limits_{{u_3}}^{{u_1}} {d{u_2}{p_{{i_1}}}\left( {{u_2}} \right)\exp \left( {\lambda {u_2}} \right)} } } }.
\end{equation}
In here, by deploying (\ref{AP:D_1}), (\ref{AP:D_1}) can be re-written as
\small \begin{eqnarray} \label{AP:D_2}
 &&\sum\limits_{\scriptstyle {i_1},{i_2},{i_3} \atop 
  \scriptstyle {i_1} \ne {i_2} \ne {i_3}}^{1,2,3} {\int\limits_{{u_5}}^{{u_1}} {d{u_4}{p_{{i_3}}}\left( {{u_4}} \right)\exp \left( {\lambda {u_4}} \right)\int\limits_{{u_4}}^{{u_1}} {d{u_3}{p_{{i_2}}}\left( {{u_3}} \right)\exp \left( {\lambda {u_3}} \right)\int\limits_{{u_3}}^{{u_1}} {d{u_2}{p_{{i_1}}}\left( {{u_2}} \right)\exp \left( {\lambda {u_2}} \right)} } } }  \nonumber \\ 
  &&= \int\limits_{{u_5}}^{{u_1}} {d{u_4}{p_1}\left( {{u_4}} \right)\exp \left( {\lambda {u_4}} \right)\int\limits_{{u_4}}^{{u_1}} {d{u_3}{p_2}\left( {{u_3}} \right)\exp \left( {\lambda {u_3}} \right)\int\limits_{{u_3}}^{{u_1}} {d{u_2}{p_3}\left( {{u_2}} \right)\exp \left( {\lambda {u_2}} \right)} } }  \nonumber \\ 
  &&\quad + \int\limits_{{u_5}}^{{u_1}} {d{u_4}{p_1}\left( {{u_4}} \right)\exp \left( {\lambda {u_4}} \right)\int\limits_{{u_4}}^{{u_1}} {d{u_3}{p_3}\left( {{u_3}} \right)\exp \left( {\lambda {u_3}} \right)\int\limits_{{u_3}}^{{u_1}} {d{u_2}{p_2}\left( {{u_2}} \right)\exp \left( {\lambda {u_2}} \right)} } }  \nonumber \\ 
  &&\quad + \int\limits_{{u_5}}^{{u_1}} {d{u_4}{p_2}\left( {{u_4}} \right)\exp \left( {\lambda {u_4}} \right)\int\limits_{{u_4}}^{{u_1}} {d{u_3}{p_1}\left( {{u_3}} \right)\exp \left( {\lambda {u_3}} \right)\int\limits_{{u_3}}^{{u_1}} {d{u_2}{p_3}\left( {{u_2}} \right)\exp \left( {\lambda {u_2}} \right)} } }  \nonumber \\ 
  &&\quad + \int\limits_{{u_5}}^{{u_1}} {d{u_4}{p_2}\left( {{u_4}} \right)\exp \left( {\lambda {u_4}} \right)\int\limits_{{u_4}}^{{u_1}} {d{u_3}{p_3}\left( {{u_3}} \right)\exp \left( {\lambda {u_3}} \right)\int\limits_{{u_3}}^{{u_1}} {d{u_2}{p_2}\left( {{u_2}} \right)\exp \left( {\lambda {u_2}} \right)} } }  \nonumber \\ 
  &&\quad + \int\limits_{{u_5}}^{{u_1}} {d{u_4}{p_3}\left( {{u_4}} \right)\exp \left( {\lambda {u_4}} \right)\int\limits_{{u_4}}^{{u_1}} {d{u_3}{p_1}\left( {{u_3}} \right)\exp \left( {\lambda {u_3}} \right)\int\limits_{{u_3}}^{{u_1}} {d{u_2}{p_2}\left( {{u_2}} \right)\exp \left( {\lambda {u_2}} \right)} } }  \nonumber \\ 
  &&\quad + \int\limits_{{u_5}}^{{u_1}} {d{u_4}{p_3}\left( {{u_4}} \right)\exp \left( {\lambda {u_4}} \right)\int\limits_{{u_4}}^{{u_1}} {d{u_3}{p_2}\left( {{u_3}} \right)\exp \left( {\lambda {u_3}} \right)\int\limits_{{u_3}}^{{u_1}} {d{u_2}{p_1}\left( {{u_2}} \right)\exp \left( {\lambda {u_2}} \right)} } }.   
\end{eqnarray} \normalsize
In (\ref{AP:D_2}), using similar manipulations with (\ref{AP:B_4}) to the ones used in the previous Appendices \ref{AP:B} and \ref{AP:C}, the first, the second and the third multiple integral terms can be also re-written as, respectively 
\small \begin{eqnarray} \label{AP:D_3}
\!\!\!\!\! &&\int\limits_{{u_5}}^{{u_1}}\! {d{u_4}{p_1}\!\left( \!{{u_4}} \!\right)\exp \!\left( \!{\lambda {u_4}}\! \right)\!\int\limits_{{u_4}}^{{u_1}}\! {d{u_3}{p_2}\!\left(\! {{u_3}}\! \right)\exp\! \left( \!{\lambda {u_3}}\! \right)\!\int\limits_{{u_3}}^{{u_1}}\! {d{u_2}{p_3}\left(\! {{u_2}}\! \right)\exp\! \left(\! {\lambda {u_2}} \!\right)} } }  \nonumber \\
\!\!\!\!\!&&=\! \int\limits_{{u_5}}^{{u_1}}\! {d{u_4}{p_1}\left(\! {{u_4}} \!\right)\exp \!\left(\! {\lambda {u_4}} \!\right)\left\{\! {{c_2}\left(\! {{u_4},\lambda } \!\right){c_3}\left(\! {{u_4},\lambda } \!\right) - {c_2}\left(\! {{u_4},\lambda } \!\right){c_3}\left(\! {{u_1},\lambda } \!\right) + \int\limits_{{u_4}}^{{u_1}}\! {d{u_3}{p_3}\left( \!{{u_3}}\! \right)\exp \!\left(\! {\lambda {u_3}}\! \right){c_2}\left(\! {{u_3},\lambda } \!\right)} } \!\right\}}  ,
\end{eqnarray}
\begin{eqnarray} \label{AP:D_4}
\!\!\!\!\! &&\int\limits_{{u_5}}^{{u_1}}\! {d{u_4}{p_1}\!\left( \!{{u_4}} \!\right)\exp \!\left(\! {\lambda {u_4}} \!\right)\!\int\limits_{{u_4}}^{{u_1}}\! {d{u_3}{p_3}\left( \!{{u_3}} \!\right)\exp \!\left( \!{\lambda {u_3}} \!\right)\!\int\limits_{{u_3}}^{{u_1}} \!{d{u_2}{p_2}\left(\! {{u_2}} \!\right)\exp\! \left( \!{\lambda {u_2}}\! \right)} } }  \nonumber \\
\!\!\!\!\!&&=\! \int\limits_{{u_5}}^{{u_1}} \!{d{u_4}{p_1}\left(\! {{u_4}} \!\right)\exp\! \left(\! {\lambda {u_4}} \!\right)\left\{\! {{c_3}\left(\! {{u_1},\lambda } \!\right){c_2}\left(\! {{u_1},\lambda } \!\right) - {c_3}\left(\! {{u_4},\lambda } \!\right){c_2}\left(\! {{u_1},\lambda } \!\right) - \int\limits_{{u_4}}^{{u_1}} \!{d{u_3}{p_3}\left(\! {{u_3}} \!\right)\exp \!\left(\! {\lambda {u_3}} \!\right){c_2}\left(\! {{u_3},\lambda } \!\right)} }\! \right\}},
\end{eqnarray}
\begin{eqnarray} \label{AP:D_5}
\!\!\!\!\!&&\int\limits_{{u_5}}^{{u_1}} \!{d{u_4}{p_2}\left( \!{{u_4}} \!\right)\exp\! \left(\! {\lambda {u_4}}\! \right)\!\int\limits_{{u_4}}^{{u_1}} \!{d{u_3}{p_1}\left( \!{{u_3}} \!\right)\exp\! \left( \!{\lambda {u_3}} \!\right)\!\int\limits_{{u_3}}^{{u_1}} \!{d{u_2}{p_3}\left(\! {{u_2}} \!\right)\exp\! \left(\! {\lambda {u_2}} \!\right)} } }  \nonumber \\
\!\!\!\!\!&&= \!\int\limits_{{u_5}}^{{u_1}}\! {d{u_4}{p_2}\left( \!{{u_4}} \!\right)\exp\! \left(\! {\lambda {u_4}} \!\right)\left\{\! {{c_1}\left( \!{{u_4},\lambda }\! \right){c_3}\left(\! {{u_4},\lambda } \!\right) - {c_1}\left(\! {{u_4},\lambda } \!\right){c_3}\left(\! {{u_1},\lambda } \!\right) + \int\limits_{{u_4}}^{{u_1}}\! {d{u_3}{p_3}\left(\! {{u_3}} \!\right)\exp \!\left( \!{\lambda {u_3}} \!\right){c_1}\left(\! {{u_3},\lambda } \!\right)} } \!\right\}}.
\end{eqnarray} \normalsize

Similarly in (\ref{AP:D_2}), the $4$-th, $5$-th and the final multiple integral terms can be also re-written as respectively 
\small \begin{eqnarray} \label{AP:D_6}
\!\!\!\!\!&&\int\limits_{{u_5}}^{{u_1}}\! {d{u_4}{p_2}\left(\! {{u_4}} \!\right)\exp \!\left(\! {\lambda {u_4}} \!\right)\!\int\limits_{{u_4}}^{{u_1}}\! {d{u_3}{p_3}\left( \!{{u_3}} \!\right)\exp \!\left(\! {\lambda {u_3}} \!\right)\!\int\limits_{{u_3}}^{{u_1}}\! {d{u_2}{p_1}\left( \!{{u_2}} \!\right)\exp \!\left( \!{\lambda {u_2}} \!\right)} } }  \nonumber \\
\!\!\!\!\!&&=\! \int\limits_{{u_5}}^{{u_1}} \!{d{u_4}{p_2}\left( \!{{u_4}} \!\right)\exp \!\left(\! {\lambda {u_4}} \!\right)\left\{ \!{{c_3}\left( \!{{u_1},\lambda } \!\right){c_1}\left(\! {{u_1},\lambda } \!\right) - {c_3}\left(\! {{u_4},\lambda } \!\right){c_1}\left( \!{{u_1},\lambda } \!\right) - \int\limits_{{u_4}}^{{u_1}}\! {d{u_3}{p_3}\left(\! {{u_3}} \!\right)\exp \!\left( \!{\lambda {u_3}} \!\right){c_1}\left( \!{{u_3},\lambda } \!\right)} } \!\right\}},
\end{eqnarray}
\begin{eqnarray} \label{AP:D_7}
\!\!\!\!\!&&\int\limits_{{u_5}}^{{u_1}} \!{d{u_4}{p_3}\left(\! {{u_4}} \!\right)\exp\! \left( \!{\lambda {u_4}}\! \right)\!\int\limits_{{u_4}}^{{u_1}}\! {d{u_3}{p_1}\left(\! {{u_3}} \!\right)\exp\! \left( \!{\lambda {u_3}} \!\right)\!\int\limits_{{u_3}}^{{u_1}} \!{d{u_2}{p_2}\left( \!{{u_2}} \!\right)\exp \!\left( \!{\lambda {u_2}} \!\right)} } }  \nonumber \\
\!\!\!\!\!&&=\! \int\limits_{{u_5}}^{{u_1}} \!{d{u_4}{p_3}\left(\! {{u_4}} \!\right)\exp \!\left(\! {\lambda {u_4}} \!\right)\left\{\! {{c_1}\left(\! {{u_4},\lambda } \!\right){c_2}\left( \!{{u_4},\lambda } \!\right) - {c_1}\left(\! {{u_4},\lambda } \!\right){c_2}\left( \!{{u_1},\lambda } \!\right) + \int\limits_{{u_4}}^{{u_1}}\! {d{u_3}{p_2}\left(\! {{u_3}} \!\right)\exp \!\left(\! {\lambda {u_3}} \!\right){c_1}\left(\! {{u_3},\lambda } \!\right)} } \!\right\}},
\end{eqnarray}
\begin{eqnarray} \label{AP:D_8}
\!\!\!\!\!&&\int\limits_{{u_5}}^{{u_1}}\! {d{u_4}{p_3}\left(\! {{u_4}} \!\right)\exp \!\left(\! {\lambda {u_4}} \!\right)\!\int\limits_{{u_4}}^{{u_1}}\! {d{u_3}{p_2}\left(\! {{u_3}} \!\right)\exp\! \left( \!{\lambda {u_3}} \!\right)\!\int\limits_{{u_3}}^{{u_1}} \!{d{u_2}{p_1}\left(\! {{u_2}} \!\right)\exp \!\left(\! {\lambda {u_2}} \!\right)} } }  \nonumber \\
\!\!\!\!\!&&=\! \int\limits_{{u_5}}^{{u_1}} \!{d{u_4}{p_3}\left(\! {{u_4}}\! \right)\exp \!\left(\! {\lambda {u_4}} \!\right)\left\{ \!{{c_2}\left(\! {{u_1},\lambda }\! \right){c_1}\left(\! {{u_1},\lambda } \!\right) - {c_2}\left( \!{{u_4},\lambda } \!\right){c_1}\left(\! {{u_1},\lambda } \!\right) - \int\limits_{{u_4}}^{{u_1}} \!{d{u_3}{p_2}\left(\! {{u_3}} \!\right)\exp\! \left(\! {\lambda {u_3}} \!\right){c_1}\left(\! {{u_3},\lambda } \!\right)} }\! \right\}}.
\end{eqnarray} \normalsize

Using all the above results from (\ref{AP:D_3}) to (\ref{AP:D_8}) in (\ref{AP:D_2}) and then after some manipulations similar to the one used in previous Appendices \ref{AP:B} and \ref{AP:C}, (\ref{AP:D_2}) can be simplified as
\small \begin{eqnarray} \label{AP:D_9}
&&\sum\limits_{\scriptstyle {i_1},{i_2},{i_3} \atop 
  \scriptstyle {i_1} \ne {i_2} \ne {i_3}}^{1,2,3} {\int\limits_{{u_5}}^{{u_1}} {d{u_4}{p_{{i_3}}}\left( {{u_4}} \right)\exp \left( {\lambda {u_4}} \right)\int\limits_{{u_4}}^{{u_1}} {d{u_3}{p_{{i_2}}}\left( {{u_3}} \right)\exp \left( {\lambda {u_3}} \right)\int\limits_{{u_3}}^{{u_1}} {d{u_2}{p_{{i_1}}}\left( {{u_2}} \right)\exp \left( {\lambda {u_2}} \right)} } } }  \nonumber \\
&&= \int\limits_{{u_5}}^{{u_1}} {d{u_4}{p_1}\left( {{u_4}} \right)\exp \left( {\lambda {u_4}} \right)\left\{ {{c_2}\left( {{u_1},\lambda } \right) - {c_2}\left( {{u_4},\lambda } \right)} \right\}\left\{ {{c_3}\left( {{u_1},\lambda } \right) - {c_3}\left( {{u_4},\lambda } \right)} \right\}}  \nonumber \\ 
  &&\quad + \int\limits_{{u_5}}^{{u_1}} {d{u_4}{p_2}\left( {{u_4}} \right)\exp \left( {\lambda {u_4}} \right)\left\{ {{c_1}\left( {{u_1},\lambda } \right) - {c_1}\left( {{u_4},\lambda } \right)} \right\}\left\{ {{c_3}\left( {{u_1},\lambda } \right) - {c_3}\left( {{u_4},\lambda } \right)} \right\}}  \nonumber \\ 
  &&\quad + \int\limits_{{u_5}}^{{u_1}} {d{u_4}{p_3}\left( {{u_4}} \right)\exp \left( {\lambda {u_4}} \right)\left\{ {{c_1}\left( {{u_1},\lambda } \right) - {c_1}\left( {{u_4},\lambda } \right)} \right\}\left\{ {{c_2}\left( {{u_1},\lambda } \right) - {c_2}\left( {{u_4},\lambda } \right)} \right\}}.
\end{eqnarray} \normalsize
In (\ref{AP:D_9}), after applying (\ref{AP:B_4}) to the first integral terms and then some manipulations, it can be simply re-written as 
\small \begin{eqnarray} \label{AP:D_10}
&& \int\limits_{{u_5}}^{{u_1}} {d{u_4}{p_1}\left( {{u_4}} \right)\exp \left( {\lambda {u_4}} \right)\left\{ {{c_2}\left( {{u_1},\lambda } \right) - {c_2}\left( {{u_4},\lambda } \right)} \right\}\left\{ {{c_3}\left( {{u_1},\lambda } \right) - {c_3}\left( {{u_4},\lambda } \right)} \right\}}  \nonumber \\
&&=  - {c_1}\left( {{u_5},\lambda } \right)\left\{ {{c_2}\left( {{u_1},\lambda } \right) - {c_2}\left( {{u_5},\lambda } \right)} \right\}\left\{ {{c_3}\left( {{u_1},\lambda } \right) - {c_3}\left( {{u_5},\lambda } \right)} \right\} \nonumber \\ 
  &&\quad + \int\limits_{{u_5}}^{{u_1}} {d{u_4}{p_2}\left( {{u_4}} \right)\exp \left( {\lambda {u_4}} \right){c_1}\left( {{u_4},\lambda } \right)\left\{ {{c_3}\left( {{u_1},\lambda } \right) - {c_3}\left( {{u_4},\lambda } \right)} \right\}}  \nonumber \\ 
  &&\quad + \int\limits_{{u_5}}^{{u_1}} {d{u_4}{p_3}\left( {{u_4}} \right)\exp \left( {\lambda {u_4}} \right){c_1}\left( {{u_4},\lambda } \right)\left\{ {{c_2}\left( {{u_1},\lambda } \right) - {c_2}\left( {{u_4},\lambda } \right)} \right\}}. 
\end{eqnarray}  \normalsize
Using (\ref{AP:D_10}) in (\ref{AP:D_9}), (\ref{AP:D_9}) can be simplified as
\small \begin{eqnarray} \label{AP:D_11}
&&\sum\limits_{\scriptstyle {i_1},{i_2},{i_3} \atop 
  \scriptstyle {i_1} \ne {i_2} \ne {i_3}}^{1,2,3} {\int\limits_{{u_5}}^{{u_1}} {d{u_4}{p_{{i_3}}}\left( {{u_4}} \right)\exp \left( {\lambda {u_4}} \right)\int\limits_{{u_4}}^{{u_1}} {d{u_3}{p_{{i_2}}}\left( {{u_3}} \right)\exp \left( {\lambda {u_3}} \right)\int\limits_{{u_3}}^{{u_1}} {d{u_2}{p_{{i_1}}}\left( {{u_2}} \right)\exp \left( {\lambda {u_2}} \right)} } } }  \nonumber \\
&&=  - {c_1}\left( {{u_5},\lambda } \right)\left\{ {{c_2}\left( {{u_1},\lambda } \right) - {c_2}\left( {{u_5},\lambda } \right)} \right\}\left\{ {{c_3}\left( {{u_1},\lambda } \right) - {c_3}\left( {{u_5},\lambda } \right)} \right\} \nonumber \\ 
  &&\quad + \int\limits_{{u_5}}^{{u_1}} {d{u_4}{p_2}\left( {{u_4}} \right)\exp \left( {\lambda {u_4}} \right){c_1}\left( {{u_1},\lambda } \right)\left\{ {{c_3}\left( {{u_1},\lambda } \right) - {c_3}\left( {{u_4},\lambda } \right)} \right\}}  \nonumber \\ 
  &&\quad + \int\limits_{{u_5}}^{{u_1}} {d{u_4}{p_3}\left( {{u_4}} \right)\exp \left( {\lambda {u_4}} \right){c_1}\left( {{u_1},\lambda } \right)\left\{ {{c_2}\left( {{u_1},\lambda } \right) - {c_2}\left( {{u_4},\lambda } \right)} \right\}} .
\end{eqnarray}  \normalsize
In (\ref{AP:D_11}), after applying (\ref{AP:B_4}) to the first integral terms with the help of  similar manipulations used in (\ref{AP:D_10}), the first integral terms in (\ref{AP:D_11}) can be simply re-written as 
\small \begin{eqnarray} \label{AP:D_12}
 &&\int\limits_{{u_5}}^{{u_1}} {d{u_4}{p_2}\left( {{u_4}} \right)\exp \left( {\lambda {u_4}} \right){c_1}\left( {{u_1},\lambda } \right)\left\{ {{c_3}\left( {{u_1},\lambda } \right) - {c_3}\left( {{u_4},\lambda } \right)} \right\}}  \nonumber \\
&&=  - {c_1}\left( {{u_1},\lambda } \right){c_2}\left( {{u_5},\lambda } \right)\left\{ {{c_3}\left( {{u_1},\lambda } \right) - {c_3}\left( {{u_5},\lambda } \right)} \right\} \nonumber \\ 
  &&\quad + \int\limits_{{u_5}}^{{u_1}} {d{u_4}{p_3}\left( {{u_4}} \right)\exp \left( {\lambda {u_4}} \right){c_1}\left( {{u_1},\lambda } \right){c_2}\left( {{u_4},\lambda } \right)} . 
\end{eqnarray}  \normalsize

Now, using (\ref{AP:D_12}), after adding  (\ref{AP:D_12}) and the second integral term in (\ref{AP:D_11}), we can obtain the following result

\small \begin{eqnarray} \label{AP:D_13}
&&\int\limits_{{u_5}}^{{u_1}} {d{u_4}{p_2}\left( {{u_4}} \right)\exp \left( {\lambda {u_4}} \right){c_1}\left( {{u_4},\lambda } \right)\left\{ {{c_3}\left( {{u_1},\lambda } \right) - {c_3}\left( {{u_4},\lambda } \right)} \right\}}  \nonumber \\ 
&&\quad + \int\limits_{{u_5}}^{{u_1}} {d{u_4}{p_3}\left( {{u_4}} \right)\exp \left( {\lambda {u_4}} \right){c_1}\left( {{u_4},\lambda } \right)\left\{ {{c_2}\left( {{u_1},\lambda } \right) - {c_2}\left( {{u_4},\lambda } \right)} \right\}} \nonumber \\
&&=  - {c_1}\left( {{u_1},\lambda } \right){c_2}\left( {{u_5},\lambda } \right)\left\{ {{c_3}\left( {{u_1},\lambda } \right) - {c_3}\left( {{u_5},\lambda } \right)} \right\} \nonumber \\ 
  &&\quad + \left\{ {{c_3}\left( {{u_1},\lambda } \right) - {c_3}\left( {{u_5},\lambda } \right)} \right\}{c_1}\left( {{u_1},\lambda } \right){c_2}\left( {{u_1},\lambda } \right) \nonumber \\ 
  &&= {c_1}\left( {{u_1},\lambda } \right)\left\{ {{c_2}\left( {{u_1},\lambda } \right) - {c_2}\left( {{u_5},\lambda } \right)} \right\}\left\{ {{c_3}\left( {{u_1},\lambda } \right) - {c_3}\left( {{u_5},\lambda } \right)} \right\} .
\end{eqnarray}  \normalsize
Finally, using (\ref{AP:D_13}) in (\ref{AP:D_11}), (\ref{AP:D_11}) can be re-written as
\small\begin{eqnarray} \label{AP:D_14}
 &&\sum\limits_{\scriptstyle {i_1},{i_2},{i_3} \atop 
  \scriptstyle {i_1} \ne {i_2} \ne {i_3}}^{1,2,3} {\int\limits_{{u_5}}^{{u_1}} {d{u_4}{p_{{i_3}}}\left( {{u_4}} \right)\exp \left( {\lambda {u_4}} \right)\int\limits_{{u_4}}^{{u_1}} {d{u_3}{p_{{i_2}}}\left( {{u_3}} \right)\exp \left( {\lambda {u_3}} \right)\int\limits_{{u_3}}^{{u_1}} {d{u_2}{p_{{i_1}}}\left( {{u_2}} \right)\exp \left( {\lambda {u_2}} \right)} } } }  \nonumber \\
&&= {c_1}\left( {{u_1},\lambda } \right)\left\{ {{c_2}\left( {{u_1},\lambda } \right) - {c_2}\left( {{u_5},\lambda } \right)} \right\}\left\{ {{c_3}\left( {{u_1},\lambda } \right) - {c_3}\left( {{u_5},\lambda } \right)} \right\} \nonumber \\ 
  &&\quad - {c_1}\left( {{u_5},\lambda } \right)\left\{ {{c_2}\left( {{u_1},\lambda } \right) - {c_2}\left( {{u_5},\lambda } \right)} \right\}\left\{ {{c_3}\left( {{u_1},\lambda } \right) - {c_3}\left( {{u_5},\lambda } \right)} \right\}. 
\end{eqnarray} \normalsize

By simplifying (\ref{AP:D_14}), we can obtain the final closed-form for special case $N=3$ and $m=3$ as
\small\begin{eqnarray} \label{AP:D_14r}\label{AP:D_special_final}
&&\sum\limits_{\scriptstyle {i_1},{i_2},{i_3} \atop 
  \scriptstyle {i_1} \ne {i_2} \ne {i_3}}^{1,2,3} {\int\limits_{{u_5}}^{{u_1}} {d{u_4}{p_{{i_3}}}\left( {{u_4}} \right)\exp \left( {\lambda {u_4}} \right)\int\limits_{{u_4}}^{{u_1}} {d{u_3}{p_{{i_2}}}\left( {{u_3}} \right)\exp \left( {\lambda {u_3}} \right)\int\limits_{{u_3}}^{{u_1}} {d{u_2}{p_{{i_1}}}\left( {{u_2}} \right)\exp \left( {\lambda {u_2}} \right)} } } }  \nonumber \\&&= \left\{ {{c_1}\left( {{u_1},\lambda } \right) - {c_1}\left( {{u_5},\lambda } \right)} \right\}\left\{ {{c_2}\left( {{u_1},\lambda } \right) - {c_2}\left( {{u_5},\lambda } \right)} \right\}\left\{ {{c_3}\left( {{u_1},\lambda } \right) - {c_3}\left( {{u_5},\lambda } \right)} \right\} \\ 
  &&= {\mu _1}\left( {{u_5},{u_1},\lambda } \right){\mu _2}\left( {{u_5},{u_1},\lambda } \right){\mu _3}\left( {{u_5},{u_1},\lambda } \right).
\end{eqnarray} \normalsize

\subsection{General Case}  \label{AP:D_General}
With arbitrary $N$ and $m$, we can also re-write (\ref{AP:D_1}) as
\small\begin{eqnarray} \label{AP:D_general}
{{J'}_{a,b}} &=& \sum\limits_{\substack{
 {i_{a + 1}}, \ldots ,{i_{b - 1}} \\ 
 {i_{a + 1}} \ne {i_{a + 2}} \ne  \cdots  \ne {i_{b - 1}} \\ 
 {i_{a + 1}} \ne {i_1}, \cdots ,{i_a},{i_b}, \ldots ,{i_N} \\ 
 {i_{a + 2}} \ne {i_1}, \cdots ,{i_a},{i_b}, \ldots ,{i_N} \\ 
  \vdots  \\ 
 {i_{b - 1}} \ne {i_1}, \cdots ,{i_a},{i_b}, \ldots ,{i_N} \\ 
 }}^{1,2, \ldots ,N} {\int\limits_{{u_b}}^{{u_a}} {d{u_{b - 1}}{p_{{i_{b - 1}}}}\left( {{u_{b - 1}}} \right)\exp \left( {\lambda {u_{b - 1}}} \right)\int\limits_{{u_{b - 1}}}^{{u_a}} {d{u_{b - 2}}{p_{{i_{b - 2}}}}\left( {{u_{b - 2}}} \right)\exp \left( {\lambda {u_{b - 2}}} \right) }}}\nonumber
\\
&&{{{\cdots \int\limits_{{u_{a + 2}}}^{{u_a}} {d{u_{a + 1}}{p_{{i_{a + 1}}}}\left( {{u_{a + 1}}} \right)\exp \left( {\lambda {u_{a + 1}}} \right)} } } }.
\end{eqnarray} \normalsize

By applying the similar process presented in \ref{AP:B} and \ref{AP:C}, the (\ref{AP:D_special_final}) can be generalized to arbitrary $N$ and $m$, which leads to the result in Eq. (\ref{eq:IntervalMGF_multiple}) as the closed-form
\begin{equation} \small \label{AP:D_general_final}
{{J'}_{a,b}} = \sum\limits_{\left\{ {{i_{a + 1}}, \ldots ,{i_{b - 1}}} \right\} \in {{\mathop{\rm P}\nolimits} _{b - a + 1}}\left( {{I_N} - \left\{ {{i_1}, \cdots ,{i_a},{i_b}, \ldots ,{i_N}} \right\}} \right)} {\prod\limits_{\scriptstyle l = a + 1 \atop 
  \scriptstyle \left\{ {{i_{a + 1}}, \ldots ,{i_{b - 1}}} \right\}}^{b - 1} {{\mu _{{i_l}}}\left( {{u_b},{u_a},\lambda } \right)} }.
\end{equation}

\section{Derivation of (\ref{eq:joint_MGF_2})}\label{AP:E}
Starting with (\ref{eq:joint_MGF_7_integralform}), with the help of integral solution, (\ref{eq:joint_MGF_7_integralform}) can be simply re-written as
\small \begin{eqnarray} \label{eq:joint_MGF_7_simplified_integralform}
\!\!\!\!\!\!\!\!\!\!\!\!\!\!\!\!\!\!\!\!\!\!\!\!\!\!\!\!&&MGF_Z \left( {\lambda _1 ,\lambda _2 } \right)\nonumber \\
\!\!\!\!\!\!\!\!\!\!\!\!\!\!\!\!\!\!\!\!\!\!\!\!\!\!\!\!&&=\!\!\!\sum\limits_{{i_m} = 1}^N {\int\limits_0^\infty \! {d{u_m}{p_{{i_m}}}\left( {{u_m}} \right)\exp \left( {{\lambda _1}{u_m}} \right)} } \nonumber \\
\!\!\!\!\!\!\!\!\!\!\!\!\!\!\!\!\!\!\!\!\!\!\!\!\!\!\!\!&&\quad \times \!\!\! \sum\limits_{\left\{\! {{i_1}, \ldots ,{i_{m - 1}}} \!\right\} \in {{\mathop{\rm P}\nolimits} _{m - 1}}\left( \!{{I_N} - \left\{\! {{i_m}} \!\right\}} \!\right)} {\int\limits_{{u_m}}^\infty \! {d{u_{m - 1}}{p_{{i_{m - 1}}}}\left( \!{{u_{m - 1}}} \!\right)\exp \!\left(\! {{\lambda _1}{u_{m - 1}}} \!\right) \cdots \int\limits_{{u_2}}^\infty  \!{d{u_1}{p_{{i_1}}}\left( \!{{u_1}} \!\right)\exp \!\left( \!{{\lambda _1}{u_1}} \!\right)} } } \nonumber \\
\!\!\!\!\!\!\!\!\!\!\!\!\!\!\!\!\!\!\!\!\!\!\!\!\!\!\!\!&&\quad \times \!\!\! \sum\limits_{\left\{\! {{i_{m + 1}}, \ldots ,{i_N}} \!\right\} \in {{\mathop{\rm P}\nolimits} _{N - m}}\left(\! {{I_N} - \left\{\! {{i_m}} \!\right\} - \left\{\! {{i_1}, \ldots ,{i_{m - 1}}} \!\right\}} \!\right)}\! {\int\limits_0^{{u_m}} \!{d{u_{m + 1}}{p_{{i_{m + 1}}}}\left(\! {{u_{m + 1}}} \!\right)\exp\! \left(\! {{\lambda _2}{u_{m + 1}}} \!\right) \cdots \!\!\! \int\limits_0^{{u_{N - 1}}}\!\!\! {d{u_N}{p_{{i_N}}}\left(\! {{u_N}} \!\right)\exp \!\left( \!{{\lambda _2}{u_N}} \!\right)} } }.
\end{eqnarray} \normalsize
In (\ref{eq:joint_MGF_7_simplified_integralform}), by simply applying (\ref{eq:CDF_MGF_multiple}) and (\ref{eq:EDF_MGF_multiple}), we can easily obtain each of the following results
\small \begin{eqnarray} 
\!\!\!\!&&\sum\limits_{\left\{\! {{i_{m + 1}}, \ldots ,{i_N}} \!\right\} \in {{\mathop{\rm P}\nolimits} _{N - m}}\left(\! {{I_N} - \left\{\! {{i_m}} \!\right\} - \left\{\! {{i_1}, \ldots ,{i_{m - 1}}} \!\right\}} \!\right)} {\int\limits_0^{{u_m}} \! {d{u_{m + 1}}{p_{{i_{m + 1}}}}\left(\! {{u_{m + 1}}} \!\right)\exp\! \left( \!{{\lambda _2}{u_{m + 1}}} \!\right) \cdots \int\limits_0^{{u_{N - 1}}} \!{d{u_N}{p_{{i_N}}}\left( \!{{u_N}} \!\right)\exp \!\left(\! {{\lambda _2}{u_N}} \!\right)} } }  \nonumber \\
\!\!\!\!&&= \sum\limits_{\left\{ {{i_{m + 1}}, \ldots ,{i_N}} \right\} \in {{\mathop{\rm P}\nolimits} _{N - m}}\left( {{I_N} - \left\{ {{i_m}} \right\} - \left\{ {{i_1}, \ldots ,{i_{m - 1}}} \right\}} \right)} {\prod\limits_{\scriptstyle l = m + 1 \atop 
  \scriptstyle \left\{ {{i_{m + 1}}, \ldots ,{i_N}} \right\}}^N {{c_{{i_l}}}\left( {{u_m},{\lambda _2}} \right)} }, \label{eq:joint_MGF_8_integralform} \\ 
\!\!\!\!&&\sum\limits_{\left\{ {{i_1}, \ldots ,{i_{m - 1}}} \right\} \in {{\mathop{\rm P}\nolimits} _{m - 1}}\left( {{I_N} - \left\{ {{i_m}} \right\}} \right)} {\int\limits_{{u_m}}^\infty  {d{u_{m - 1}}{p_{{i_{m - 1}}}}\left( {{u_{m - 1}}} \right)\exp \left( {{\lambda _1}{u_{m - 1}}} \right) \cdots \int\limits_{{u_2}}^\infty  {d{u_1}{p_{{i_1}}}\left( {{u_1}} \right)\exp \left( {{\lambda _1}{u_1}} \right)} } }  \nonumber\\
\!\!\!\!&&= \sum\limits_{\left\{ {{i_1}, \ldots ,{i_{m - 1}}} \right\} \in {{\mathop{\rm P}\nolimits} _{m - 1}}\left( {{I_N} - \left\{ {{i_m}} \right\}} \right)} {\prod\limits_{\scriptstyle k = 1 \atop 
  \scriptstyle \left\{ {{i_1}, \ldots ,{i_{m - 1}}} \right\}}^{m - 1} {{e_{{i_k}}}\left( {{u_m},{\lambda _1}} \right)} }. \label{eq:joint_MGF_8_1_integralform}
\end{eqnarray} \normalsize
By inserting (\ref{eq:joint_MGF_8_integralform}) and (\ref{eq:joint_MGF_8_1_integralform}) in order into (\ref{eq:joint_MGF_7_integralform}), we can obtain the second order MGF of $Z_1  = \sum\limits_{n = 1}^m {u_n }$ and $Z_2  = \sum\limits_{n = m + 1}^N {u_n }$ as

\small
\begin{eqnarray} \label{eq:APP_joint_MGF_2}
MGF_Z \left( {\lambda _1 ,\lambda _2 } \right) &=& \sum\limits_{{i_m} = 1}^N {\int\limits_0^\infty  {d{u_m}{p_{{i_m}}}\left( {{u_m}} \right)\exp \left( {{\lambda _1}{u_m}} \right)} } \nonumber \\
&&\times \sum\limits_{\left\{ {{i_1}, \ldots ,{i_{m - 1}}} \right\} \in {{\mathop{\rm P}\nolimits} _{m - 1}}\left( {{I_N} - \left\{ {{i_m}} \right\}} \right)} {\prod\limits_{\scriptstyle k = 1 \atop 
  \scriptstyle \left\{ {{i_1}, \ldots ,{i_{m - 1}}} \right\}}^{m - 1} {{e_{{i_k}}}\left( {{u_m},{\lambda _1}} \right)} } \nonumber \\
&&\times \sum\limits_{\left\{ {{i_{m + 1}}, \ldots ,{i_N}} \right\} \in {{\mathop{\rm P}\nolimits} _{N - m}}\left( {{I_N} - \left\{ {{i_m}} \right\} - \left\{ {{i_1}, \ldots ,{i_{m - 1}}} \right\}} \right)} {\prod\limits_{\scriptstyle l = m + 1 \atop 
  \scriptstyle \left\{ {{i_{m + 1}}, \ldots ,{i_N}} \right\}}^N {{c_{{i_l}}}\left( {{u_m},{\lambda _2}} \right)} }.
\end{eqnarray} \normalsize

\section{Derivation of the joint PDF of $u_{m}$ and $\sum\limits_{\scriptstyle n = 1 \hfill \atop \scriptstyle n \ne m \hfill}^{N_s } {u_n }$ for $1<m<N_s-1$ among $N$ ordered RVs}\label{AP:G}

In this Appendix, we derive the joint PDF of $u_{m}$ and $\sum\limits_{\scriptstyle n = 1 \hfill \atop \scriptstyle n \ne m \hfill}^{N_s } {u_{n} }$ among $N$ ordered RVs by considering $1<m<N_s-1$.

Let $Z_1  = \sum\limits_{n = 1}^{m - 1} {u_{n} }$, $Z_2  = u_m$, $Z_3 = \sum\limits_{n = m + 1}^{N_s  - 1} {u_{n} }$ and $Z_4  = u_{N_s}$. The 4-dimensional MGF of $Z=\left[Z_1,Z_2,Z_3,Z_4\right]$ is given by the expectation
\small \begin{eqnarray} \label{eq:APP_joint_MGF_GSC_2_integralform}
\!\!\!\!\!\!\!\!\!\!\!\! && MGF_Z \left( {\lambda _1 ,\lambda _2 ,\lambda _3 ,\lambda _4 } \right) \!\! = \!\! E\left\{ {\exp \left( {\lambda _1 Z_1  + \lambda _2 Z_2  + \lambda _3 Z_3  + \lambda _4 Z_4 } \right)} \right\} \nonumber
\\ 
\!\!\!\!\!\!\!\!\!\!\!\!  &&=\sum\limits_{\substack{
   {{i_1},{i_2}, \cdots ,{i_N}}  \\
   {{i_1} \ne {i_2} \ne  \cdots  \ne {i_N}}  \\
}}^{1,2, \cdots ,N} {\int\limits_0^\infty  {d{u_1}{p_{{i_1}}}\left( {{u_1}} \right)\exp \left( {{\lambda _1}{u_1}} \right) \cdots \int\limits_0^{{u_{m - 2}}} {d{u_{m - 1}}{p_{{i_{m - 1}}}}\left( {{u_{m - 1}}} \right)\exp \left( {{\lambda _1}{u_{m - 1}}} \right)} } } \nonumber
\\ 
\!\!\!\!\!\!\!\!\!\!\!\!  && \quad \times \int\limits_0^{{u_{m - 1}}} {d{u_m}{p_{{i_m}}}\left( {{u_m}} \right)\exp \left( {{\lambda _2}{u_m}} \right)} \nonumber
\\ 
\!\!\!\!\!\!\!\!\!\!\!\!  && \quad \times \int\limits_0^{{u_m}} {d{u_{m + 1}}{p_{{i_{m + 1}}}}\left( {{u_{m + 1}}} \right)\exp \left( {{\lambda _3}{u_{m + 1}}} \right) \cdots \int\limits_0^{{u_{{N_s} - 2}}} {d{u_{{N_s} - 1}}{p_{{i_{{N_s} - 1}}}}\left( {{u_{{N_s} - 1}}} \right)\exp \left( {{\lambda _3}{u_{{N_s} - 1}}} \right)} } \nonumber
\\ 
\!\!\!\!\!\!\!\!\!\!\!\!  && \quad \times \int\limits_0^{{u_{{N_s} - 1}}} {d{u_{{N_s}}}{p_{{i_{{N_s}}}}}\left( {{u_{{N_s}}}} \right)\exp \left( {{\lambda _4}{u_{{N_s}}}} \right)\prod\limits_{j = {N_s} + 1}^N {{P_{{i_j}}}\left( {{u_{{N_s}}}} \right)} }. 
\end{eqnarray} \normalsize
With the help of integral solution presented in \cite{kn:unified_approach}, (\ref{eq:CDF_MGF_multiple}), (\ref{eq:EDF_MGF_multiple}) and (\ref{eq:IntervalMGF_multiple}), we can easily obtain the 4-dimensional MGF of $Z_1$, $Z_2$, $Z_3$ and $Z_4$ as
\small \begin{eqnarray} \label{eq:APP_joint_MGF_GSC_2}
&& MGF_Z \left( {\lambda _1 ,\lambda _2 ,\lambda _3 ,\lambda _4 } \right) \nonumber
\\
&&= \sum\limits_{\scriptstyle {i_{{N_s}}}, \ldots ,{i_N} \atop 
  \scriptstyle {i_{{N_s}}} \ne  \cdots  \ne {i_N}}^{1,2, \ldots ,N} {\int\limits_0^\infty  {d{u_{{N_s}}}{p_{{i_{{N_s}}}}}\left( {{u_{{N_s}}}} \right)\exp \left( {{\lambda _4}{u_{{N_s}}}} \right)\prod\limits_{\scriptstyle j = {N_s} + 1 \atop 
  \scriptstyle \left\{ {{i_{{N_s} + 1}}, \ldots ,{i_N}} \right\}}^N {{P_{{i_j}}}\left( {{u_{{N_s}}}} \right)} } } \nonumber
\\
&&\quad \times \sum\limits_{\scriptstyle {i_m} = 1 \atop 
  \scriptstyle {i_m} \ne {i_{{N_s}, \ldots ,}}{i_N}}^N {\int\limits_{{u_{{N_s}}}}^\infty  {d{u_m}{p_{{i_m}}}\left( {{u_m}} \right)\exp \left( {{\lambda _2}{u_m}} \right)} } \nonumber
\\
&& \quad \times \sum\limits_{\left\{ {{i_{m + 1}}, \ldots ,{i_{{N_s} - 1}}} \right\} \in {{\mathop{\rm P}\nolimits} _{{N_s} - m - 1}}\left( {{I_N} - \left\{ {{i_m}} \right\} - \left\{ {{i_{{N_s}}}, \ldots ,{i_N}} \right\}} \right)} {\prod\limits_{\scriptstyle k = m + 1 \atop 
  \scriptstyle \left\{ {{i_{m + 1}}, \ldots ,{i_{{N_s} - 1}}} \right\}}^{{N_s} - 1} {{\mu _{{i_k}}}\left( {{u_{{N_s}}},{u_m},{\lambda _3}} \right)} } \nonumber
\\
&&\quad \times \sum\limits_{\left\{ {{i_1}, \ldots ,{i_{m - 1}}} \right\} \in {{\mathop{\rm P}\nolimits} _{m - 1}}\left( {{I_N} - \left\{ {{i_m}} \right\} - \left\{ {{i_{{N_s}}}, \ldots ,{i_N}} \right\} - \left\{ {{i_{m + 1}}, \ldots ,{i_{{N_s} - 1}}} \right\}} \right)} {\prod\limits_{\scriptstyle l = 1 \atop 
  \scriptstyle \left\{ {{i_1}, \ldots ,{i_{m - 1}}} \right\}}^{m - 1} {{e_{{i_l}}}\left( {{u_m},{\lambda _1}} \right)} }.
\end{eqnarray} \normalsize

Having a MGF expression given in (\ref{eq:APP_joint_MGF_GSC_2}), we are now in the position to derive the 4-dimensional joint PDF of $Z_1  = \sum\limits_{n = 1}^{m - 1} {u_{n} }$, $Z_2  = u_{m}$, $Z_3 = \sum\limits_{n = m + 1}^{N_s  - 1} {u_{n} }$ and $Z_4  = u_{N_s}$. Letting $\lambda _1  =  - S_1$, $\lambda _2  =  - S_2$, $\lambda _3  =  - S_3$,  and $\lambda _4  =  - S_4$  we can derive the 4-dimensional PDF of $Z_1$, $Z_2$, $Z_3$ and $Z_4$ by applying an inverse Laplace transform yielding
\small \begin{eqnarray} \label{eq:APP_joint_PDF_GSC_2}
\!\!\!\!\!\!\!\!\!\!\!\!\!\!\!\! && p_Z \left( {z_1 ,z_2 ,z_3 ,z_4 } \right) = \mathcal{L}_{S_1 ,S_2 ,S_3 ,S_4 }^{ - 1} \left\{ {MGF_Z \left( { - S_1 , - S_2 , - S_3 , - S_4 } \right)} \right\} \nonumber
\\ 
\!\!\!\!\!\!\!\!\!\!\!\!\!\!\!\!  &&= \sum\limits_{\scriptstyle {i_{{N_s}}}, \ldots ,{i_N} \atop 
  \scriptstyle {i_{{N_s}}} \ne  \cdots  \ne {i_N}}^{1,2, \ldots ,N} {\int\limits_0^\infty  {d{u_{{N_s}}}{p_{{i_{{N_s}}}}}\left( {{u_{{N_s}}}} \right)L_{{S_4}}^{ - 1}\left\{ {\exp \left( { - {S_4}{u_{{N_s}}}} \right)} \right\}\prod\limits_{\scriptstyle j = {N_s} + 1 \atop 
  \scriptstyle \left\{ {{i_{{N_s} + 1}}, \ldots ,{i_N}} \right\}}^N {{P_{{i_j}}}\left( {{u_{{N_s}}}} \right)} } } \nonumber
\\ 
\!\!\!\!\!\!\!\!\!\!\!\!\!\!\!\!  &&\quad \times \sum\limits_{\scriptstyle {i_m} = 1 \atop 
  \scriptstyle {i_m} \ne {i_{{N_s}, \ldots ,}}{i_N}}^N {\int\limits_{{u_{{N_s}}}}^\infty  {d{u_m}{p_{{i_m}}}\left( {{u_m}} \right)L_{{S_2}}^{ - 1}\left\{ {\exp \left( { - {S_2}{u_m}} \right)} \right\}} } \nonumber
\\
\!\!\!\!\!\!\!\!\!\!\!\!\!\!\!\!  &&\quad \times \sum\limits_{\left\{ {{i_{m + 1}}, \ldots ,{i_{{N_s} - 1}}} \right\} \in {{\mathop{\rm P}\nolimits} _{{N_s} - m - 1}}\left( {{I_N} - \left\{ {{i_m}} \right\} - \left\{ {{i_{{N_s}}}, \ldots ,{i_N}} \right\}} \right)} {L_{{S_3}}^{ - 1}\left\{ {\prod\limits_{\scriptstyle k = m + 1 \atop 
  \scriptstyle \left\{ {{i_{m + 1}}, \ldots ,{i_{{N_s} - 1}}} \right\}}^{{N_s} - 1} {{\mu _{{i_k}}}\left( {{u_{{N_s}}},{u_m}, - {S_3}} \right)} } \right\}} \nonumber
\\ 
\!\!\!\!\!\!\!\!\!\!\!\!\!\!\!\!  &&\quad \times \sum\limits_{\left\{ {{i_1}, \ldots ,{i_{m - 1}}} \right\} \in {{\mathop{\rm P}\nolimits} _{m - 1}}\left( {{I_N} - \left\{ {{i_m}} \right\} - \left\{ {{i_{{N_s}}}, \ldots ,{i_N}} \right\} - \left\{ {{i_{m + 1}}, \ldots ,{i_{{N_s} - 1}}} \right\}} \right)} {L_{{S_1}}^{ - 1}\left\{ {\prod\limits_{\scriptstyle l = 1 \atop 
  \scriptstyle \left\{ {{i_1}, \ldots ,{i_{m - 1}}} \right\}}^{m - 1} {{e_{{i_l}}}\left( {{u_m}, - {S_1}} \right)} } \right\}}. 
\end{eqnarray} \normalsize

With this 4-dimensional joint PDF, letting $X=Z_2$ and $Y=Z_1+Z_3+Z_4$ we can obtain the 2-dimensional joint PDF of $Z^{'}=[X,Y]$ by integrating over $z_1$ and $z_4$ yielding
\begin{equation} \small\label{eq:AP_final_1}
\small p_{Z^{'}} \left( {x,y} \right) = \int_0^x {\int_{\left( {m - 1} \right)x}^{y - \left(N_s-m\right)z_4 } {p_Z \left( {z_1 ,x,y - z_4 ,z_4 } \right)dz_1 } dz_4 }, 
\end{equation}
or equivalently we can obtain the 2-dimensional joint PDF of $Z^{'}=[X,Y]$ by integrating over $z_3$ and $z_4$ giving
\begin{equation} \small\label{eq:AP_final_2}
\small p_{Z^{'}} \left( {x,y} \right) = \int_0^x {\int_{\left( {N_s  - m - 1} \right)z_4 }^{\left( {N_s  - m - 1} \right)x} {p_Z \left( {y - z_3  - z_4 ,x,z_3 ,z_4 } \right)dz_3 } dz_4 } .
\end{equation}

\section{Derivation of Multiple Product of Common Functions}\label{AP:H}
In \ref{SEC:VI}-ii), (\ref{eq:common_function_Rayleigh_1}), (\ref{eq:common_function_Rayleigh_2}), and (\ref{eq:common_function_Rayleigh_3}) have the form of multiple product of (\ref{eq:common_function_Rayleigh_1_s}), (\ref{eq:common_function_Rayleigh_2_s}), and (\ref{eq:common_function_Rayleigh_3_s}), respectively. Therefore, to apply an inverse LT for deriving final PDF closed-form expressions from MGF expressions, a multiple product expression needs to be converted to a summation expression of $\lambda$ function. In this appendix, we derive simple summation expressions of $\lambda$ function from multiple product expressions. To derive them, the following four formulas should be converted to a summation expression. 
\begin{enumerate}
\item[i)] $\frac{1}{{\prod\limits_l {\left( {1 - {{\bar \gamma }_{{i_l}}}\lambda } \right)} }}$

At first, we derive special case for a) the multiple product from $1$ to $n$ and then we extend this result to general case for b) the multiple product from arbitrary $n_1$ to $n_2$.

For case a), we need to convert the following multiple product from $1$ to $n$ to a summation expression.
\begin{equation} \small \label{AP:H_1}
\frac{1}{{\prod\limits_{l = {1}}^{n} {\left( {1 - {{\bar \gamma }_{{i_l}}}\lambda } \right)} }} .
\end{equation}
With (\ref{AP:H_1}), after deploying the multiple product term and then rearrange and simplify them,  the multiple product term can be converted to the summation expression of just $\lambda$ as
\begin{equation} \small \label{AP:H_2}
 \frac{1}{{\prod\limits_{l = {1}}^{n} {\left( {1 - {{\bar \gamma }_{{i_l}}}\lambda } \right)} }}  
  = \sum\limits_{l = {1}}^{n} {\frac{{{C_{l,1,n}}}}{{\left( {\lambda  - \frac{1}{{{{\bar \gamma }_{{i_l}}}}}} \right)}}},
\end{equation}
where $j_0=0$,
\begin{equation}\small \label{AP:H_3}
{C_{l,1,n}} = \frac{1}{{\prod\limits_{l = {1}}^{n} {\left( { - {{\bar \gamma }_{{i_l}}}} \right)} }{F'\left( {\frac{1}{{{{\bar \gamma }_{{i_l}}}}}} \right)}},
\end{equation}
\begin{equation}\small \label{AP:H_4}
F'\left( x \right) = \left[ {\sum\limits_{l = 1}^{{n} - {1}} {\left( {{n} -  l} \right){x^{{n} - {1} - l}}{{\left( { - 1} \right)}^l}\sum\limits_{{j_1} = {j_0} + {1}}^{{n} - l + 1} { \cdots \sum\limits_{{j_l} = {j_{l - 1}} + 1}^{{n}} {\prod\limits_{m = 1}^l {\frac{1}{{{{\bar \gamma }_{{i_{{j_m}}}}}}}} } } } } \right] + \left( {{n}} \right){x^{{n} - {1}}}.
\end{equation}

For the case of the multiple product from arbitrary $n_1$ to $n_2$, after applying the same derivation progress as (\ref{AP:H_2}), we can obtain the final result as 
\begin{equation} \small \label{AP:H_5}
\frac{1}{{\prod\limits_{l = {n_1}}^{n_2} {\left( {1 - {{\bar \gamma }_{{i_l}}}\lambda } \right)} }}= \sum\limits_{l = {n_1}}^{n_2} {\frac{{{C_{l,n_1,n_2}}}}{{\left( {\lambda  - \frac{1}{{{{\bar \gamma }_{{i_l}}}}}} \right)}}},
\end{equation}
where
\begin{equation} \small \label{AP:H_6} 
{C_{l,n_1,n2}} = \frac{1}{{\prod\limits_{l = {n_1}}^{{n_2}} {\left( { - {{\bar \gamma }_{{i_l}}}} \right)} }{F'\left( {\frac{1}{{{{\bar \gamma }_{{i_l}}}}}} \right)}},
\end{equation}
\small\begin{eqnarray} \label{AP:H_7} 
F'\left( x \right) &&= \left[ {\sum\limits_{l = 1}^{{n_2} - {n_1}} {\left( {{n_2} - {n_1} - l + 1} \right){x^{{n_2} - {n_1} - l}}{{\left( { - 1} \right)}^l}\sum\limits_{{j_1} = {j_0} + {n_1}}^{{n_2} - l + 1} { \cdots \sum\limits_{{j_l} = {j_{l - 1}} + 1}^{{n_2}} {\prod\limits_{m = 1}^l {\frac{1}{{{{\bar \gamma }_{{i_{{j_m}}}}}}}} } } } } \right] \nonumber
\\
&&\quad + \left( {{n_2} - {n_1} + 1} \right){x^{{n_2} - {n_1}}}.
\end{eqnarray} \normalsize

\item[ii)] $\prod\limits_l {\left[ {1 - \exp \left( {\left( {\lambda  - \frac{1}{{{{\bar \gamma }_{{i_l}}}}}} \right){z_a}} \right)} \right]}$

Similar to \ref{AP:H}-i), at first, we derive special case for a) the multiple product from $1$ to $n$ and then we extend this result to general case for b) the multiple product from arbitrary $n_1$ to $n_2$.

For case a), after deploying the multiple product term of exponential function from $1$ to $n$ and then simplify them, the multiple product term can be converted to the summation expression of $\lambda$ as
\small\begin{eqnarray} \label{AP:H_9}
&&\prod\limits_{l = 1}^n {\left[ {1 - \exp \left( {\left( {\lambda  - \frac{1}{{{{\bar \gamma }_{{i_l}}}}}} \right){z_a}} \right)} \right]}  \nonumber
\\
&&= 1 + \left[ {\sum\limits_{l = 1}^n {\exp \left( {l \cdot {z_a} \cdot \lambda } \right)\left\{ {{{\left( { - 1} \right)}^l}\sum\limits_{{j_1} = {j_0} + 1}^{n - l + 1} { \cdots \sum\limits_{{j_l} = {j_{l - 1}} + 1}^n {\exp \left( { - \sum\limits_{m = 1}^l {\frac{{{z_a}}}{{{{\bar \gamma }_{{i_{{j_m}}}}}}}} } \right)} } } \right\}} } \right],
\end{eqnarray} \normalsize
where $j_0=0$.

For case b), after applying the same derivation progress as (\ref{AP:H_9}), the multiple product from arbitrary $n_1$ to $n_2$ can be obtained as 
\small\begin{eqnarray} \label{AP:H_10} 
&&\prod\limits_{l = {n_1}}^{{n_2}} {\left[ {1 - \exp \left( {\left( {\lambda  - \frac{1}{{{{\bar \gamma }_{{i_l}}}}}} \right){z_a}} \right)} \right]}  \nonumber
\\
&&= 1 + \left[ {\sum\limits_{l = 1}^{{n_2} - {n_1} + 1} {\exp \left( {l \cdot {z_a} \cdot \lambda } \right)\left\{ {{{\left( { - 1} \right)}^l}\sum\limits_{{j_1} = {j_0} + {n_1}}^{{n_2} - l + 1} { \cdots \sum\limits_{{j_l} = {j_{l - 1}} + 1}^{{n_2}} {\exp \left( { - \sum\limits_{m = 1}^l {\frac{{{z_a}}}{{{{\bar \gamma }_{{i_{{j_m}}}}}}}} } \right)} } } \right\}} } \right].
\end{eqnarray} \normalsize

\item[iii)] $\prod\limits_l {\left[ {\exp \left( {\left( {\lambda  - \frac{1}{{{{\bar \gamma }_{{i_l}}}}}} \right){z_a}} \right) - \exp \left( {\left( {\lambda  - \frac{1}{{{{\bar \gamma }_{{i_l}}}}}} \right){z_b}} \right)} \right]}$

Similar to \ref{AP:H}-i) and ii), especially, using the similar manipulation used in \ref{AP:H}-i) and ii) in (\ref{AP:H_10}), the final simple summation expression from arbitrary $n_1$ to $n_2$ can be obtained as

\small \begin{eqnarray} \label{AP:H_13}
\!\!\!\!\!\!\!\!\!\!\!\!\!\!\!\!&&\prod\limits_{l = {n_1}}^{{n_2}} {\left[ {\exp \left( {\left( {\lambda  - \frac{1}{{{{\bar \gamma }_{{i_l}}}}}} \right){z_a}} \right) - \exp \left( {\left( {\lambda  - \frac{1}{{{{\bar \gamma }_{{i_l}}}}}} \right){z_b}} \right)} \right]} \nonumber
\\
\!\!\!\!\!\!\!\!\!\!\!\!\!\!\!\!&&= \prod\limits_{l = {n_1}}^{{n_2}} {\exp \left( {\left( {\lambda  - \frac{1}{{{{\bar \gamma }_{{i_l}}}}}} \right){z_a}} \right)} \prod\limits_{l = {n_1}}^{{n_2}} {\left[ {1 - \exp \left( {\left( {\lambda  - \frac{1}{{{{\bar \gamma }_{{i_l}}}}}} \right)\left( {{z_b} - {z_a}} \right)} \right)} \right]}   \nonumber 
\\ 
\!\!\!\!\!\!\!\!\!\!\!\!\!\!\!\!&&= \exp \left( {\left( {{n_2} - {n_1} + 1} \right) \cdot \lambda  \cdot {z_a}} \right)\exp \left( { - \sum\limits_{l = {n_1}}^{{n_2}} {\frac{{{z_a}}}{{{{\bar \gamma }_{{i_l}}}}}} } \right)\nonumber
\\
\!\!\!\!\!\!\!\!\!\!\!\!\!\!\!\!&&\quad \times\left[ {1 + \sum\limits_{l = {n_1}}^{{n_2} - {n_1} + 1} {\exp \left( {l \cdot \left( {{z_b} - {z_a}} \right) \cdot \lambda } \right)\left\{ {{{\left( { - 1} \right)}^l}\sum\limits_{{j_1} = {j_0} + {n_1}}^{{n_2} - l + 1} { \cdots \sum\limits_{{j_l} = {j_{l - 1}} + 1}^{{n_2}} {\exp \left( { - \sum\limits_{m = 1}^l {\frac{{{z_b} - {z_a}}}{{{{\bar \gamma }_{{i_{{j_m}}}}}}}} } \right)} } } \right\}} } \right].
\end{eqnarray} \normalsize

\item[iv)] $\prod\limits_l {\exp \left( {\left( {\lambda  - \frac{1}{{{{\bar \gamma }_{{i_l}}}}}} \right){z_a}} \right)}$

In this case, with the help of the property of exponential multiplication, we can easily derive the summation expression from the multiple product expression from arbitrary $n_1$ to $n_2$, respectively, as

\small \begin{eqnarray} \label{AP:H_15}
 \prod\limits_{l = {n_1}}^{{n_2}} {\exp \left( {\left( {\lambda  - \frac{1}{{{{\bar \gamma }_{{i_l}}}}}} \right){z_a}} \right)}  &=& \exp \left( {\left\{ {\sum\limits_{l = {n_1}}^{{n_2}} {\left( {\lambda  - \frac{1}{{{{\bar \gamma }_{{i_l}}}}}} \right)} } \right\}{z_a}} \right) \nonumber
\\ 
  &=& \exp \left( {\left\{ { - \sum\limits_{l = {n_1}}^{{n_2}} {\left( {\frac{{{z_a}}}{{{{\bar \gamma }_{{i_l}}}}}} \right)} } \right\}} \right)\exp \left( {\left( {{n_2} - {n_1} + 1} \right){z_a}\lambda } \right).
\end{eqnarray} \normalsize
\end{enumerate}
Based on the above results, we can now obtain the summation expressions of (\ref{eq:common_function_Rayleigh_1_s}), (\ref{eq:common_function_Rayleigh_2_s}), and (\ref{eq:common_function_Rayleigh_3_s}) for arbitrary $n_1$ to $n_2$.  With (\ref{eq:common_function_Rayleigh_1_s}), (\ref{eq:common_function_Rayleigh_2_s}), and (\ref{eq:common_function_Rayleigh_3_s}), we can write the multiple product of (\ref{eq:common_function_Rayleigh_1_s}), (\ref{eq:common_function_Rayleigh_2_s}), and (\ref{eq:common_function_Rayleigh_3_s}) for arbitrary $n_1$ to $n_2$ respectively as
\begin{equation} \small  \label{AP:H_16}
\prod\limits_{l = {n_1}}^{{n_2}} {{c_{{i_l}}}\left( {{z_a},\lambda } \right)}= \frac{1}{{\prod\limits_{l = {n_1}}^{{n_2}} {\left( {1 - {{\bar \gamma }_{{i_l}}}\lambda } \right)} }}\prod\limits_{l = {n_1}}^{{n_2}} {\left[ {1 - \exp \left( {\left( {\lambda  - \frac{1}{{{{\bar \gamma }_{{i_l}}}}}} \right){z_a}} \right)} \right]} ,
\end{equation}
\begin{equation} \small  \label{AP:H_17}
\prod\limits_{l = {n_1}}^{n_2} {e_{i_l}}\left( z_a,\lambda  \right)= \frac{1}{{\prod\limits_{l = {n_1}}^{{n_2}} {\left( {1 - {{\bar \gamma }_{{i_l}}}\lambda } \right)} }}\prod\limits_{l = {n_1}}^{{n_2}} {\left[ {\exp \left( {\left( {\lambda  - \frac{1}{{{{\bar \gamma }_{{i_l}}}}}} \right){z_a}} \right)} \right]},
\end{equation}
\begin{equation} \small \label{AP:H_18}
\prod\limits_{l = {n_1}}^{{n_2}} {{\mu _{{i_l}}}\left( {{z_a},{z_b},\lambda } \right)} = \frac{1}{{\prod\limits_{l = {n_1}}^{{n_2}} {\left( {1 - {{\bar \gamma }_{{i_l}}}\lambda } \right)} }}\prod\limits_{l = {n_1}}^{{n_2}} {\left[ {\exp \left( {\left( {\lambda  - \frac{1}{{{{\bar \gamma }_{{i_l}}}}}} \right){z_a}} \right) - \exp \left( {\left( {\lambda  - \frac{1}{{{{\bar \gamma }_{{i_l}}}}}} \right){z_b}} \right)} \right]}.
\end{equation}
For the summation expression of the multiple product of (\ref{eq:common_function_Rayleigh_1_s}) for arbitrary $n_1$ to $n_2$, using (\ref{AP:H_5}) and (\ref{AP:H_10}) in (\ref{AP:H_16}), we can obtain the final summation closed-form expression (\ref{eq:common_function_Rayleigh_1}).

For the summation expression of the multiple product of (\ref{eq:common_function_Rayleigh_2_s}) for arbitrary $n_1$ to $n_2$, using (\ref{AP:H_5}) and (\ref{AP:H_15}) in (\ref{AP:H_17}), we can obtain the final summation closed-form expression (\ref{eq:common_function_Rayleigh_2}).

Finally, for the summation expression of the multiple product of (\ref{eq:common_function_Rayleigh_3_s}) for arbitrary $n_1$ to $n_2$, using (\ref{AP:H_5}) and (\ref{AP:H_13}) in (\ref{AP:H_18}), we can obtain the final summation closed-form expression (\ref{eq:common_function_Rayleigh_3}).

\section{Capture Probability of GSC RAKE receivers}\label{AP:capture_prob_GSC}

\subsection{Joint PDF} \label{AP:joint_PDF}
Starting from (\ref{eq:non_closed_form_3}), we can re-write the joint PDF (\ref{eq:non_closed_form_3}) as
\footnotesize\begin{eqnarray} \label{AP:joint_PDF_1}
\!\!\!\!\!\!\!\!\!\!\!\!\!\!\!\!\!\!&&{p_Z}\left( {{z_1},{z_2}} \right) \nonumber
\\
\!\!\!\!\!\!\!\!\!\!\!\!\!\!\!\!\!\!&&= \!\sum\limits_{{i_m} = 1}^N \!{\frac{1}{{{{\bar \gamma }_{{i_m}}}}}\!\sum\limits_{\left\{\! {{i_1}, \ldots ,{i_{m - 1}}} \!\right\} \in {{\mathop{\rm P}\nolimits} _{m - 1}}\left(\! {{I_N} - \left\{\! {{i_m}} \!\right\}} \!\right)} \!{\sum\limits_{\scriptstyle k = 1 \atop 
  \scriptstyle \left\{ {{i_1}, \ldots ,{i_{m - 1}}} \right\}}^{m - 1} \!{{C_{k,1,m-1}}\sum\limits_{\left\{\! {{i_{m + 1}}, \ldots ,{i_N}} \!\right\} \in {{\mathop{\rm P}\nolimits} _{N - m}}\left( \!{{I_N} - \left\{ {{i_m}} \!\right\} - \left\{\! {{i_1}, \ldots ,{i_{m - 1}}} \!\right\}} \!\right)} {\sum\limits_{\scriptstyle q = m + 1 \atop 
  \scriptstyle \left\{ {{i_{m + 1}}, \ldots ,{i_N}} \right\}}^N \!{{C_{q,m+1,N}}} } } } } \nonumber
\\
\!\!\!\!\!\!\!\!\!\!\!\!\!\!\!\!\!\!&&\quad\quad\times \!\exp \!\left(\! { - \frac{{{z_2}}}{{{{\bar \gamma }_{{i_q}}}}}} \!\right)\exp\! \left(\! { - \frac{{{z_1}}}{{{{\bar \gamma }_{{i_k}}}}}} \!\right)\!\int\limits_0^{\frac{{{z_1}}}{m}} \!{d{u_m}\exp \!\left(\! { - \left(\! {\sum\limits_{l = 1}^m \!{\left(\! {\frac{1}{{{{\bar \gamma }_{{i_l}}}}}} \!\right) \!- \!\frac{m}{{{{\bar \gamma }_{{i_k}}}}}} } \!\right){u_m}} \!\right)} \nonumber
\\
\!\!\!\!\!\!\!\!\!\!\!\!\!\!\!\!\!\!&&\quad+ \!\sum\limits_{{i_m} = 1}^N \!{\frac{1}{{{{\bar \gamma }_{{i_m}}}}}\!\sum\limits_{\left\{\! {{i_1}, \ldots ,{i_{m - 1}}} \!\right\} \in {{\mathop{\rm P}\nolimits} _{m - 1}}\left(\! {{I_N} - \left\{\! {{i_m}} \!\right\}} \!\right)} \!{\sum\limits_{\scriptstyle k = 1 \atop 
  \scriptstyle \left\{ {{i_1}, \ldots ,{i_{m - 1}}} \right\}}^{m - 1} \!{{C_{k,1,m-1}}\sum\limits_{\left\{\! {{i_{m + 1}}, \ldots ,{i_N}} \!\right\} \in {{\mathop{\rm P}\nolimits} _{N - m}}\left(\! {{I_N} - \left\{\! {{i_m}} \!\right\} - \left\{\! {{i_1}, \ldots ,{i_{m - 1}}} \!\right\}} \!\right)} \!{\sum\limits_{\scriptstyle q = m + 1 \atop 
  \scriptstyle \left\{ {{i_{m + 1}}, \ldots ,{i_N}} \right\}}^N \!{{C_{q,m+1,N}}} } } } } \nonumber
\\
\!\!\!\!\!\!\!\!\!\!\!\!\!\!\!\!\!\!&&\quad\quad\times \!\Vast[ \!{\sum\limits_{h = 1}^{N - m}\! {{{\left( \!{ - 1} \!\right)}^h}\!\sum\limits_{{j_1} = {j_0} + m + 1}^{N - h + 1} { \cdots \!\sum\limits_{{j_h} = {j_{h - 1}} + 1}^N \!{\exp \!\left(\! { - \frac{{{z_1}}}{{{{\bar \gamma }_{{i_k}}}}}} \!\right)\exp\! \left(\! { - \frac{{{z_2}}}{{{{\bar \gamma }_{{i_q}}}}}} \!\right)} } } } \nonumber
\\
\!\!\!\!\!\!\!\!\!\!\!\!\!\!\!\!\!\!&& \quad\quad \quad \quad\times \int\limits_0^\infty  \!{d{u_m}} \exp \!\left( \!{ - \left(\! {\sum\limits_{m = 1}^h \! {\left(\! {\frac{1}{{{{\bar \gamma }_{{i_{{j_m}}}}}}}} \!\right) \!+ \!\sum\limits_{l = 1}^m \!{\left(\! {\frac{1}{{{{\bar \gamma }_{{i_l}}}}}} \!\right) \!- \!\frac{m}{{{{\bar \gamma }_{{i_k}}}}} \!-\! \frac{h}{{{{\bar \gamma }_{{i_q}}}}}} } } \!\right){u_m}} \!\right)U\left(\! {{z_1} - m{u_m}} \!\right)U\left(\! {{z_2} - h{u_m}} \!\right) \!\Vast].
\end{eqnarray}\normalsize

In (\ref{AP:joint_PDF_1}), there are two integral expressions and the first integral part can be directly derived as the following closed form expression
\begin{equation} \small \label{AP:joint_PDF_2}
\int\limits_0^{\frac{{{z_1}}}{m}} {d{u_m}\exp \left( { - \left( {\sum\limits_{l = 1}^m {\left( {\frac{1}{{{{\bar \gamma }_{{i_l}}}}}} \right) - \frac{m}{{{{\bar \gamma }_{{i_k}}}}}} } \right){u_m}} \right)}  = \frac{{1 - \exp \left( { - \left( {\sum\limits_{l = 1}^m {\left( {\frac{1}{{{{\bar \gamma }_{{i_l}}}}}} \right) - \frac{m}{{{{\bar \gamma }_{{i_k}}}}}} } \right)\frac{{{z_1}}}{m}} \right)}}{{\left( {\sum\limits_{l = 1}^m {\left( {\frac{1}{{{{\bar \gamma }_{{i_l}}}}}} \right) - \frac{m}{{{{\bar \gamma }_{{i_k}}}}}} } \right)}}. 
\end{equation}
However, for the second integral part, we need to consider two cases separately based on the valid integral region of $z_1$, $z_2$, and $u_m$ as

\small\begin{eqnarray} \label{AP:joint_PDF_3}
&&\int\limits_0^\infty  {d{u_m}} \exp \left( { - \left( {\sum\limits_{m = 1}^h {\left( {\frac{1}{{{{\bar \gamma }_{{i_{{j_m}}}}}}}} \right) + \sum\limits_{l = 1}^m {\left( {\frac{1}{{{{\bar \gamma }_{{i_l}}}}}} \right) - \frac{m}{{{{\bar \gamma }_{{i_k}}}}} - \frac{h}{{{{\bar \gamma }_{{i_q}}}}}} } } \right){u_m}} \right)U\left( {{z_1} - m{u_m}} \right)U\left( {{z_2} - h{u_m}} \right) \nonumber
\\
&&= \int\limits_0^{\frac{{{z_2}}}{h}} {d{u_m}} \exp \left( { - \left( {\sum\limits_{m = 1}^h {\left( {\frac{1}{{{{\bar \gamma }_{{i_{{j_m}}}}}}}} \right) + \sum\limits_{l = 1}^m {\left( {\frac{1}{{{{\bar \gamma }_{{i_l}}}}}} \right) - \frac{m}{{{{\bar \gamma }_{{i_k}}}}} - \frac{h}{{{{\bar \gamma }_{{i_q}}}}}} } } \right){u_m}} \right)U\left( {\frac{{{z_1}}}{m} - \frac{{{z_2}}}{h}} \right)\nonumber
\\
&&\quad+ \int\limits_0^{\frac{{{z_1}}}{m}} {d{u_m}} \exp \left( { - \left( {\sum\limits_{m = 1}^h {\left( {\frac{1}{{{{\bar \gamma }_{{i_{{j_m}}}}}}}} \right) + \sum\limits_{l = 1}^m {\left( {\frac{1}{{{{\bar \gamma }_{{i_l}}}}}} \right) - \frac{m}{{{{\bar \gamma }_{{i_k}}}}} - \frac{h}{{{{\bar \gamma }_{{i_q}}}}}} } } \right){u_m}} \right)\left[ {1 - U\left( {\frac{{{z_1}}}{m} - \frac{{{z_2}}}{h}} \right)} \right].
\end{eqnarray} \normalsize
With simplified (\ref{AP:joint_PDF_3}), we can get the following closed-form expressions, respectively, as

\small\begin{eqnarray}\label{AP:joint_PDF_4}
&&\int\limits_0^{\frac{{{z_2}}}{h}} {d{u_m}} \exp \left( { - \left( {\sum\limits_{m = 1}^h {\left( {\frac{1}{{{{\bar \gamma }_{{i_{{j_m}}}}}}}} \right) + \sum\limits_{l = 1}^m {\left( {\frac{1}{{{{\bar \gamma }_{{i_l}}}}}} \right) - \frac{m}{{{{\bar \gamma }_{{i_k}}}}} - \frac{h}{{{{\bar \gamma }_{{i_q}}}}}} } } \right){u_m}} \right)U\left( {\frac{{{z_1}}}{m} - \frac{{{z_2}}}{h}} \right)\nonumber
\\
&&= \frac{{1 - \exp \left( { - \left( {\sum\limits_{m = 1}^h {\left( {\frac{1}{{{{\bar \gamma }_{{i_{{j_m}}}}}}}} \right) + \sum\limits_{l = 1}^m {\left( {\frac{1}{{{{\bar \gamma }_{{i_l}}}}}} \right) - \frac{m}{{{{\bar \gamma }_{{i_k}}}}} - \frac{h}{{{{\bar \gamma }_{{i_q}}}}}} } } \right)\frac{{{z_2}}}{h}} \right)}}{{\left( {\sum\limits_{m = 1}^h {\left( {\frac{1}{{{{\bar \gamma }_{{i_{{j_m}}}}}}}} \right) + \sum\limits_{l = 1}^m {\left( {\frac{1}{{{{\bar \gamma }_{{i_l}}}}}} \right) - \frac{m}{{{{\bar \gamma }_{{i_k}}}}} - \frac{h}{{{{\bar \gamma }_{{i_q}}}}}} } } \right)}}U\left( {\frac{{{z_1}}}{m} - \frac{{{z_2}}}{h}} \right),
\end{eqnarray}\normalsize
and
\small\begin{eqnarray} \label{AP:joint_PDF_5}
&&\int\limits_0^{\frac{{{z_1}}}{m}} {d{u_m}} \exp \left( { - \left( {\sum\limits_{m = 1}^h {\left( {\frac{1}{{{{\bar \gamma }_{{i_{{j_m}}}}}}}} \right) + \sum\limits_{l = 1}^m {\left( {\frac{1}{{{{\bar \gamma }_{{i_l}}}}}} \right) - \frac{m}{{{{\bar \gamma }_{{i_k}}}}} - \frac{h}{{{{\bar \gamma }_{{i_q}}}}}} } } \right){u_m}} \right)\left[ {1 - U\left( {\frac{{{z_1}}}{m} - \frac{{{z_2}}}{h}} \right)} \right] \nonumber
\\
&&= \frac{{1 - \exp \left( { - \left( {\sum\limits_{m = 1}^h {\left( {\frac{1}{{{{\bar \gamma }_{{i_{{j_m}}}}}}}} \right) + \sum\limits_{l = 1}^m {\left( {\frac{1}{{{{\bar \gamma }_{{i_l}}}}}} \right) - \frac{m}{{{{\bar \gamma }_{{i_k}}}}} - \frac{h}{{{{\bar \gamma }_{{i_q}}}}}} } } \right)\frac{{{z_1}}}{m}} \right)}}{{\left( {\sum\limits_{m = 1}^h {\left( {\frac{1}{{{{\bar \gamma }_{{i_{{j_m}}}}}}}} \right) + \sum\limits_{l = 1}^m {\left( {\frac{1}{{{{\bar \gamma }_{{i_l}}}}}} \right) - \frac{m}{{{{\bar \gamma }_{{i_k}}}}} - \frac{h}{{{{\bar \gamma }_{{i_q}}}}}} } } \right)}}\left[ {1 - U\left( {\frac{{{z_1}}}{m} - \frac{{{z_2}}}{h}} \right)} \right].
\end{eqnarray}\normalsize
\subsection{Capture Probability} \label{AP:capture_prob_CF}
Starting from (\ref{eq:Capture_probability_closed_form_1}), inserting the closed-form expression of (\ref{eq:non_closed_form_3}) presented in \ref{AP:joint_PDF} into (\ref{eq:Capture_probability_closed_form_1}), the closed-form expression for i.n.d. Rayleigh fading conditions can be written in (\ref{eq:Capture_probability_closed_form_2}).
In (\ref{eq:Capture_probability_closed_form_2}), there are six double-integral expressions.
For the first and second cases, we can directly obtain the closed-from expression as shown in (\ref{eq:Capture_probability_closed_form_int_1}) and (\ref{eq:Capture_probability_closed_form_int_2}). However, for others, we need to carefully consider the valid integral region respectively as
\begin{enumerate}
\item[iii)] The third integral expression:
\\
In this case, for valid integration, we need to consider two cases separately. If $\frac{h}{m} \ge \frac{{1 - T}}{T}$, then ${z_2} \le \frac{{1 - T}}{T}{z_1}$ and $\frac{1}{m} \ge \frac{{1 - T}}{{T \cdot h}}$. If $\frac{h}{m} < \frac{{1 - T}}{T}$, then ${z_2} \le \frac{h}{m}{z_1}$ and $\frac{1}{m} < \frac{{1 - T}}{{T \cdot h}}$. As a result, we can re-write the third integral expression as
\small\begin{eqnarray} \label{AP:capture_probability__CF_1}
&&\int_0^\infty  {\int_0^{\left( {\frac{{1 - T}}{T}} \right){z_1}} {\exp \left( { - \frac{{{z_1}}}{{{{\bar \gamma }_{{i_k}}}}}} \right)\exp \left( { - \frac{{{z_2}}}{{{{\bar \gamma }_{{i_q}}}}}} \right)U\left( {\frac{{{z_1}}}{m} - \frac{{{z_2}}}{h}} \right)d{z_2}d{z_1}} }\nonumber
\\
&&=\int_0^\infty  {\exp \left( { - \frac{{{z_1}}}{{{{\bar \gamma }_{{i_k}}}}}} \right)\int_0^{\left( {\frac{{1 - T}}{T}} \right){z_1}} {\exp \left( { - \frac{{{z_2}}}{{{{\bar \gamma }_{{i_q}}}}}} \right)U\left( {\frac{1}{m} - \frac{{1 - T}}{{T \cdot h}}} \right)d{z_2}} d{z_1}} \nonumber
\\
&&\quad + \int_0^\infty  {\exp \left( { - \frac{{{z_1}}}{{{{\bar \gamma }_{{i_k}}}}}} \right)\int_0^{\left( {\frac{h}{m}} \right){z_1}} {\exp \left( { - \frac{{{z_2}}}{{{{\bar \gamma }_{{i_q}}}}}} \right)\left[ {1 - U\left( {\frac{1}{m} - \frac{{1 - T}}{{T \cdot h}}} \right)} \right]d{z_2}} d{z_1}}.
\end{eqnarray} \normalsize
From (\ref{AP:capture_probability__CF_1}), we can directly derive the closed-form expressions as
\small\begin{eqnarray}
&&\int_0^\infty  {\exp \left( { - \frac{{{z_1}}}{{{{\bar \gamma }_{{i_k}}}}}} \right)\int_0^{\left( {\frac{{1 - T}}{T}} \right){z_1}} {\exp \left( { - \frac{{{z_2}}}{{{{\bar \gamma }_{{i_q}}}}}} \right)U\left( {\frac{1}{m} - \frac{{1 - T}}{{T \cdot h}}} \right)d{z_2}} d{z_1}} \nonumber
\\
&& = {{\bar \gamma }_{{i_q}}}{{\bar \gamma }_{{i_k}}}U\left( {\frac{1}{m} - \frac{{1 - T}}{{T \cdot h}}} \right) - \frac{{{{\bar \gamma }_{{i_q}}}}}{{\left( {\frac{{1 - T}}{{{{\bar \gamma }_{{i_q}}}T}} + \frac{1}{{{{\bar \gamma }_{{i_k}}}}}} \right)}}U\left( {\frac{1}{m} - \frac{{1 - T}}{{T \cdot h}}} \right),
\end{eqnarray} \normalsize
and
\small\begin{eqnarray}
&&\int_0^\infty  {\exp \left( { - \frac{{{z_1}}}{{{{\bar \gamma }_{{i_k}}}}}} \right)\int_0^{\left( {\frac{h}{m}} \right){z_1}} {\exp \left( { - \frac{{{z_2}}}{{{{\bar \gamma }_{{i_q}}}}}} \right)\left[ {1 - U\left( {\frac{1}{m} - \frac{{1 - T}}{{T \cdot h}}} \right)} \right]d{z_2}} d{z_1}} \nonumber
\\
&& = {{\bar \gamma }_{{i_q}}}{{\bar \gamma }_{{i_k}}}\left[ {1 - U\left( {\frac{1}{m} - \frac{{1 - T}}{{T \cdot h}}} \right)} \right] - \frac{{{{\bar \gamma }_{{i_q}}}}}{{\left( {\frac{h}{{{{\bar \gamma }_{{i_q}}}m}} + \frac{1}{{{{\bar \gamma }_{{i_k}}}}}} \right)}}\left[ {1 - U\left( {\frac{1}{m} - \frac{{1 - T}}{{T \cdot h}}} \right)} \right].
\end{eqnarray} \normalsize
\item[iv)] The forth integral expression:
\\
In this case, similar to the case iii), we also need to consider two cases separately. As a result, we can re-write the forth integral expression as
\small \begin{eqnarray} \label{AP:capture_probability__CF_2}
\!\!\!\!\!\!\!\!\!\!\!\!\!\!\!\!\!\!\!\!\!\!\!\!\!\!\!\!\!\!\!\!\!\!\!\!&&\int_0^\infty \! {\int_0^{\left(\! {\frac{{1 - T}}{T}} \!\right){z_1}}\! {\exp \!\left( \!{ - \frac{{{z_1}}}{{{{\bar \gamma }_{{i_k}}}}}} \!\right)\exp\! \left(\! { - \left( \!{\sum\limits_{m = 1}^h \!{\left(\! {\frac{1}{{{{\bar \gamma }_{{i_{{j_m}}}}}}}} \!\right) \!+ \!\sum\limits_{l = 1}^m \!{\left( \!{\frac{1}{{{{\bar \gamma }_{{i_l}}}}}} \!\right) - \frac{m}{{{{\bar \gamma }_{{i_k}}}}}} } } \!\right)\frac{{{z_2}}}{h}} \!\right)U\left(\! {\frac{{{z_1}}}{m} - \frac{{{z_2}}}{h}} \!\right)\!d{z_2}d{z_1}} } \nonumber
\\
\!\!\!\!\!\!\!\!\!\!\!\!\!\!\!\!\!\!\!\!\!\!\!\!\!\!\!\!\!\!\!\!\!\!\!\!&&= \!\int_0^\infty \! {\exp \!\left(\! { - \frac{{{z_1}}}{{{{\bar \gamma }_{{i_k}}}}}} \!\right)\!\int_0^{\left(\! {\frac{{1 - T}}{T}} \!\right){z_1}}\! {\exp \!\left(\! { - \left(\! {\sum\limits_{m = 1}^h \!{\left( \!{\frac{1}{{{{\bar \gamma }_{{i_{{j_m}}}}}}}} \!\right) \!+ \!\sum\limits_{l = 1}^m \!{\left( \!{\frac{1}{{{{\bar \gamma }_{{i_l}}}}}} \!\right)\! - \!\frac{m}{{{{\bar \gamma }_{{i_k}}}}}} } } \!\right)\frac{{{z_2}}}{h}} \!\right)U\left(\! {\frac{1}{m} - \frac{{1 - T}}{{T \cdot h}}} \!\right)\!d{z_2}} d{z_1}} \nonumber
\\
\!\!\!\!\!\!\!\!\!\!\!\!\!\!\!\!\!\!\!\!\!\!\!\!\!\!\!\!\!\!\!\!\!\!\!\!&& \quad + \!\int_0^\infty \!\! {\exp \!\left(\! { - \frac{{{z_1}}}{{{{\bar \gamma }_{{i_k}}}}}} \!\right)\!\int_0^{\left(\! {\frac{h}{m}} \!\right){z_1}} \!\!\!\!{\!\!\exp \!\left(\! { - \left(\! {\sum\limits_{m = 1}^h \!{\left(\! {\frac{1}{{{{\bar \gamma }_{{i_{{j_m}}}}}}}} \!\right) \!+ \!\sum\limits_{l = 1}^m \!{\left(\! {\frac{1}{{{{\bar \gamma }_{{i_l}}}}}} \!\right) \!- \!\frac{m}{{{{\bar \gamma }_{{i_k}}}}}} } } \!\right)\frac{{{z_2}}}{h}} \!\right)\! \left[\! {1 - U\left(\! {\frac{1}{m} - \frac{{1 - T}}{{T \cdot h}}} \!\right)} \!\right]\!d{z_2}} d{z_1}}.
\end{eqnarray} \normalsize
With (\ref{AP:capture_probability__CF_2}), we can also directly derive the closed-form expressions as

\small\begin{eqnarray}
\!\!\!\!\!\!\!\!\!\!\!\!\!\!\!\!\!\!\!\!\!\!\!\!\!\!\!\!\!\!\!\!\!\!\!\!\!\!\!\!\!\!\!\!\!\!\!\!&&\int_0^\infty \! {\exp \!\left(\! { - \frac{{{z_1}}}{{{{\bar \gamma }_{{i_k}}}}}} \!\right)\!\int_0^{\left(\! {\frac{{1 - T}}{T}} \!\right){z_1}} \!{\exp \!\left(\! { - \left(\! {\sum\limits_{m = 1}^h \!{\left(\! {\frac{1}{{{{\bar \gamma }_{{i_{{j_m}}}}}}}} \!\right) \!+ \!\sum\limits_{l = 1}^m \! {\left(\! {\frac{1}{{{{\bar \gamma }_{{i_l}}}}}} \!\right)\! - \!\frac{m}{{{{\bar \gamma }_{{i_k}}}}}} } } \!\right)\frac{{{z_2}}}{h}} \!\right)U\left(\! {\frac{1}{m} - \frac{{1 - T}}{{T \cdot h}}} \!\right)\!d{z_2}} d{z_1}} \nonumber
\\
\!\!\!\!\!\!\!\!\!\!\!\!\!\!\!\!\!\!\!\!\!\!\!\!\!\!\!\!\!\!\!\!\!\!\!\!\!\!\!\!\!\!\!\!\!\!\!\!&& = \!\frac{{{{\bar \gamma }_{{i_k}}}h}}{{\left(\! {\sum\limits_{m = 1}^h \!{\left(\! {\frac{1}{{{{\bar \gamma }_{{i_{{j_m}}}}}}}} \!\right) \!+ \!\sum\limits_{l = 1}^m \!{\left(\! {\frac{1}{{{{\bar \gamma }_{{i_l}}}}}} \!\right) \!- \!\frac{m}{{{{\bar \gamma }_{{i_k}}}}}} } } \!\right)}}U\left(\! {\frac{1}{m} - \frac{{1 - T}}{{T \cdot h}}} \!\right)\nonumber
\\
\!\!\!\!\!\!\!\!\!\!\!\!\!\!\!\!\!\!\!\!\!\!\!\!\!\!\!\!\!\!\!\!\!\!\!\!\!\!\!\!\!\!\!\!\!\!\!\!&& \quad - \!\frac{h}{{\left(\! {\sum\limits_{m = 1}^h \!{\left(\! {\frac{1}{{{{\bar \gamma }_{{i_{{j_m}}}}}}}} \!\right) \!+ \!\sum\limits_{l = 1}^m \!{\left(\! {\frac{1}{{{{\bar \gamma }_{{i_l}}}}}} \!\right) \!- \!\frac{m}{{{{\bar \gamma }_{{i_k}}}}}} } } \!\right)\left\{\! {\left(\! {\sum\limits_{m = 1}^h \!{\left(\! {\frac{1}{{{{\bar \gamma }_{{i_{{j_m}}}}}}}} \!\right) \!+ \!\sum\limits_{l = 1}^m \!{\left(\! {\frac{1}{{{{\bar \gamma }_{{i_l}}}}}} \!\right) \!- \!\frac{m}{{{{\bar \gamma }_{{i_k}}}}}} } } \!\right)\frac{{1 - T}}{{T \cdot h}} \!+ \!\frac{1}{{{{\bar \gamma }_{{i_k}}}}}} \!\right\}}}U\left(\! {\frac{1}{m} - \frac{{1 - T}}{{T \cdot h}}} \!\right),
\end{eqnarray} \normalsize
and
\small\begin{eqnarray}
\!\!\!\!\!\!\!\!\!\!\!\!\!\!\!\!\!\!\!\!\!\!\!\!\!\!\!\!\!\!\!\!\!\!\!\!\!\!\!\!\!\!&&\int_0^\infty \! {\exp \!\left(\! { - \frac{{{z_1}}}{{{{\bar \gamma }_{{i_k}}}}}} \!\right)\!\int_0^{\left(\! {\frac{h}{m}} \!\right){z_1}} \!{\exp \!\left(\! { - \left( \!{\sum\limits_{m = 1}^h \!{\left(\! {\frac{1}{{{{\bar \gamma }_{{i_{{j_m}}}}}}}} \!\right) \!+ \!\sum\limits_{l = 1}^m \!{\left(\! {\frac{1}{{{{\bar \gamma }_{{i_l}}}}}} \!\right) \!- \!\frac{m}{{{{\bar \gamma }_{{i_k}}}}}} } } \!\right)\frac{{{z_2}}}{h}} \!\right)\left[\! {1 - U\left(\! {\frac{1}{m} - \frac{{1 - T}}{{T \cdot h}}} \!\right)} \!\right]\! d{z_2}} d{z_1}} \nonumber
\\
\!\!\!\!\!\!\!\!\!\!\!\!\!\!\!\!\!\!\!\!\!\!\!\!\!\!\!\!\!\!\!\!\!\!\!\!\!\!\!\!\!\!&&= \!\frac{{{{\bar \gamma }_{{i_k}}}h}}{{\left(\! {\sum\limits_{m = 1}^h \!{\left(\! {\frac{1}{{{{\bar \gamma }_{{i_{{j_m}}}}}}}} \!\right) \!+ \!\sum\limits_{l = 1}^m \!{\left(\! {\frac{1}{{{{\bar \gamma }_{{i_l}}}}}} \!\right) \!- \!\frac{m}{{{{\bar \gamma }_{{i_k}}}}}} } } \!\right)}}\left[\! {1 - U\left(\! {\frac{1}{m} - \frac{{1 - T}}{{T \cdot h}}} \!\right)} \!\right]\nonumber
\\
\!\!\!\!\!\!\!\!\!\!\!\!\!\!\!\!\!\!\!\!\!\!\!\!\!\!\!\!\!\!\!\!\!\!\!\!\!\!\!\!\!\!&& \quad - \!\frac{h}{{\left(\! {\sum\limits_{m = 1}^h \!{\left(\! {\frac{1}{{{{\bar \gamma }_{{i_{{j_m}}}}}}}} \!\right) \!+ \!\sum\limits_{l = 1}^m \!{\left(\! {\frac{1}{{{{\bar \gamma }_{{i_l}}}}}} \!\right) \!- \!\frac{m}{{{{\bar \gamma }_{{i_k}}}}}} } } \!\right)\left\{\! {\left(\! {\sum\limits_{m = 1}^h \!{\left(\! {\frac{1}{{{{\bar \gamma }_{{i_{{j_m}}}}}}}} \!\right) \!+\! \sum\limits_{l = 1}^m \!{\left(\! {\frac{1}{{{{\bar \gamma }_{{i_l}}}}}} \!\right) \!- \!\frac{m}{{{{\bar \gamma }_{{i_k}}}}}} } } \!\right)\frac{1}{m} \!+\! \frac{1}{{{{\bar \gamma }_{{i_k}}}}}} \!\right\}}}\left[\! {1 - U\left(\! {\frac{1}{m} - \frac{{1 - T}}{{T \cdot h}}} \!\right)} \!\right].
\end{eqnarray} \normalsize
\item[v)] The fifth integral expression:
\\
In this case, we need to consider two cases separately for valid integration. If $\frac{{1 - T}}{T} \ge \frac{h}{m}$, then $\frac{h}{m}{z_1}{\rm{ < }}{z_2} \le \frac{{1 - T}}{T}{z_1}$ and $\frac{{1 - T}}{{T \cdot h}} \ge \frac{1}{m}$. If $\frac{{1 - T}}{T} < \frac{h}{m}$, then there is no valid overlap integration region. As a result, we can re-write the third integral expression as
\small \begin{eqnarray} \label{AP:capture_probability__CF_3}
&&\int_0^\infty  {\int_0^{\left( {\frac{{1 - T}}{T}} \right){z_1}} {\exp \left( { - \frac{{{z_1}}}{{{{\bar \gamma }_{{i_k}}}}}} \right)\exp \left( { - \frac{{{z_2}}}{{{{\bar \gamma }_{{i_q}}}}}} \right)\left[ {1 - U\left( {\frac{{{z_1}}}{m} - \frac{{{z_2}}}{h}} \right)} \right]d{z_2}d{z_1}} } \nonumber
\\
&&= \int_0^\infty  {\exp \left( { - \frac{{{z_1}}}{{{{\bar \gamma }_{{i_k}}}}}} \right)\int_{\left( {\frac{h}{m}} \right){z_1}}^{\left( {\frac{{1 - T}}{T}} \right){z_1}} {\exp \left( { - \frac{{{z_2}}}{{{{\bar \gamma }_{{i_q}}}}}} \right)U\left( {\frac{{1 - T}}{{T \cdot h}} - \frac{1}{m}} \right)d{z_2}d{z_1}} }.
\end{eqnarray} \normalsize
With (\ref{AP:capture_probability__CF_3}), we can also directly derive the closed-form expressions as
\small\begin{eqnarray}
&&\int_0^\infty  {\exp \left( { - \frac{{{z_1}}}{{{{\bar \gamma }_{{i_k}}}}}} \right)\int_{\left( {\frac{h}{m}} \right){z_1}}^{\left( {\frac{{1 - T}}{T}} \right){z_1}} {\exp \left( { - \frac{{{z_2}}}{{{{\bar \gamma }_{{i_q}}}}}} \right)U\left( {\frac{{1 - T}}{{T \cdot h}} - \frac{1}{m}} \right)d{z_2}d{z_1}} } \nonumber
\\
&&= \frac{{{{\bar \gamma }_{{i_q}}}}}{{\left( {\frac{h}{{m \cdot {{\bar \gamma }_{{i_q}}}}} + \frac{1}{{{{\bar \gamma }_{{i_k}}}}}} \right)}}U\left( {\frac{{1 - T}}{{T \cdot h}} - \frac{1}{m}} \right) - \frac{{{{\bar \gamma }_{{i_q}}}}}{{\left( {\frac{{1 - T}}{{T \cdot {{\bar \gamma }_{{i_q}}}}} + \frac{1}{{{{\bar \gamma }_{{i_k}}}}}} \right)}}U\left( {\frac{{1 - T}}{{T \cdot h}} - \frac{1}{m}} \right).
\end{eqnarray} \normalsize
\item[vi)] The sixth integral expression:
\\
In this case, similar to the case v), we also need to consider two cases separately. As a result, we can re-write the forth integral expression as

\small \begin{eqnarray} \label{AP:capture_probability_CF_4}
\!\!\!\!\!\!\!\!\!\!\!\!\!\!\!\!\!\!\!\!\!\!\!\!\!\!\!\!\!\!\!\!\!\!\!\!\!\!\!\!\!\!&&\int_0^\infty \! {\int_0^{\left(\! {\frac{{1 - T}}{T}} \!\right){z_1}} \!{\exp \!\left(\! { - \frac{{{z_2}}}{{{{\bar \gamma }_{{i_q}}}}}} \!\right)\exp\! \left(\! { - \left(\! {\sum\limits_{m = 1}^h \!{\left(\! {\frac{1}{{{{\bar \gamma }_{{i_{{j_m}}}}}}}} \!\right) \!+ \!\sum\limits_{l = 1}^m \!{\left(\! {\frac{1}{{{{\bar \gamma }_{{i_l}}}}}} \!\right) \!-\! \frac{h}{{{{\bar \gamma }_{{i_q}}}}}} } } \!\right)\frac{{{z_1}}}{m}} \!\right)\left[\! {1 - U\left(\! {\frac{{{z_1}}}{m} - \frac{{{z_2}}}{h}} \!\right)} \!\right]\!d{z_2}d{z_1}} } \nonumber
\\
\!\!\!\!\!\!\!\!\!\!\!\!\!\!\!\!\!\!\!\!\!\!\!\!\!\!\!\!\!\!\!\!\!\!\!\!\!\!\!\!\!\!&&= \!\int_0^\infty \! {\exp\! \left(\! { - \left(\! {\sum\limits_{m = 1}^h \!{\left(\! {\frac{1}{{{{\bar \gamma }_{{i_{{j_m}}}}}}}} \!\right)\! + \!\sum\limits_{l = 1}^m \!{\left(\! {\frac{1}{{{{\bar \gamma }_{{i_l}}}}}} \!\right) \!- \!\frac{h}{{{{\bar \gamma }_{{i_q}}}}}} } }\! \right)\frac{{{z_1}}}{m}} \!\right)\!\int_{\left(\! {\frac{h}{m}} \!\right){z_1}}^{\left(\! {\frac{{1 - T}}{T}} \!\right){z_1}} \!{\exp\! \left(\! { - \frac{{{z_2}}}{{{{\bar \gamma }_{{i_q}}}}}} \!\right)U\left(\! {\frac{{1 - T}}{{T \cdot h}} - \frac{1}{m}} \!\right)\!d{z_2}} d{z_1}}.
\end{eqnarray} \normalsize
With (\ref{AP:capture_probability_CF_4}), we can also directly derive the closed-form expressions as
\small\begin{eqnarray}
\!\!\!\!\!\!\!\!\!\!&&\int_0^\infty  {\exp \left( { - \left( {\sum\limits_{m = 1}^h {\left( {\frac{1}{{{{\bar \gamma }_{{i_{{j_m}}}}}}}} \right) + \sum\limits_{l = 1}^m {\left( {\frac{1}{{{{\bar \gamma }_{{i_l}}}}}} \right) - \frac{h}{{{{\bar \gamma }_{{i_q}}}}}} } } \right)\frac{{{z_1}}}{m}} \right)\int_{\left( {\frac{h}{m}} \right){z_1}}^{\left( {\frac{{1 - T}}{T}} \right){z_1}} {\exp \left( { - \frac{{{z_2}}}{{{{\bar \gamma }_{{i_q}}}}}} \right)U\left( {\frac{{1 - T}}{{T \cdot h}} - \frac{1}{m}} \right)d{z_2}} d{z_1}} \nonumber
\\
\!\!\!\!\!\!\!\!\!\!&&= \frac{{m \cdot {{\bar \gamma }_{{i_q}}}}}{{\left( {\sum\limits_{m = 1}^h {\left( {\frac{1}{{{{\bar \gamma }_{{i_{{j_m}}}}}}}} \right) + \sum\limits_{l = 1}^m {\left( {\frac{1}{{{{\bar \gamma }_{{i_l}}}}}} \right)} } } \right)}}U\left( {\frac{{1 - T}}{{T \cdot h}} - \frac{1}{m}} \right)\nonumber
\\
\!\!\!\!\!\!\!\!\!\!&& \quad - \frac{{m \cdot {{\bar \gamma }_{{i_q}}}}}{{\left\{ {\left( {\sum\limits_{m = 1}^h {\left( {\frac{1}{{{{\bar \gamma }_{{i_{{j_m}}}}}}}} \right) + \sum\limits_{l = 1}^m {\left( {\frac{1}{{{{\bar \gamma }_{{i_l}}}}}} \right) - \frac{h}{{{{\bar \gamma }_{{i_q}}}}}} } } \right) + \frac{{m\left( {1 - T} \right)}}{{T \cdot {{\bar \gamma }_{{i_q}}}}}} \right\}}}U\left( {\frac{{1 - T}}{{T \cdot h}} - \frac{1}{m}} \right).
\end{eqnarray} \normalsize

\end{enumerate}

\bibliographystyle{ieeetran}
\bibliography{IEEEabrv,thesis}

\clearpage

\begin{figure}
\centering
\subfigure[Example of original $M$-dimensional groups\label{Example_a}]{\includegraphics[width=5.5in,trim=0.5cm 1.5cm 0.5cm 0cm]{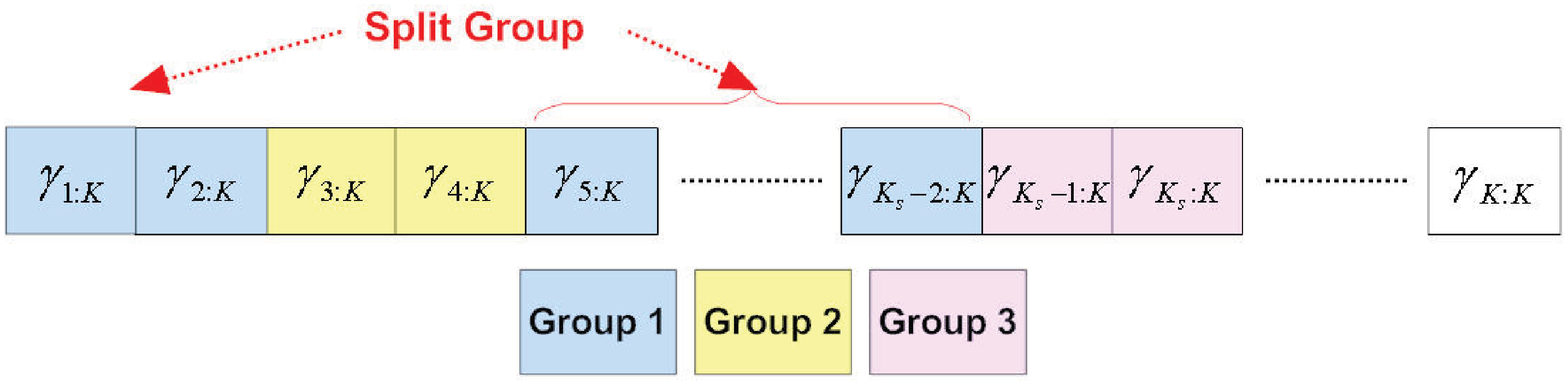}}\\
\subfigure[Example of substituted split groups\label{Example_b}]{\includegraphics[width=5.5in,trim=0.5cm 1.5cm 0.5cm 0cm]{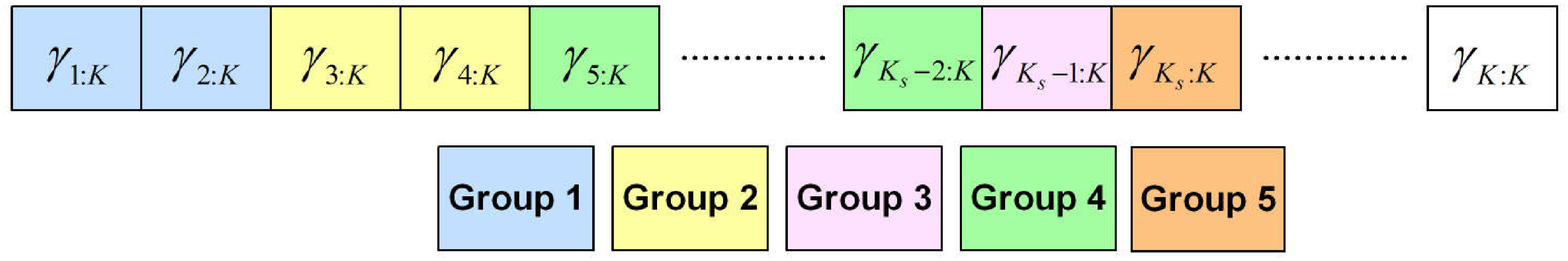}}
\caption{Examples for 3-dimensional joint PDF with split groups.}
\label{Example_2}
\end{figure}

\begin{figure}
\centering
\includegraphics[width=5.5in,trim=0.5cm 1.5cm 0.5cm 0cm]{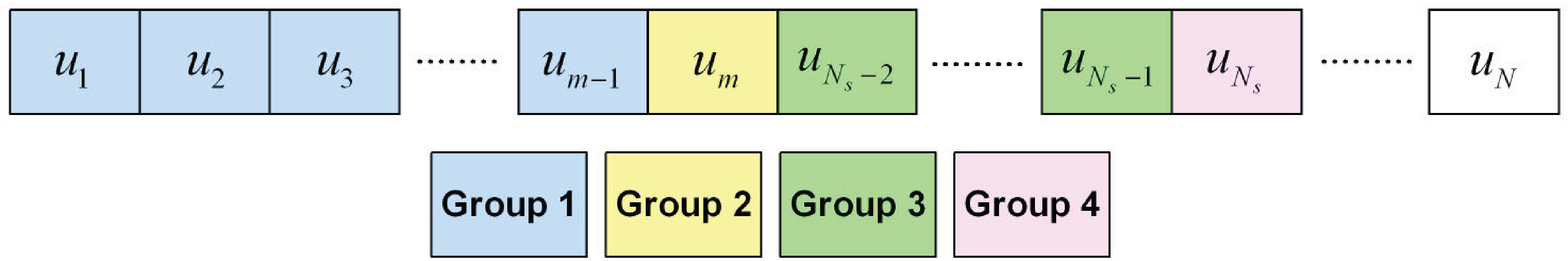}
\caption{Joint MGF of $u _{m}$ and $\sum\limits_{\scriptstyle n = 1 \hfill \atop \scriptstyle n \ne m \hfill}^{N_s } {u _{n} }$ for $1<m<N_s-1$.}
\label{Example_3_2}
\end{figure}

\end{document}